\documentclass[12pt]{article}
\usepackage[utf8]{inputenc}
\usepackage[sectionbib]{natbib}

\title{\vspace{-1.3cm}Scale-Invariant Robust Estimation of High-Dimensional Kronecker-Structured Matrices}
\author{Xiaoyu Zhang$^\dagger$, Zhiyun Fan$^\ddagger$, Wenyang Zhang$^\S$, and Di Wang$^\ddagger$\\ \small{$\dagger$ Tongji University, $\ddagger$ Shanghai Jiao Tong University, $\S$ University of Macau}}

\usepackage[margin=1in]{geometry}
\usepackage{tikz}
\usepackage{pgfplots}
\usepackage{graphicx}
\usepackage{mathrsfs}
\usepackage{multirow}
\usepackage{amsfonts}
\usepackage{amsmath}
\usepackage{comment}
\usepackage{amsthm}
\usepackage{color}
\usepackage{mathtools}
\usepackage{algorithm}
\usepackage{sectsty}
\usepackage[colorlinks]{hyperref}
\usepackage{amsmath}
\usepackage{amssymb}
\usepackage[toc,page]{appendix}
\usepackage[mathscr]{euscript} 
\usepackage[toc]{appendix}

\let\counterwithin\relax
\usepackage{chngcntr}
\usepackage{apptools}
\usepackage{booktabs}
\usepackage{lscape}
\usepackage{soul}
\usepackage{setspace}
\usepackage{epstopdf}
\usepackage{xr}
\usepackage{booktabs}
\usepackage{array}
\usepackage{makecell}
\usepackage{subcaption}
\usepackage{arydshln}

\makeatletter
\newcommand*{\addFileDependency}[1]{
  \typeout{(#1)}
  \@addtofilelist{#1}
  \IfFileExists{#1}{}{\typeout{No file #1.}}
}
\makeatother

\newcommand*{\myexternaldocument}[1]{%
    \externaldocument{#1}%
    \addFileDependency{#1.tex}%
    \addFileDependency{#1.aux}%
}

\myexternaldocument{RobustKP_supp}

\AtAppendix{\counterwithin{lemma}{section}}

\newtheorem{assumption}{Assumption}
\newtheorem{definition}{Definition}
\newtheorem{lemma}{Lemma}
\newtheorem{proposition}{Proposition}
\newtheorem{theorem}{Theorem}
\newtheorem{corollary}{Corollary}
\newtheorem{remark}{Remark}

\theoremstyle{definition}

\hypersetup{
    colorlinks=true,
    linkcolor=blue,
    citecolor=blue,
    filecolor=magenta,
    urlcolor=magenta,
}

\DeclareMathOperator*{\argmin}{arg\,min}

\newcommand{\bm}{\mathbf}
\newcommand{\bbm}{\boldsymbol}

\newcommand{\bb}{\mathbb}

\newcommand{\norm}[1]{\left\| #1 \right\|}

\newcommand{\inner}[2]{\left\langle #1, #2\right\rangle}
\newcommand{\wt}[1]{\widetilde{#1}}
\newcommand{\wh}[1]{\widehat{#1}}

\mathtoolsset{showonlyrefs}

\begin{document}

\setlength{\parindent}{16pt}

\maketitle

\vspace{-1cm}
\begin{abstract}

High-dimensional Kronecker-structured estimation faces a conflict between non-convex scaling ambiguities and statistical robustness. The arbitrary factor scaling distorts gradient magnitudes, rendering standard fixed-threshold robust methods ineffective. We resolve this via Scaled Robust Gradient Descent (SRGD), which stabilizes optimization by de-scaling gradients before truncation. To further enforce interpretability, we introduce Scaled Hard Thresholding (SHT) for invariant variable selection. A two-step estimation procedure, built upon robust initialization and SRGD--SHT iterative updates, is proposed for canonical matrix problems, such as trace regression, matrix GLMs, and bilinear models. The convergence rates are established for heavy-tailed predictors and noise, identifying a phase transition where optimal convergence rates recover under finite noise variance and degrade optimally for heavier tails. Experiments on simulated data and two real-world applications confirm superior robustness and efficiency of the proposed procedure.

\end{abstract}


\newpage

\setlength\abovedisplayskip{4pt}
\setlength\belowdisplayskip{4pt}

\section{Introduction}

\subsection{Kronecker Structure and the Robustness-Scaling Tension}

High-dimensional structured matrix estimation is a cornerstone of modern statistical methodology, offering parsimonious representations that effectively capture multi-modal interactions while substantially reducing dimensionality. Prominent examples include matrix regression \citep{negahban2011estimation}, matrix completion \citep{negahban2012restricted}, and spatio-temporal modeling in neuroimaging \citep{hung2013matrix} and econometrics \citep{chen2021autoregressive}. The Kronecker product structure, wherein a target matrix is decomposed into the product of smaller factors, is particularly effective in settings where interactions separate along distinct but interdependent modes.

Despite these modeling advantages, reliable estimation of Kronecker-structured matrices faces a unique intersection of challenges: non-convex optimization and heavy-tailed data. While these issues are individually well-studied, their interaction creates a specific failure mode for existing algorithms. Specifically, Kronecker factorizations are non-identifiable up to a rescaling of their factors. This invariance induces a scaling ambiguity where the norms of the latent factors can become arbitrarily unbalanced. In standard optimization, this leads to ill-conditioned Hessians. However, in robust statistical estimation, the consequences are far more severe.

Standard robust methods rely on fixed-threshold mechanisms, such as gradient truncation and Huber loss, to suppress outliers. In the Kronecker setting, however, the efficacy of these thresholds becomes arbitrarily coupled with the scaling ambiguity of the factorization. When a factor scale is large, the gradient is artificially amplified, causing valid signals to exceed the threshold and be indiscriminately clipped, thereby discarding valuable information. Conversely, when a factor scale is small, the gradient diminishes, allowing gross outliers to fall below the threshold and leak into the estimate, effectively bypassing the robustification mechanism. Consequently, standard robust algorithms fail to be scale-invariant, as their statistical performance fluctuates wildly based on the arbitrary parameterization of the factorization.

Furthermore, enforcing sparsity, often desirable for interpretability, introduces a structural contradiction. Standard row-wise hard thresholding is sensitive to basis transformations; rescaling the factors changes the relative row norms, leading to inconsistent support recovery. Achieving interpretable, robust estimation requires a methodology that decouples statistical robustification from the scale indeterminacy of the model.

\subsection{Contributions}

In this article, we propose a unified framework for the robust estimation of Kronecker-structured matrices that resolves the tension between scaling ambiguity, heavy-tailed noise, and sparsity. Our contributions are three-fold:

\begin{itemize}
  \item[1.] \textit{Scale-Invariant Robust Optimization.} We introduce Scaled Robust Gradient Descent (SRGD), an algorithm that stabilizes updates via a three-step mechanism: (i) \textit{de-scaling} to isolate the statistical error signal from the factor norms; (ii) \textit{truncation} in the normalized domain to suppress heavy-tailed outliers; and (iii) \textit{re-scaling} to respect the optimization geometry. This ensures that the robustification mechanism is invariant to the factorization scaling, preventing the failure modes of naive truncation.
  
  \item[2.] \textit{Scale-Invariant Sparsity Recovery.} To enforce structured sparsity, we develop Scaled Hard Thresholding (SHT). Unlike standard thresholding, SHT operates on \textit{scaled factors} whose row norms are invariant to the factorization scaling. This ensures that the selected support set is intrinsic to the Kronecker product rather than an artifact of the initialization.
  
  \item[3.] \textit{Provable Convergence and Sharp Phase Transitions.} We establish the convergence of a two-stage estimation procedure, based on robust initialization and SRGD--SHT iterates, for three canonical models: matrix trace regression, matrix GLMs, and bilinear regression. Our theory accommodates heavy tails in both predictors (finite $(2+2\lambda)$-th moment) and noise (finite $(1+\epsilon)$-th moment). We derive non-asymptotic error bounds that reveal a smooth phase transition: for finite-variance noise ($\epsilon=1$), we recover the optimal parametric rate $O(\sqrt{s/n})$, showing that robustness incurs no efficiency loss in ideal settings; for heavy-tailed noise ($\epsilon < 1$), the rate slows to $O(s^{1/2}n^{-\epsilon/(1+\epsilon)})$, matching information-theoretic lower bounds. Crucially, the error bounds are independent of the factorization condition number, confirming the stability of the approach.
\end{itemize}

\subsection{Related Literature}

Our work synthesizes and extends research in structured matrix estimation, robust optimization, and sparsity-constrained nonconvex optimization.

\textit{Kronecker and Structured Matrix Estimation.} Convex approaches (nuclear-norm and structured-norm regularization) provide a baseline for low-rank estimation \citep{candes2011tight,negahban2012restricted}, but scale poorly in high-dimensional Kronecker models. Factorized nonconvex methods offer computational advantages and have been extensively analyzed for low-rank problems under light-tailed errors \citep{tu2016low,chen2017solving,wang2017unified,ma2018implicit}. Recent contributions addressing scaling invariance \citep{tong2021accelerating,tong2021low} provide important algorithmic foundations (e.g., gradient normalization) that we adopt and adapt for robust estimation. Prior work on Kronecker-structured estimation \citep{cai2022kopa,cai2023hybrid,wu2023sparse} typically does not treat heavy tails and non-identifiability jointly.

\textit{Robust Estimation.} Robustification methods operate at several levels: loss modification \citep{sun2020adaptive,tan2023sparse,shen2025computationally}, data transformation via truncation or shrinkage \citep{fan2021shrinkage,wang2023rate,lu2025robust}, and gradient-level techniques \citep{prasad2020robust}. Recently, median-based or truncated gradient descent methods have yielded strong guarantees for heavy-tailed vector and tensor problems \citep{chi2019median,li2020non,zhang2024robust}. Our SRGD aligns with gradient robustification approaches but differs by explicitly correcting the scaling ambiguities of Kronecker factorizations and accommodating heavy tails in both predictors and responses. 

\textit{Sparsity and Hard Thresholding.} Iterative hard thresholding is a standard tool for sparse recovery \citep{amini2008high,cai2013sparse,yang2014sparse,chen2022fast,cheng2024two}. However, existing extensions to matrix and tensor settings typically assume a fixed parameterization. Our SHT adapts the hard-thresholding principle to the Kronecker factorization by operating in a normalized domain that respects the equivalence class of factorizations.

\begin{table}[t]
\caption{Comparison of our proposed SRGD--SHT framework with existing robust and structured estimation methods: model type, infinite noise variance, infinite predictor kurtosis, scale invariance, and parsimony type.}
\label{tab:comparison}
\centering
\begin{tabular}{>{\raggedright\arraybackslash}m{3.7cm} cccc>{\raggedright\arraybackslash}m{3.1cm}}
\toprule
\textbf{Methodology} & \textbf{Model} & \textbf{Inf. Noise} & \textbf{Inf. Pred.} & \textbf{Scale} & \textbf{Parsimony} \\
& \textbf{Type} & \textbf{Variance} & \textbf{Kurtosis} & \textbf{Inv.} & \textbf{Type} \\
\midrule
Robust Gradient \newline \citep{prasad2020robust} & Vector & $\times$ & $\times$ & N/A & N/A \\
\\[0.01cm]
Data Shrinkage \newline \citep{fan2021shrinkage} & Matrix & $\times$ & $\times$ & N/A & Low-Rank \\
\\[0.01cm] 
Matrix Huber Reg. \newline \citep{tan2023sparse} & Matrix & \checkmark & $\times$ & N/A & Low-Rank \\
\\[0.01cm] 
Scaled GD \newline \citep{tong2021accelerating} & Matrix & $\times$ & $\times$ & \checkmark & Low-Rank \\
\\[0.01cm]
Robust Tensor \newline \citep{zhang2024robust} & Tensor & \checkmark & \checkmark & $\times$ & Low-Rank \\
\\[0.01cm]
\textbf{SRGD--SHT \newline (This Work)} & \textbf{Matrix} & \textbf{\checkmark} & \textbf{\checkmark} & \textbf{\checkmark} & \textbf{Kronecker \newline and Sparsity} \\
\bottomrule
\end{tabular}
\end{table}

In summary, the key advantages of our work are summarized in Table \ref{tab:comparison}. Compared to the existing works, our approach is the only one that simultaneously addresses the nonconvex Kronecker structure, heavy-tailed distributions in both predictors (with infinite kurtosis) and noise (with infinite variance), and the intrinsic scaling ambiguity of the factorization.

\subsection{Notations and Outline}

For vectors, $\|\cdot\|_0$, $\|\cdot\|_1$, $\|\cdot\|_2$, and $\|\cdot\|_\infty$ denote the number of nonzero entries, $\ell_1$, $\ell_2$, and $\ell_\infty$ norms, respectively. For a generic matrix $\bm{X}$, let $\bm{X}^\top$, $\text{vec}(\bm{X})$, and $\sigma_j(\bm{X})$ denote its transpose, vectorization, and $j$-th largest singular value. Matrix norms are denoted as follows: $\|\bm X\|_0$, $\|\bm X\|_1$, and $\|\bm X\|_\infty$ refer to the vector norms of $\text{vec}(\bm X)$; $\|\bm X\|_\text{F}$ is the Frobenius norm; $\|\bm X\|_\text{op}$ is the spectral norm; and $\|\bm X\|_{1,\infty}$ and $\|\bm X\|_{2,\infty}$ denote the maximum $\ell_1$ and $\ell_2$ norms of the rows, respectively. For symmetric matrices, $\lambda_{\min}(\bm{X})$ and $\lambda_{\max}(\bm{X})$ denote the minimum and maximum eigenvalues. The general linear group of degree $n$ is denoted by $\text{GL}(n)$. A generic positive constant is denoted by $C$. For sequences $x_k$ and $y_k$, we write $x_k\lesssim y_k$ if $x_k\leq Cy_k$ for some constant $C>0$, $x_k\gtrsim y_k$ if $y_k\lesssim x_k$, and $x_k\asymp y_k$ if both hold.

The remainder of the article is organized as follows. Section \ref{sec:2} introduces the Kronecker-structured estimation problem, discusses challenges of scaling invariance, and presents the SRGD algorithm with theoretical guarantees. Section \ref{sec:3} extends the framework to structured sparsity via SHT. Section \ref{sec:4} instantiates the framework for matrix trace regression, matrix generalized linear models, and bilinear models. Sections \ref{sec:5} and \ref{sec:6} present simulation studies and real data applications, respectively. We conclude with a discussion in Section \ref{sec:7}.

\section{Scaled Robust Gradient Descent}\label{sec:2}

\subsection{Problem Formulation and Scaling Ambiguity}\label{sec:2.1}

We consider the estimation of a structured matrix $\bm{\Theta}\in\mathbb{R}^{p\times q}$ by minimizing the empirical risk $\mathcal{L}_n(\bm{\Theta};\mathcal{D}_n)=n^{-1}\sum_{i=1}^n\mathcal{L}(\bm{\Theta};z_i)$. We assume $\bm{\Theta}$ admits a rank-$K$ Kronecker product decomposition
\begin{equation}\label{eq:KP}
  \bm{\Theta}=\sum_{k=1}^K\bm{A}_{k}\otimes\bm{B}_k,
\end{equation}
where $\bm{A}_k\in\mathbb{R}^{p_1\times q_1}$ and $\bm{B}_k\in\mathbb{R}^{p_2\times q_2}$.
To facilitate optimization, we utilize the permutation operator $\mathcal{P}: \mathbb{R}^{p \times q} \to \mathbb{R}^{d_1 \times d_2}$ with $d_1=p_1q_1$ and $d_2=p_2q_2$, which maps the Kronecker product to a standard low-rank factorization \citep{cai2022kopa}:
\begin{equation}
  \mathcal{P}(\bm{\Theta}) = \sum_{k=1}^K\text{vec}(\bm{A}_k)\text{vec}(\bm{B}_k)^\top = \bm{L}\bm{R}^\top,
\end{equation}
where $\bm{L}\in\mathbb{R}^{d_1\times K}$ and $\bm{R}\in\mathbb{R}^{d_2\times K}$. The estimation task is thus transformed into optimizing the latent factors $\bm{L}$ and $\bm{R}$ via the composite objective $f(\bm{L},\bm{R};z_i) = \mathcal{L}(\mathcal{P}^{-1}(\bm{L}\bm{R}^\top);z_i)$, where $\mathcal{P}^{-1}(\cdot)$ is the inverse operator of $\mathcal{P}$.

A fundamental challenge in optimizing $f(\bm{L},\bm{R})$ is the intrinsic non-identifiability of the factorization. For any nonsingular matrix $\bm{Q}\in\text{GL}(K)$, the transformation $(\bm{L}\bm{Q}, \bm{R}\bm{Q}^{-\top})$ preserves the product $\bm{L}\bm{R}^\top$. This invariance implies that the norms of the factors are arbitrary, yet the partial gradients depend explicitly on these scales. Let $\bm{G}_i = \mathcal{P}(\nabla\mathcal{L}(\mathcal{P}^{-1}(\bm{L}\bm{R}^\top);z_i)) \in \mathbb{R}^{d_1 \times d_2}$ denote the gradient of the loss with respect to the permuted parameter, capturing the prediction error signal. By the chain rule, the partial gradients are:
\begin{equation}
  \nabla_{\bm{L}} f(\bm{L},\bm{R};z_i) = \bm{G}_i \bm{R} \quad \text{and} \quad \nabla_{\bm{R}} f(\bm{L},\bm{R};z_i) = \bm{G}_i^\top \bm{L}.
\end{equation}
Consequently, if $\bm{R}$ has a large norm (and $\bm{L}$ a correspondingly small norm), the gradient with respect to $\bm{L}$ is artificially amplified by the factor $\bm{R}$, while the gradient with respect to $\bm{R}$ is diminished by $\bm{L}$. This scaling ambiguity creates a critical vulnerability when introducing robust gradient techniques.

\subsection{Motivation: The Failure of Fixed-Threshold Robustness}\label{sec:2.2}

Standard methods for robust estimation typically rely on modifying the gradient via a threshold-based influence function. Before detailing our algorithm, we illustrate why a naive application of these techniques fails in the Kronecker setting. The core issue is that the activation of the robust mechanism becomes arbitrarily coupled with the scaling of the factors.

Consider a simplified rank-$1$ scalar setting where the target is $\theta_\ast = l_\ast r_\ast$. Suppose we observe $y = \theta_\ast + \xi$, where $\xi$ is a heavy-tailed noise term. Evaluated at the true parameter, the gradient of the quadratic loss with respect to $l$ is $g_l = (l_\ast r_\ast - y)r_\ast = -\xi r_\ast$.
A standard robust update replaces the raw gradient $g_l$ with a robustified version $\psi_\tau(g_l)$. The most common approach, utilized in both truncated gradient descent and Huber regression, is to clip the gradient magnitude at a fixed threshold $\tau$:
\begin{equation}
    \psi_\tau(u) = \text{sgn}(u)\min(|u|, \tau).
\end{equation}

The efficacy of this robust estimator is critically compromised by the arbitrary scaling factor $\alpha$ in the reparameterization $(l,r) \to (\alpha l, \alpha^{-1} r)$. When the factorization is scaled such that $|r|$ is small, the gradient magnitude $|g_l| = |\xi r|$ diminishes; consequently, even gross outliers may satisfy $|\xi r| < \tau$, rendering the truncation inactive and reverting the method to non-robust least squares. Conversely, when $|r|$ is large, the gradient is artificially amplified, causing valid signals to exceed $\tau$ and be indiscriminately clipped. This leads to a vanishing effective step size and arbitrarily slow convergence, demonstrating that a fixed threshold cannot simultaneously accommodate the conflicting scales induced by the factorization.

The fundamental flaw is that fixed-threshold robustifiers conflate the statistical noise $\xi$ with the arbitrary scaling factor $r$. To resolve this, we must \textit{de-scale} the gradient to isolate the statistical component $\xi$, apply the robust function to the noise, and then \textit{re-scale} to match the current parameterization. This ensures that robustness is scale invariant, motivating the three-stage mechanism of our proposed Scaled Robust Gradient Descent (SRGD).

\subsection{Algorithm: De-scaling, Truncation, and Re-scaling}\label{sec:2.3}

We formally implement the scale-invariant strategy via a three-stage gradient modification process. The geometric intuition behind this procedure is visualized in Figure \ref{fig:gradient_geometry}.

\textbf{Step 1: De-scaling.} We first rescale the gradients to remove the influence of the current factor magnitudes. The \textit{de-scaled gradients} are obtained by post-multiplying the raw gradients with the inverse square root of the Gram matrices:
\begin{equation}
  \begin{split}
    \nabla_{\bm{L}} f(\bm{L},\bm{R};z_i)(\bm{R}^\top\bm{R})^{-1/2} & = \mathcal{P}(\nabla\mathcal{L}(\mathcal{P}^{-1}(\bm{L}\bm{R}^\top);z_i))(\bm{R}(\bm{R}^\top\bm{R})^{-1/2}), \\
    \nabla_{\bm{R}} f(\bm{L},\bm{R};z_i)(\bm{L}^\top\bm{L})^{-1/2} & = \mathcal{P}(\nabla\mathcal{L}(\mathcal{P}^{-1}(\bm{L}\bm{R}^\top);z_i))^\top(\bm{L}(\bm{L}^\top\bm{L})^{-1/2}).
  \end{split}
\end{equation}
This operation captures the descent \textit{direction} invariant to the relative scaling of $\bm{L}$ and $\bm{R}$. As illustrated in Figure \ref{fig:gradient_geometry}(b), de-scaling effectively preconditions the gradient space, restoring an isotropic geometry where the statistical magnitude of the noise is decoupled from the parameter scale.

\textbf{Step 2: Truncation.} To enforce robustness, we apply element-wise truncation to the de-scaled gradients. For a matrix $\bm{M}$ and threshold $\tau>0$, let $\text{T}(\bm{M},\tau)_{j,k}=\text{sgn}(\bm{M}_{j,k})\cdot\min(|\bm{M}_{j,k}|,\tau)$. The averaged truncated gradients are:
\begin{equation}\label{eq:truncated}
  \begin{split}
    \bm{G}_L(\bm{L},\bm{R};\tau) & = \frac{1}{n}\sum_{i=1}^n\text{T}\Big(\nabla_{\bm{L}}f(\bm{L},\bm{R};z_i)(\bm{R}^\top\bm{R})^{-1/2},\tau\Big),\\
    \bm{G}_R(\bm{L},\bm{R};\tau) & = \frac{1}{n}\sum_{i=1}^n\text{T}\Big(\nabla_{\bm{R}}f(\bm{L},\bm{R};z_i)(\bm{L}^\top\bm{L})^{-1/2},\tau\Big).
  \end{split}
\end{equation}
Crucially, because truncation is applied in the de-scaled domain, the threshold $\tau$ acts uniformly on the gradient directions. Figure \ref{fig:gradient_geometry}(a) demonstrates the consequence of skipping the de-scaling step: if the factors are imbalanced (e.g., $\|\bm{R}\|_\text{F} \gg \|\bm{L}\|_\text{F}$), the gradients are distorted. Valid signals along the $\bm{L}$-direction are artificially amplified and clipped by the fixed $\tau$-box, while outliers along the $\bm{R}$-direction are shrunk and thus omitted from the robustness control. Our de-scaled approach (Figure \ref{fig:gradient_geometry}(b)) avoids these failure modes.

\textbf{Step 3: Re-scaling.} Finally, we map the robustified update directions back to the parameter space by post-multiplying with the scaling factors $(\bm{R}^\top\bm{R})^{-1/2}$ and $(\bm{L}^\top\bm{L})^{-1/2}$, respectively. This restores the correct scale for the parameter update. 

The resulting Scaled Robust Gradient Descent (SRGD) procedure is summarized in Algorithm \ref{alg:1}. The algorithm requires the Gram matrices $\bm{L}_j^\top\bm{L}_j$ and $\bm{R}_j^\top\bm{R}_j$ to remain nonsingular. Theorem \ref{thm:1} guarantees that, given suitable initialization, this condition holds throughout the optimization trajectory.


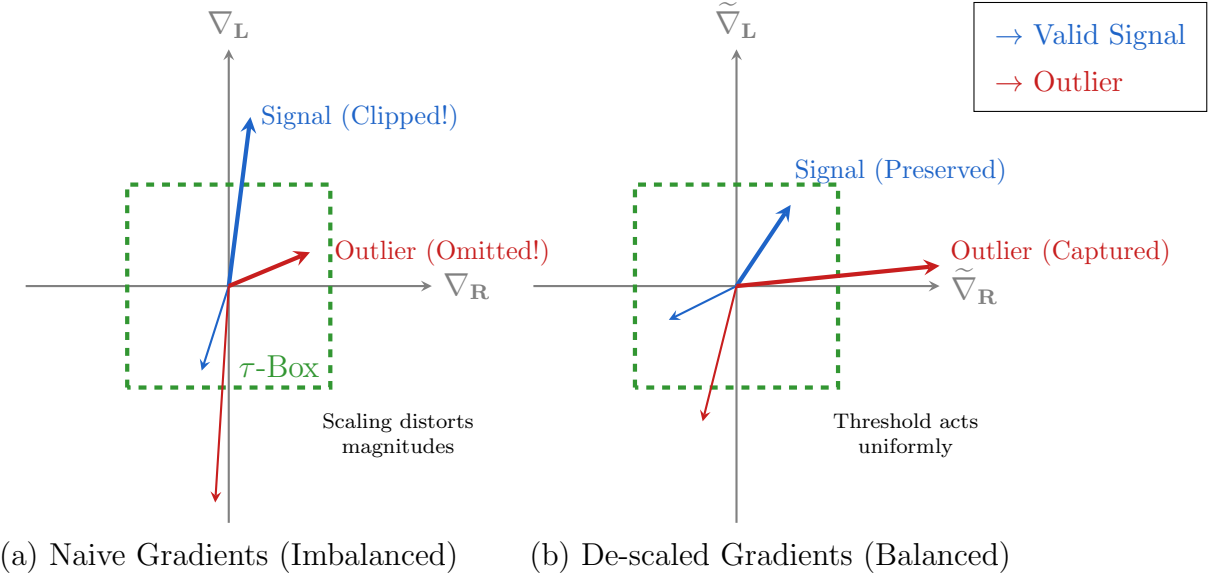
\begin{figure}[t]
\centering
\begin{tikzpicture}[>=stealth, scale=0.9]

\definecolor{signalcolor}{RGB}{30, 100, 200} 
\definecolor{outliercolor}{RGB}{200, 30, 30} 
\definecolor{boxcolor}{RGB}{50, 150, 50}     

\begin{scope}[local bounding box=panelA]
    \draw[->, thick, gray] (-3,0) -- (3,0) node[right] {$\nabla_{\bm{R}}$};
    \draw[->, thick, gray] (0,-3.5) -- (0,3.5) node[above] {$\nabla_{\bm{L}}$};
    \node at (0,-4) {(a) Naive Gradients (Imbalanced)};

    \draw[dashed, ultra thick, boxcolor] (-1.5,-1.5) rectangle (1.5,1.5);
    \node[boxcolor, anchor=south east] at (1.5,-1.5) {$\tau$-Box};

    
    \draw[->, ultra thick, signalcolor] (0,0) -- (0.32, 2.5);
    \node[signalcolor, anchor=west] at (0.32, 2.5) {\footnotesize Signal (Clipped!)};

    \draw[->, thick, signalcolor] (0,0) -- (-0.4, -1.25);

    \draw[->, ultra thick, outliercolor] (0,0) -- (1.2, 0.5);
    \node[outliercolor, anchor=west] at (1.4, 0.5) {\footnotesize Outlier (Omitted!)};

    \draw[->, thick, outliercolor] (0,0) -- (-0.2, -3.2); 
    
    \node[align=center, font=\scriptsize] at (2.5, -2.2) {Scaling distorts\\magnitudes};
\end{scope}

\begin{scope}[xshift=7.5cm, local bounding box=panelB]
    \draw[->, thick, gray] (-3,0) -- (3,0) node[right] {$\widetilde{\nabla}_{\bm{R}}$};
    \draw[->, thick, gray] (0,-3.5) -- (0,3.5) node[above] {$\widetilde{\nabla}_{\bm{L}}$};
    \node at (0.5,-4) {(b) De-scaled Gradients (Balanced)};

    \draw[dashed, ultra thick, boxcolor] (-1.5,-1.5) rectangle (1.5,1.5);

    
    \draw[->, ultra thick, signalcolor] (0,0) -- (0.8, 1.2);
    \node[signalcolor, anchor=west] at (0.7, 1.7) {\footnotesize Signal (Preserved)};

    \draw[->, thick, signalcolor] (0,0) -- (-1.0, -0.5);

    \draw[->, ultra thick, outliercolor] (0,0) -- (3.0, 0.3);
    \node[outliercolor, anchor=west] at (3.0, 0.5) {\footnotesize Outlier (Captured)};

    \draw[->, thick, outliercolor] (0,0) -- (-0.5, -2.0);
    
    \node[align=center, font=\scriptsize] at (2.5, -2.2) {Threshold acts\\uniformly};
\end{scope}

\matrix [draw, fill=white, below right] at (11, 4.2) {
  \node [signalcolor] {\small{$\to$ Valid Signal}}; \\
  \node [outliercolor] {\small{$\to$ Outlier}}; \\
};

\end{tikzpicture}
\caption{Visualization of robustification. (a) Imbalanced factors ($\|\bm{R}\|_\text{F} \gg \|\bm{L}\|_\text{F}$) distort gradients: valid signals (blue) are stretched and clipped by the fixed threshold $\tau$-box, while outliers (red) shrink and omitted inside. (b) SRGD de-scales gradients, restoring a balanced geometry where the threshold correctly separates signal from noise.}
\label{fig:gradient_geometry}
\end{figure}

\begin{algorithm}
  \caption{Scaled Robust Gradient Descent (SRGD) algorithm}
  \label{alg:1}
  \setstretch{1.4}
  \textbf{Input}: initial value $\bm{F}_0$, step size $\eta>0$, truncation parameter $\tau>0$, number of iterations $J$\\
  \textbf{For }$j=0,\dots,J-1$\\
  \hspace*{1cm}$\bm{L}_{j+1} = \bm{L}_j - \eta\cdot \bm{G}_{L}(\bm{L}_j,\bm{R}_j;\tau)(\bm{R}_j^\top\bm{R}_j)^{-1/2}$\\
  \hspace*{1cm}$\bm{R}_{j+1} = \bm{R}_j - \eta\cdot \bm{G}_{R}(\bm{L}_j,\bm{R}_j;\tau)(\bm{L}_j^\top\bm{L}_j)^{-1/2}$\\
  \textbf{End for}\\
  \textbf{Return} $\widehat{\bm{\Theta}}=\mathcal{P}^{-1}(\bm{L}_J\bm{R}_J^\top)$
\end{algorithm}\vspace{-0.3cm}

\subsection{Computational Convergence Analysis}\label{sec:2.4}

We establish the local convergence of SRGD under mild regularity conditions. Let $\bm{\Theta}_\ast$ denote the ground truth, with SVD $\mathcal{P}(\bm{\Theta}_\ast) = \bm{U}_\ast\bm{\Sigma}_\ast\bm{V}_\ast^\top$. To resolve the factorization ambiguity in our analysis, we fix the target factors as $\bm{L}_\ast=\bm{U}_\ast\bm{\Sigma}_\ast^{1/2}$ and $\bm{R}_\ast=\bm{V}_\ast\bm{\Sigma}_\ast^{1/2}$. We adopt the factorization distance metric from \citet{tong2021accelerating}:

\begin{definition}[Factorization Distance]\label{def:3}
  Let $\bm{F}=(\bm{L},\bm{R})$ and $\bm{F}_\ast=(\bm{L}_\ast,\bm{R}_\ast)$. The distance between $\bm{F}$ and $\bm{F}_\ast$ is defined as
  \begin{equation}
    d(\bm{F},\bm{F}_\ast) = \inf_{\bm{Q}\in\textup{GL}(K)}\sqrt{\left\|(\bm{L}\bm{Q}-\bm{L}_\ast)\bm{\Sigma}_\ast^{1/2}\right\|_\textup{F}^2 + \left\|(\bm{R}\bm{Q}^{-\top}-\bm{R}_\ast)\bm{\Sigma}_\ast^{1/2}\right\|_\textup{F}^2}~.
  \end{equation}
\end{definition}

We assume the loss function satisfies the Restricted Correlated Gradient (RCG) condition. This condition is imposed on the \textit{expected} gradient to accommodate heavy-tailed distributions where the empirical loss may not concentrate uniformly.

\begin{definition}[Restricted Correlated Gradient]\label{def:1}
  The loss function $\mathcal{L}$ is said to satisfy the Restricted Correlated Gradient (RCG) condition if, for any $\bm{\Theta}$ satisfying \eqref{eq:KP}, we have
  \begin{equation}
    \langle\mathbb{E}[\nabla\mathcal{L}(\bm{\Theta};z_i)],\bm{\Theta}-\bm{\Theta}_\ast\rangle \geq \frac{\alpha}{2}\|\bm{\Theta}-\bm{\Theta}_\ast\|_\textup{F}^2 + \frac{1}{2\beta}\|\mathbb{E}[\nabla\mathcal{L}(\bm{\Theta};z_i)]\|_\textup{F}^2,
  \end{equation}
  where $\alpha$ and $\beta$ are RCG parameters satisfying $0<\alpha\leq\beta$.
\end{definition}

To ensure statistical stability, we require the de-scaled robust gradients to be accurate estimators of the population gradients.

\begin{definition}[De-scaled Robust Gradient Stability]\label{def:2}
  Given $(\bm{L},\bm{R})$, the de-scaled robust gradient functions $\bm{G}_{L}(\bm{L},\bm{R})$ and $\bm{G}_{R}(\bm{L},\bm{R})$ are said to be stable if there exist positive constants $\phi$, $\xi_L$, and $\xi_R$ such that for all $(\bm{L},\bm{R})$,
  \begin{equation}
    \begin{split}
      \left\|\bm{G}_{L}(\bm{L},\bm{R}) - \mathbb{E}\left[\nabla_{\bm{L}}f(\bm{L},\bm{R};z_i)\right](\bm{R}^\top\bm{R})^{-1/2}\right\|_\textup{F}^2 & \leq \phi\|\mathcal{P}^{-1}(\bm{L}\bm{R}^\top)-\bm{\Theta}_\ast\|_\textup{F}^2 + \xi_L^2,\\
      \left\|\bm{G}_{R}(\bm{L},\bm{R}) - \mathbb{E}\left[\nabla_{\bm{R}}f(\bm{L},\bm{R};z_i)\right](\bm{L}^\top\bm{L})^{-1/2}\right\|_\textup{F}^2 & \leq \phi\|\mathcal{P}^{-1}(\bm{L}\bm{R}^\top)-\bm{\Theta}_\ast\|_\textup{F}^2 + \xi_R^2.
    \end{split}
  \end{equation}
\end{definition}

Here, $\xi_{L}$ and $\xi_R$ represent the statistical error floor induced by truncation bias and variance, while $\phi$ controls the sensitivity to parameter estimation error.

\begin{theorem}[Local Convergence of SRGD]\label{thm:1}
  Suppose the loss function $\mathcal{L}$ satisfies Definition \ref{def:1}, and the stability condition of Definition \ref{def:2} holds uniformly for all iterates with parameters $\phi\lesssim\alpha^2$ and $\xi_{L}^2+\xi_{R}^2\lesssim\alpha^3\beta^{-1}\sigma_K^2$. Furthermore, assume the initial distance satisfies $d(\bm{F}_0,\bm{F}_\ast)^2 \lesssim \alpha\beta^{-1}\sigma_K^2$, and the step size is chosen as $\eta\asymp\beta^{-1}$. Then, for all $j=1,2,\dots,J$, the iterates produced by Algorithm \ref{alg:1} satisfy
  \begin{equation}
    d(\bm{F}_j,\bm{F}_\ast)^2 \asymp \|\bm{\Theta}_j-\bm{\Theta}_\ast\|_\textup{F}^2 \leq (1-C\alpha\beta^{-1})^j d(\bm{F}_0,\bm{F}_\ast)^2 + C\alpha^{-2}(\xi_L^2 + \xi_R^2),
  \end{equation}
  where $\bm{\Theta}_j=\mathcal{P}^{-1}(\bm{L}_j\bm{R}_j^\top)$. Moreover, $\bm{L}_j$ and $\bm{R}_j$ are nonsingular for all $j=1,\dots,J$.
\end{theorem}

Theorem \ref{thm:1} underscores two fundamental advantages of the SRGD framework. First, the algorithm achieves local linear convergence at a rate determined solely by the condition number of the expected loss ($\beta/\alpha$). Crucially, this rate is independent of the condition number of the underlying matrix factorization ($\kappa$), indicating that the de-scaling step effectively preconditions the optimization landscape against geometric ill-conditioning. Second, the estimation error stabilizes at a floor determined by the statistical quality of the robust gradients ($\xi_L, \xi_R$). This clean separation between optimization geometry and statistical robustness validates our design: the algorithm maintains efficient convergence and robust error control even when the target matrix is highly ill-conditioned or the factors are arbitrarily scaled.

\section{Sparsity and Scaled Hard Thresholding}\label{sec:3}

\subsection{Sparse Kronecker Factors and Scaled Hard Thresholding}\label{sec:3.1}

In high-dimensional regimes where the ambient dimensions $d_1, d_2$ exceed the sample size $n$, exploiting low-dimensional structure is essential for consistent estimation. While the Kronecker rank $K$ captures global correlation, additional parsimony is often achieved by assuming the factor matrices $\{\bm{A}_k, \bm{B}_k\}$ are sparse. 
In the permuted domain $\mathcal{P}(\bm{\Theta}) = \bm{L}\bm{R}^\top$, this corresponds to row-wise sparsity in $\bm{L}$ and $\bm{R}$. Specifically, if the rows of $\bm{L}$ are sparse, entire blocks of the Kronecker product are zeroed out.

However, enforcing this sparsity is non-trivial due to the factorization ambiguity. A standard approach would be to apply row-wise hard thresholding directly to the iterates $\bm{L}_j$. This approach is flawed by the scaling ambiguities.
Consider the transformation $(\bm{L}', \bm{R}') = (\bm{L}\bm{Q}, \bm{R}\bm{Q}^{-\top})$. The row norms of the transformed factor $\bm{L}'$ are:
\begin{equation}
    \|\bm{e}_i^\top \bm{L}'\|_2 = \|\bm{e}_i^\top \bm{L}\bm{Q}\|_2,
\end{equation}
where $\bm{e}_i$ is the coordinate vector whose $i$-th element is one and others zero.
Since $\bm{Q}$ is arbitrary, it can arbitrarily inflate or deflate specific row norms, thereby changing which rows survive the thresholding operation. Consequently, the recovered support set depends on the arbitrary basis $\bm{Q}$ rather than the intrinsic signal, leading to inconsistent variable selection.

To resolve this, we introduce Scaled Hard Thresholding (SHT). The core principle is to perform variable selection in a \textit{scale-invariant domain}. We define the \textit{scaled factors} as:
\begin{equation}\label{eq:scaled_factors}
    \overline{\bm{L}} = \bm{L}(\bm{R}^\top\bm{R})^{1/2} \quad \text{and} \quad \overline{\bm{R}} = \bm{R}(\bm{L}^\top\bm{L})^{1/2}.
\end{equation}
The row norms of these scaled factors are invariant to the parameterization. For any $\bm{Q} \in \text{GL}(K)$, the $i$-th row norm of the transformed scaled factor $\overline{\bm{L}}'$ satisfies:
\begin{align}
    \|\bm{e}_i^\top \overline{\bm{L}}'\|_2^2 &= \|\bm{e}_i^\top \bm{L}\bm{Q} (\bm{Q}^{-1}\bm{R}^\top\bm{R}\bm{Q}^{-\top})^{1/2}\|_2^2 \notag 
    = \bm{e}_i^\top \bm{L} (\bm{R}^\top\bm{R}) \bm{L}^\top \bm{e}_i = \|\bm{e}_i^\top \overline{\bm{L}}\|_2^2.
\end{align}
Thus, thresholding based on the row norms of $\overline{\bm{L}}$ selects a support set that depends strictly on the underlying parameter $\bm{\Theta}$, respecting the geometry of the Kronecker manifold.

The SRGD--SHT algorithm (Algorithm \ref{alg:2}) integrates this operator into the iterative update. In each iteration $j$, we compute robust gradient updates to obtain intermediate iterates $\widetilde{\bm{L}}_{j+1}, \widetilde{\bm{R}}_{j+1}$, apply hard thresholding to their scaled counterparts, and then restore the original scaling to maintain the factorization structure.

\begin{algorithm}[!htp]
  \setstretch{1.4}
  \caption{SRGD with Scaled Hard Thresholding (SRGD--SHT)}
  \label{alg:2}
  \textbf{Input}: initial value $\bm{F}_0$, step size $\eta>0$, truncation parameter $\tau>0$, number of iterations $J$, and sparsity levels $s_{L},s_{R}$\\
  \textbf{For }$j=0,\dots,J-1$\\
  \hspace*{0.5cm}\textit{// 1. Scaled Robust Gradient Descent}\\
  \hspace*{1cm}$\widetilde{\bm{L}}_{j+1} = \bm{L}_j - \eta\cdot \bm{G}_{L}(\bm{L}_j,\bm{R}_j;\tau)(\bm{R}_j^\top\bm{R}_j)^{-1/2}$\\
  \hspace*{1cm}$\widetilde{\bm{R}}_{j+1} = \bm{R}_j - \eta\cdot \bm{G}_{R}(\bm{L}_j,\bm{R}_j;\tau)(\bm{L}_j^\top\bm{L}_j)^{-1/2}$\\
  \hspace*{0.5cm}\textit{// 2. Scaled Hard Thresholding}\\
  \hspace*{1cm}$\bm{L}_{j+1} = \text{HT}(\widetilde{\bm{L}}_{j+1}(\widetilde{\bm{R}}_{j+1}^\top\widetilde{\bm{R}}_{j+1})^{1/2},s_L)(\widetilde{\bm{R}}_{j+1}^\top\widetilde{\bm{R}}_{j+1})^{-1/2}$\\
  \hspace*{1cm}$\bm{R}_{j+1} = \text{HT}(\widetilde{\bm{R}}_{j+1}(\widetilde{\bm{L}}_{j+1}^\top\widetilde{\bm{L}}_{j+1})^{1/2},s_R)(\widetilde{\bm{L}}_{j+1}^\top\widetilde{\bm{L}}_{j+1})^{-1/2}$\\
  \textbf{End for}\\
  \textbf{Output} $\widehat{\bm{F}}=(\bm{L}_J,\bm{R}_J)$ and $\widehat{\bm{\Theta}}=\mathcal{P}^{-1}(\bm{L}_J\bm{R}_J^\top)$
\end{algorithm}

\subsection{Computational Convergence Analysis with Sparsity}\label{sec:3.3}

We analyze the local convergence of SRGD--SHT assuming the ground truth factors $\bm{L}_\ast$ and $\bm{R}_\ast$ are row-sparse with supports of size $s_{L,\ast}$ and $s_{R,\ast}$. Let $s_L \ge s_{L,\ast}$ and $s_R \ge s_{R,\ast}$ be the user-specified sparsity levels.
To handle high-dimensional constraints, we relax the gradient stability condition (Definition \ref{def:2}) to hold only over sparse subsets.

\begin{definition}[Stability on Sparse Sets]\label{def:4}
    The de-scaled robust gradient functions $\bm{G}_L, \bm{G}_R$ are stable on sparse sets if there exist constants $\phi$, $\xi_{L,s}$, and $\xi_{R,s}$ such that for all row-subsets $S_1, S_2$ with $|S_1|\leq s_L + s_{L,\ast}$ and $|S_2|\leq s_R + s_{R,\ast}$:
    \begin{equation}
        \begin{split}
            \|\{\bm{G}_L - \mathbb{E}[\nabla_{\bm{L}}f](\bm{R}^\top\bm{R})^{-1/2}\}_{S_1}\|_\textup{F}^2 & \leq \phi \|\mathcal{P}^{-1}(\bm{L}\bm{R}^\top)-\bm{\Theta}_\ast\|_\textup{F}^2 + \xi_{L,s}^2,\\
            \|\{\bm{G}_R - \mathbb{E}[\nabla_{\bm{R}}f](\bm{L}^\top\bm{L})^{-1/2}\}_{S_2}\|_\textup{F}^2 & \leq \phi \|\mathcal{P}^{-1}(\bm{L}\bm{R}^\top)-\bm{\Theta}_\ast\|_\textup{F}^2 + \xi_{R,s}^2,
        \end{split}
    \end{equation}
    where $\bm{M}_S$ denotes the submatrix of $\bm{M}$ formed by rows indexed by $S$.
\end{definition}

This condition implies that the truncated gradients concentrate around their population means when projected onto low-dimensional subspaces, even if the ambient dimension is large.

\begin{theorem}[Local Convergence of SRGD--SHT]\label{thm:2}
  Suppose the loss function $\mathcal{L}$ satisfies the RCG condition (Definition \ref{def:1}) and the sparse stability condition (Definition \ref{def:4}) holds with $\phi\lesssim\alpha^2$ and $\xi_{L,s}^2 + \xi_{R,s}^2\lesssim\alpha^3\beta^{-1}\sigma_K^2$. 
  If the sparsity levels satisfy $s_L\gtrsim(\beta/\alpha)^2 s_{L,\ast}$ and $s_R\gtrsim(\beta/\alpha)^2 s_{R,\ast}$, the initial distance satisfies $d(\bm{F}_0,\bm{F}_\ast)^2\lesssim\alpha\beta^{-1}\sigma_K^2$, and $\eta\asymp\beta^{-1}$, then for all $j=1,2,\dots,J$:
  \begin{equation}
    \|\bm{\Theta}_j-\bm{\Theta}_\ast\|_\textup{F}^2 \leq (1-C\alpha\beta^{-1})^j d(\bm{F}_0,\bm{F}_\ast)^2 + C\alpha^{-2}(\xi_{L,s}^2 + \xi_{R,s}^2).
  \end{equation}
  where $\bm{\Theta}_j = \mathcal{P}^{-1}(\bm{L}_j\bm{R}_j^\top)$. Moreover, $\bm{L}_j$ and $\bm{R}_j$ are nonsingular for all $j$.
\end{theorem}

Theorem \ref{thm:2} establishes that SRGD--SHT attains local linear convergence up to a statistical error floor. Crucially, because the gradient error is measured only on sparse sets ($\xi_{L,s}, \xi_{R,s}$), the resulting statistical rate scales with the sparsity level $s$ rather than the ambient dimensions $d_1, d_2$.

\section{Application to Matrix Models}\label{sec:4}

In this section, we instantiate the proposed SRGD--SHT framework for three canonical matrix estimation problems: matrix trace regression, matrix generalized linear models (GLMs), and bilinear models. We specify the gradient oracles and derive explicit sample size and moment requirements for each case. We focus on Algorithm \ref{alg:2}, with Algorithm \ref{alg:1} recovered when the sparsity levels satisfy $s_L=d_1$ and $s_R=d_2$. For convenience, we define $s=\max(\min(s_L+s_{L,\ast},d_1),\min(s_R+s_{R,\ast},d_2))$ and $d=\max(d_1,d_2)$.

\subsection{Matrix Trace Regression}\label{sec:4.1}

We first consider the matrix trace regression model
\begin{equation}\label{eq:MatrixTraceReg}
  Y_i = \langle\bm{X}_i,\bm{\Theta}_\ast\rangle + E_i,\quad i=1,\dots,n,
\end{equation}
where $\bm{\Theta}_\ast$ admits the Kronecker structure in \eqref{eq:KP}. We minimize the quadratic loss function $\mathcal{L}_n(\bm{\Theta}) = (2n)^{-1}\sum_{i=1}^n(Y_i-\langle\bm{X}_i,\bm{\Theta}\rangle)^2$.

To ensure consistent estimation, we employ a two-stage estimation procedure.
\begin{enumerate}
    \item \textbf{Initialization (Robust Dantzig Selector):} We first compute a coarse initial estimate using a robustified Dantzig selector. Let $\widehat{\bm{\Sigma}}_x(\tau_x) = n^{-1}\sum_{i=1}^n\text{T}(\bm{x}_i\bm{x}_i^\top,\tau_x)$ and $\widehat{\bbm{\sigma}}_{yx}(\tau_{yx}) = n^{-1}\sum_{i=1}^n\text{T}(Y_i\bm{x}_i,\tau_{yx})$ be truncated covariance estimators, where $\bm{x}_i = \text{vec}(\bm{X}_i)$. We solve:
    \begin{equation}\label{eq:DS_init}
      \widehat{\bbm{\theta}}_{\text{DS}} = \argmin_{\bbm{\theta}\in\mathbb{R}^{p_1p_2}}\|\bbm{\theta}\|_1\quad\text{s.t.}\quad\|\widehat{\bm{\Sigma}}_{x}(\tau_x)\bbm{\theta} - \widehat{\bbm{\sigma}}_{yx}(\tau_{yx})\|_\infty \leq R.
    \end{equation}
    We reshape $\widehat{\bbm{\theta}}_{\text{DS}}$ into a matrix, apply the permutation $\mathcal{P}$, and perform SVD followed by Scaled Hard Thresholding (SHT) to obtain initial factors $\bm{L}_0, \bm{R}_0$.
    \item \textbf{Refinement (SRGD--SHT):} Starting from $(\bm{L}_0, \bm{R}_0)$, we run Algorithm \ref{alg:2} using the de-scaled robust gradients:
    \begin{equation}
      \begin{split}
        \bm{G}_L & = \frac{1}{n}\sum_{i=1}^n\text{T}\left(\left\{\langle\mathcal{P}(\bm{X}_i),\bm{L}\bm{R}^\top\rangle-Y_i\right\}\mathcal{P}(\bm{X}_i)\bm{R}(\bm{R}^\top\bm{R})^{-1/2},\tau\right), \\
        \bm{G}_R & = \frac{1}{n}\sum_{i=1}^n\text{T}\left(\left\{\langle\mathcal{P}(\bm{X}_i),\bm{L}\bm{R}^\top\rangle-Y_i\right\}\mathcal{P}(\bm{X}_i)^\top\bm{L}(\bm{L}^\top\bm{L})^{-1/2},\tau\right).
      \end{split}
    \end{equation}
\end{enumerate}

To quantify robustness, we define the requisite moment conditions. For any $\eta>0$, define the global moments $M_{x,\eta} = \sup_{\|\bm{v}\|_2=1}\mathbb{E}[|\textup{vec}(\bm{X}_i)^\top\bm{v}|^{\eta}]$ and $M_{e,\eta} = \mathbb{E}[|E_i|^{\eta}|\bm{X}_i]$. For high-dimensional inference, we control the predictor moments on sparse subspaces:
\begin{equation}\label{eq:P_X_moment}
  M_{x,\eta,s} = \max\left(\max_{j}\sup_{\bm{v}\in\mathbb{S}(d_2,s_2)}\mathbb{E}\left[|\bm{c}_j^\top\mathcal{P}(\bm{X}_i)\bm{v}|^{\eta}\right],\max_{k}\sup_{\bm{v}\in\mathbb{S}(d_1,s_2)}\mathbb{E}\left[|\bm{v}^\top\mathcal{P}(\bm{X}_i)\bm{c}_k|^{\eta}\right]\right),
\end{equation}
where $\mathbb{S}(d,s)=\{\bm{v}\in\mathbb{R}^{d}:\|\bm{v}\|_0=s,\|\bm{v}\|_2=1\}$.

\begin{assumption}[Moment Conditions]
  \label{asmp:moment}
  The vectorized covariate $\textup{vec}(\bm{X}_i)$ has zero mean and positive-definite covariance $\bm{\Sigma}_x$ with condition number $\kappa_x = \beta_x/\alpha_x$. There exist $\epsilon\in(0,1]$ and $\lambda\in(0,1]$ such that $M_{e,1+\epsilon}<\infty$ and $M_{x,2+2\lambda}<\infty$.
\end{assumption}

We now state the main result establishing the convergence of the two-stage procedure.

\begin{theorem}[Convergence of Matrix Trace Regression]
  \label{thm:matrix_trace_reg}
  Suppose Assumption \ref{asmp:moment} holds. Let the tuning parameters for initialization $(\tau_x, \tau_{yx}, R)$ be chosen as in Proposition \ref{prop:matreg_initial} (see Appendix \ref{append:B}), and the SRGD parameters satisfy $\tau \asymp [nM_{\textup{eff},1+\epsilon,s}/\log d]^{1/(1+\epsilon)}$ and $\eta\asymp\beta_x^{-1}$.
  If the sample size $n$ satisfies:
  \begin{equation}\label{eq:sample_size_global}
    n \gtrsim \mathfrak{C}_1 \cdot \left(s K \right)^{\frac{1+\min(\epsilon,\lambda)}{2\min(\epsilon,\lambda)}}\log d,
  \end{equation}
  where $\mathfrak{C}_1$ is a polynomial function of the moment constants and condition numbers (see Appendix \ref{append:B}), then with probability at least $1-C\exp(-c\log d)$, the final output $\widehat{\bm{\Theta}}$ satisfies:
  \begin{equation}\label{eq:final_rate}
    \|\widehat{\bm{\Theta}}-\bm{\Theta}_\ast\|_\textup{F} \lesssim \alpha_x^{-1}\sqrt{s}\left[\frac{M_{\textup{eff},1+\epsilon,s}^{1/\epsilon}\log d}{n}\right]^{\frac{\epsilon}{1+\epsilon}}.
  \end{equation}
\end{theorem}

Theorem \ref{thm:matrix_trace_reg} guarantees that the proposed pipeline converges to the optimal error floor provided the sample size is sufficient. Importantly, the result elucidates the distinct roles played by the predictor tail index $\lambda$ and the noise tail index $\epsilon$.

\begin{remark}[Optimality and Phase Transitions]\label{rem:optimality}
    The theoretical guarantees exhibit two distinct phase transitions governed by the tail behaviors of the data. 
    \begin{itemize}
      \item[1.] Statistical Efficiency ($\epsilon$): The final estimation error in \eqref{eq:final_rate} depends solely on the noise moment $\epsilon$. In the finite variance regime ($\epsilon=1$), the estimator recovers the optimal parametric rate $O(\sqrt{s \log d / n})$ characterizing Gaussian settings. As the noise tails become heavier ($\epsilon < 1$), the rate naturally slows to $O(n^{-\frac{\epsilon}{1+\epsilon}})$, reflecting the information-theoretic limit of robust estimation. This phase transition phenomenon has also been observed in other regression literature and is proved to be minimax optimal \citep{sun2020adaptive,tan2023sparse}.
      \item[2.] Sample Complexity ($\lambda$): The sample size requirement \eqref{eq:sample_size_global} depends on $\min(\epsilon, \lambda)$. This implies that the bottleneck for estimation is determined by the heavier of the two tails (predictor or noise). Specifically, if the predictors are very heavy-tailed ($\lambda < \epsilon$), the sample complexity scales as $O(s^{\frac{1+\lambda}{2\lambda}}\log d)$, which is super-linear in the dimension. This increased requirement reflects the difficulty of ensuring uniform concentration of the robust gradient and covariance estimators when the predictors themselves are prone to large deviations.
    \end{itemize}
    
\end{remark}

\subsection{Matrix Generalized Linear Models}\label{sec:4.2}

We next consider Matrix Generalized Linear Models (Matrix GLMs), which relate a matrix predictor $\bm{X}_i\in\mathbb{R}^{p\times q}$ to a scalar response $Y_i\in\mathbb{R}$ via an exponential family distribution:
\begin{equation}\label{eq:matrixGLM}
  f_{Y_i}(y_i|\bm{X}_i) \propto \exp\left\{\frac{y_i\langle\bm{X}_i,\bm{\Theta}\rangle-g(\langle\bm{X}_i,\bm{\Theta}\rangle)}{c(\sigma)}\right\},\quad i=1,\dots,n.
\end{equation}
The negative log-likelihood loss is $\mathcal{L}_n(\bm{\Theta})=n^{-1}\sum_{i=1}^n [g(\langle\bm{X}_i,\bm{\Theta}\rangle) - Y_i\langle\bm{X}_i,\bm{\Theta}\rangle]$. We impose the following smoothness condition on the link function.

\begin{assumption}[Lipschitz Smoothness]
  \label{asmp:Lipschitz}
  The cumulant function $g(\cdot)$ has an $L$-Lipschitz continuous derivative; i.e., $|g'(a)-g'(b)|\leq L|a-b|$ for all $a,b\in\mathbb{R}$.
\end{assumption}

Analogous to the trace regression case, we employ a two-stage estimation procedure:
\begin{enumerate}
    \item \textbf{Initialization (Robust Lasso):} We solve a robustified Lasso problem on the vectorized parameter. Let $\widetilde{\bm{X}}_i(\tau_x) = \text{T}(\bm{X}_i,\tau_x)$ be the truncated predictors. We compute:
    \begin{equation}\label{eq:Lasso_init}
      \widehat{\bbm{\theta}}_{\textup{LA}} = \argmin_{\bbm{\theta}}\left\{\frac{1}{n}\sum_{i=1}^n\left[g(\langle\text{vec}(\widetilde{\bm{X}}_i(\tau_x)),\bbm{\theta}\rangle) - Y_i\langle\text{vec}(\widetilde{\bm{X}}_i(\tau_x)),\bbm{\theta}\rangle\right] + R\|\bbm{\theta}\|_1\right\}.
    \end{equation}
    The solution is matricized, permuted, and subjected to SVD-based SHT to yield $\bm{L}_0, \bm{R}_0$.
    \item \textbf{Refinement (SRGD--SHT):} We refine the estimates using Algorithm \ref{alg:2} with gradients:
    \begin{equation}
      \begin{split}
        \bm{G}_L & = \frac{1}{n}\sum_{i=1}^n\text{T}\left(\left\{g'(\langle\mathcal{P}(\bm{X}_i),\bm{L}\bm{R}^\top\rangle)-Y_i\right\}\mathcal{P}(\bm{X}_i)\bm{R}(\bm{R}^\top\bm{R})^{-1/2};\tau\right), \\
        \bm{G}_R & = \frac{1}{n}\sum_{i=1}^n\text{T}\left(\left\{g'(\langle\mathcal{P}(\bm{X}_i),\bm{L}\bm{R}^\top\rangle)-Y_i\right\}\mathcal{P}(\bm{X}_i)^\top\bm{L}(\bm{L}^\top\bm{L})^{-1/2};\tau\right).
      \end{split}
    \end{equation}
\end{enumerate}

We adopt the same moment definitions as in Section \ref{sec:4.1}, denoting the effective conditional noise moment by $M_{e,\eta} = \mathbb{E}[|Y_i - g'(\langle\bm{X}_i,\bm{\Theta}_\ast\rangle)|^{\eta}|\bm{X}_i]$.

\begin{theorem}[Convergence of Matrix GLM]
  \label{thm:matrixGLM}
  Suppose Assumptions \ref{asmp:moment} and \ref{asmp:Lipschitz} hold. Let the initialization parameters be tuned as in Proposition \ref{prop:logistic_initial} (Appendix \ref{append:C}), and SRGD parameters satisfy $\tau \asymp [nM_{\textup{eff},1+\epsilon,s}/\log d]^{1/(1+\epsilon)}$ and $\eta\asymp\beta_x^{-1}$.
  If the sample size satisfies  
  \begin{equation}
    n \gtrsim \mathfrak{C}_2 \cdot \left(s K \right)^{\frac{1+\lambda}{2\lambda}}\log d,
  \end{equation}
  with a constant $\mathfrak{C}_2$ adapted for GLM (see Appendix \ref{append:C}), then with probability at least $1-C\exp(-c\log d)$, the two-stage estimator $\widehat{\bm{\Theta}}$ satisfies:
  \begin{equation}
    \|\widehat{\bm{\Theta}}-\bm{\Theta}_\ast\|_\textup{F} \lesssim \alpha_x^{-1}\sqrt{s}\left[\frac{M_{\textup{eff},1+\epsilon,s}^{1/\epsilon}\log d}{n}\right]^{\frac{\epsilon}{1+\epsilon}}.
  \end{equation}
\end{theorem}

Theorem \ref{thm:matrixGLM} extends our robust guarantees to the exponential family. A particularly important instance is \textit{matrix logistic regression}, which is widely used in classification tasks.

\begin{corollary}[Matrix Logistic Regression]
  \label{cor:logistic_rates}
  Consider the logistic regression model with $g(t)=\log(1+e^t)$. Since the response $Y_i$ is bounded, we have $M_{e,\eta} < \infty$ for all $\eta$. Setting $\epsilon=1$, and assuming predictors have finite $(2+2\lambda)$-th moments, the two-stage estimator achieves the rate
  \begin{equation}
    \|\widehat{\bm{\Theta}} - \bm{\Theta}_\ast\|_\textup{F} \lesssim \alpha_x^{-1}\sqrt{\frac{s M_{x,2,s} \log d}{n}}.
  \end{equation}
\end{corollary}

This result highlights a key advantage: for logistic regression, the statistical rate is not penalized by the heavy tails of the predictors; robustness only imposes a sample complexity cost via $\lambda$ to ensure uniform convergence of the gradients. This relaxes the finite fourth-moment conditions commonly required in prior robust GLM literature \citep{prasad2020robust,zhu2021taming}.

\subsection{Bilinear Regression Models}\label{sec:4.3_bilinear}

Finally, we consider the Bilinear Regression model, which captures the relationship between a matrix predictor $\bm{X}_i\in\mathbb{R}^{q_1\times q_2}$ and a matrix response $\bm{Y}_i\in\mathbb{R}^{p_1\times p_2}$:
\begin{equation}
  \bm{Y}_i = \bm{A}\bm{X}_i\bm{B}^\top + \bm{E}_i, \quad i=1,\dots,n,
\end{equation}
where $\bm{E}_i$ is the noise matrix. Vectorizing this relationship yields the Kronecker structure:
\begin{equation}\label{eq:bilinear}
  \bm y_i=\text{vec}(\bm{Y}_i^\top) = (\bm{A}\otimes\bm{B})\text{vec}(\bm{X}_i^\top) + \text{vec}(\bm{E}_i^\top) = \bm{\Theta}_\ast\bm x_i + \bm e_i.
\end{equation}
Here, the target $\bm{\Theta}_\ast = \bm{A} \otimes \bm{B}$ has Kronecker rank $K=1$. This model is fundamental in spatio-temporal modeling and matrix autoregression \citep{chen2021autoregressive,fan2025matrix}.

We estimate $\bm{\Theta}_\ast$ by minimizing the least-squares loss $\mathcal{L}_n(\bm{\Theta}) = (2n)^{-1}\sum_{i}\|\bm y_i - \bm{\Theta}\bm x_i\|_2^2$. We employ the following two-stage estimation procedure:

\begin{enumerate}
    \item \textbf{Initialization (Robust Dantzig Selector):} 
    Recognizing that $\bm{\Theta}_\ast = \bm{\Sigma}_{yx}\bm{\Sigma}_x^{-1}$, we estimate the cross-covariance $\widehat{\bm{\Sigma}}_{yx}(\tau_{yx}) = n^{-1}\sum_i \text{T}(\bm{y}_i\bm{x}_i^\top, \tau_{yx})$ and auto-covariance $\widehat{\bm{\Sigma}}_x(\tau_x)$ via truncation. We solve the robust Dantzig selector:
    \begin{equation}\label{eq:bilinear_initial}
      \widehat{\bm{\Theta}}_{\text{DS}} = \argmin_{\bm{\Theta}}\|\bm{\Theta}\|_1\quad\text{s.t.}\quad \|\widehat{\bm{\Sigma}}_{yx}(\tau_{yx}) - \bm{\Theta}\widehat{\bm{\Sigma}}_x(\tau_x)\|_\infty\leq R.
    \end{equation}
    The solution is matricized and subjected to rank-1 SVD and SHT to yield initial factors $\bm{L}_0, \bm{R}_0$.
    
    \item \textbf{Refinement (SRGD--SHT):} We refine the factors using Algorithm \ref{alg:2} ($K=1$) with gradients:
    \begin{equation}
      \begin{split}
        \bm{G}_L & = \frac{1}{n}\sum_{i=1}^n\text{T}\Big(\left(\bm X_i\otimes (\bm A\bm X_i\bm B^\top-\bm Y_i)\right)\bm{R}(\bm{R}^\top\bm{R})^{-1/2};\tau\Big),\\
        \bm{G}_R & = \frac{1}{n}\sum_{i=1}^n\text{T}\Big(\left(\bm X_i^\top \otimes (\bm A\bm X_i\bm B^\top-\bm Y_i)^\top\right)\bm{L}(\bm{L}^\top\bm{L})^{-1/2};\tau\Big).
      \end{split}
    \end{equation}
\end{enumerate}

To rigorously characterize robustness, we must control the tail behavior of the gradient noise. Unlike trace regression, the gradient noise here involves the interaction term $\bm{J}_i = \bm X_i \otimes \bm E_i$. We define the effective interaction moment on sparse sets as:
\begin{equation}\label{eq:bilinear_gradient}
  M_{\textup{eff},\eta,s} := \max\left\{\max_{j}\sup_{\bm{v}\in\mathbb{S}(d_2,s_R)}\mathbb{E}\left[\left|\bm{c}_j^\top\bm{J}_i\bm{v}\right|^{\eta}\right],\max_{k}\sup_{\bm{v}\in\mathbb{S}(d_1,s_L)}\mathbb{E}\left[\left|\bm{v}^\top\bm{J}_i\bm{c}_k\right|^{\eta}\right]\right\}.
\end{equation}

\begin{assumption}[Moment Conditions]
  \label{asmp:moment2}
  The predictor $\bm{x}_i$ satisfies the covariance and moment conditions of Assumption \ref{asmp:moment} with tail index $\lambda \in (0,1]$. Furthermore, there exists $\epsilon \in (0,1]$ such that the interaction moment $M_{\textup{eff},1+\epsilon,s} < \infty$.
\end{assumption}

\begin{theorem}[Convergence of Bilinear Estimation]
  \label{thm:bilinear_model}
  Suppose Assumption \ref{asmp:moment2} holds. Let initialization parameters be tuned as in Proposition \ref{prop:bilinear_init} (Appendix \ref{append:D}), and SRGD parameters satisfy $\tau \asymp [nM_{\textup{eff},1+\epsilon,s}/\log d]^{1/(1+\epsilon)}$ and $\eta\asymp\beta_x^{-1}$. 
  If the sample size satisfies:
  \begin{equation}
    n \gtrsim \mathfrak{C}_3 \cdot s^{\frac{1+\min(\epsilon,\lambda)}{2\min(\epsilon,\lambda)}}\log d,
  \end{equation}
  where $\mathfrak{C}_3$ is a constant defined in Appendix \ref{append:D}, then with probability at least $1-C\exp(-c\log d)$, the estimator $\widehat{\bm{\Theta}}$ satisfies:
  \begin{equation}
    \|\widehat{\bm{\Theta}}-\bm{\Theta}_\ast\|_\textup{F} \lesssim \alpha_x^{-1}\sqrt{s}\left[\frac{M_{\textup{eff},1+\epsilon,s}^{1/\epsilon}\log d}{n}\right]^{\frac{\epsilon}{1+\epsilon}}.
  \end{equation}
\end{theorem}

\begin{remark}[The Challenge of Multiplicative Noise]
  The convergence rate is governed by $\epsilon$, which characterizes the tail of the interaction $\bm{X}_i\otimes \bm{E}_i$. This highlights a unique challenge in bilinear models: even if the noise $\bm{E}_i$ and predictors $\bm{X}_i$ are independently moderately heavy-tailed, their product can exhibit extremely heavy tails. For instance, if both possess only finite second moments and they are independent, their product has only a finite second moment, making robustification essential. Our SRGD framework handles this naturally via the truncation of the interaction term $\bm{J}_i$.
\end{remark}

\section{Simulation Experiments}\label{sec:5}

In this section, we conduct numerical experiments to validate our theoretical findings and evaluate the empirical performance of the proposed SRGD--SHT framework. The simulations are designed to achieve three main objectives: (1) verify the theoretical relationship between statistical convergence rates and the tail behavior of the noise (Theorems \ref{thm:matrix_trace_reg} and \ref{thm:bilinear_model}); (2) demonstrate that our scale-invariant optimization scheme is computationally efficient and robust to the condition number $\kappa$, unlike standard gradient methods; and (3) highlight the superiority of the proposed two-stage estimation procedure over existing baselines in settings characterized by heavy-tailed predictors, heavy-tailed noise, and scaling ambiguity.

For all experiments, we initialize using the Robust Dantzig Selector as described in Section \ref{sec:4}. For the refinement stage, the step size is set to $\eta=0.01$. Regarding the truncation threshold $\tau$, while our theory suggests a scaling of $\tau \asymp n^{1/(1+\epsilon)}$, the exact moment $\epsilon$ is often unknown. In practice, we select $\tau$ via five-fold cross-validation on a logarithmic grid. Our experiments confirm that the performance is stable around the optimal $\tau$, suggesting the method is not overly sensitive to precise tuning.

Regarding the third objective, since the qualitative patterns are consistent across different model specifications, we report findings for matrix trace regression in the main text. Comprehensive results for bilinear models and matrix logistic regression are provided in Appendix \ref{append:E} of the Supplementary Materials.

\subsection{Experiment 1: Phase Transitions and Tail Behavior}

We first validate the theoretical prediction that the statistical error rate is governed by the tail index $\epsilon$ of the noise (and interaction terms). We focus on matrix trace regression and bilinear models.
According to Theorems \ref{thm:matrix_trace_reg} and \ref{thm:bilinear_model}, the estimation error satisfies $\|\widehat{\bm{\Theta}}-\bm{\Theta}_\ast\|_\text{F} \propto n^{-\epsilon/(1+\epsilon)}$. Consequently, in a log-log plot of error vs. sample size, we expect a linear relationship with slope $-\epsilon/(1+\epsilon)$.

We adopt the following settings.
\begin{itemize}
  \item \textit{Matrix Trace Regression:} $p_1=q_1=p_2=q_2=6$, rank $K=1$, with sparsity $s_{L,\ast}=s_{R,\ast}=5$. The true factors are set to $\bm L_\ast=\bm R_\ast=(1,\dots, 1,0,\dots, 0)^\top$.
  \item \textit{Bilinear Model:} $p_1=p_2=5$, $q_1=q_2=10$, rank $K=1$, sparsity $s_{L,\ast}=s_{R,\ast}=8$. Nonzero entries of $\bm A_\ast$ and $\bm B_\ast$ are set to $\sqrt{5/8}$.
  \item \textit{Data Generation:} Predictors $\bm X_i$ are generated with i.i.d. $t_{2.5}$ entries. Noise terms ($e_i$ or entries of $\bm E_i$) are drawn from a $t_{1.05+\epsilon}$ distribution, with $\epsilon$ varying in $\{0.2, 0.4, \dots, 2.0\}$. This ensures the existence of the $(1+\epsilon)$-th moment.
\end{itemize}  

To eliminate the initialization effects, we use the ground truth initialization in this experiment. Optimization parameters are set to $\eta=0.01$ and $\tau \propto n^{1/(1+\epsilon)}$ as per theory. We vary $n$ and report the mean estimation error over 200 repetitions. Figure \ref{fig:exp1} displays the log estimation error trajectories. For $\epsilon \in [0.2, 1.0]$, the slopes become progressively steeper (more negative) as $\epsilon$ increases, closely matching the theoretical prediction $-\epsilon/(1+\epsilon)$.
Notably, for $\epsilon > 1$, the slopes stabilize and become nearly constant. This empirically confirms the \textit{phase transition} discussed in Remark \ref{rem:optimality}: once the noise possesses a finite variance ($\epsilon \ge 1$), the estimator recovers the optimal parametric rate $O(n^{-1/2})$, and further lightening of the tails yields no additional gain in the polynomial rate.

\begin{figure}[htbp]
    \centering
    \begin{subfigure}[b]{0.48\textwidth}
        \centering
        \includegraphics[width=\linewidth]{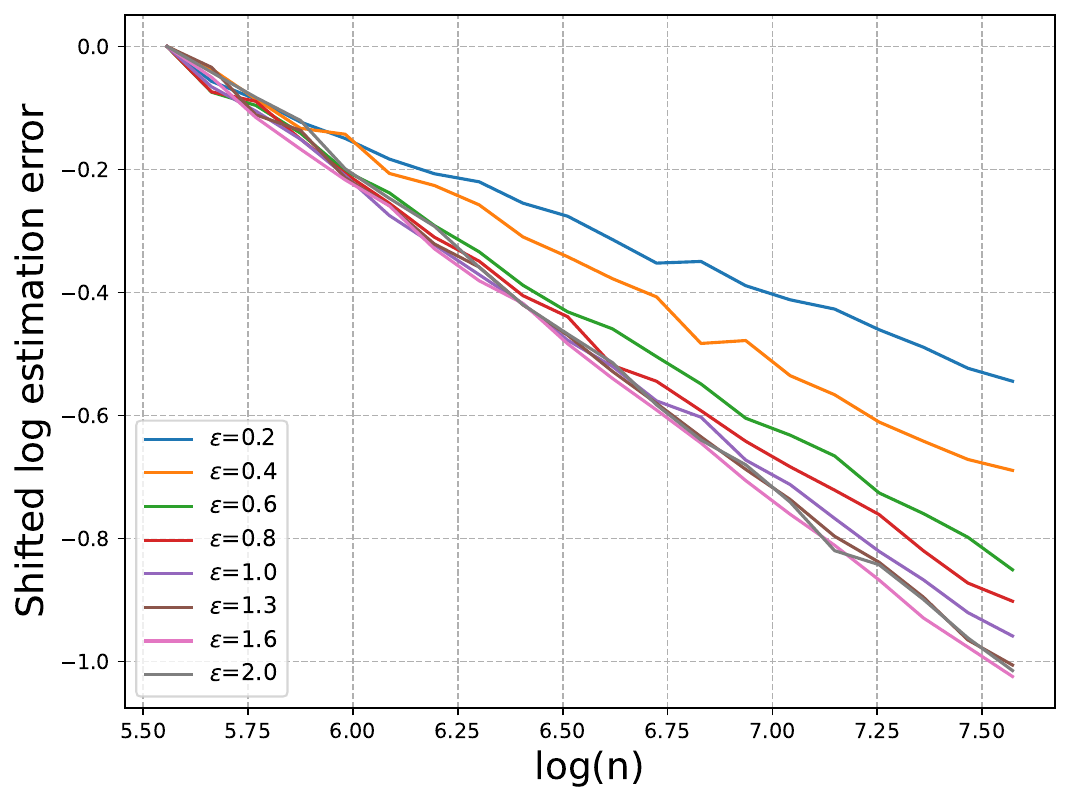}
        \subcaption{Matrix Trace Regression}
        \label{subfig:exp1_matreg}
    \end{subfigure}
    \hfill
    \begin{subfigure}[b]{0.48\textwidth}
        \centering
        \includegraphics[width=\linewidth]{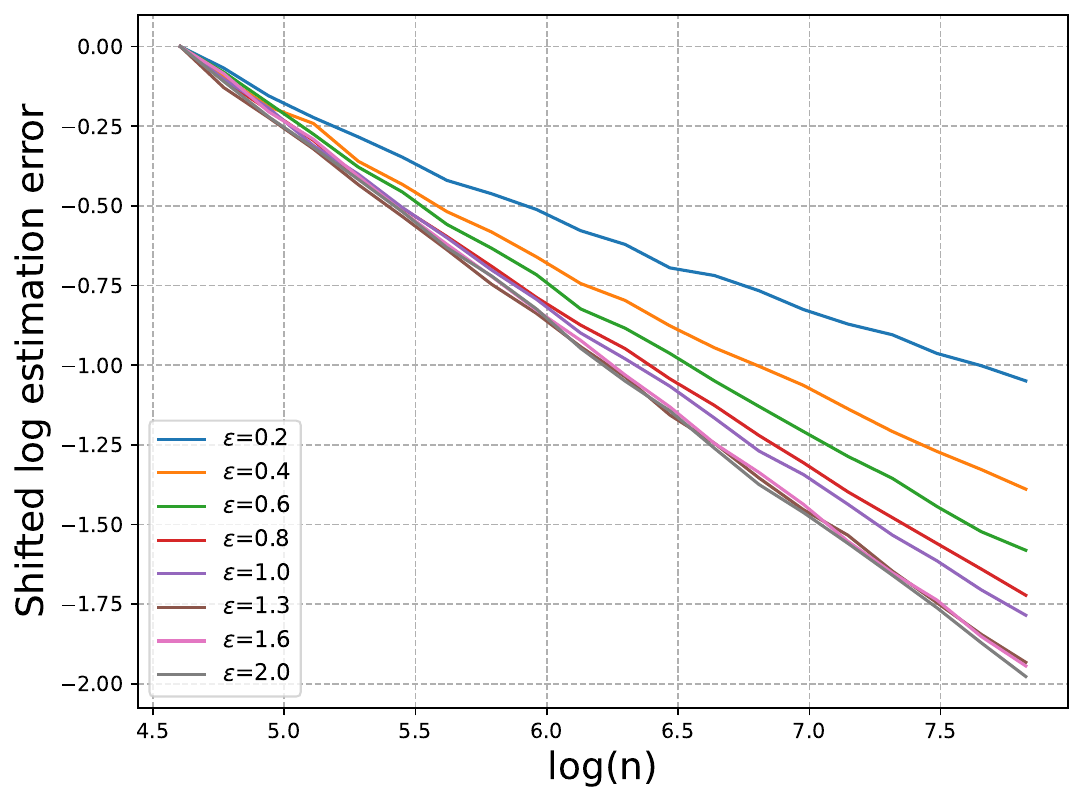}
        \subcaption{Bilinear Model}
        \label{subfig:exp1_bilinear}
    \end{subfigure}
    \caption{Log of the relative estimation error versus $\log(n)$ under varying tail indices $\epsilon$. Steeper slopes indicate faster convergence. The results are averages over 200 repetitions.}
    \label{fig:exp1}
\end{figure}

\subsection{Experiment 2: Scale-Invariance and Conditioning}\label{subsec:exp2}

This experiment demonstrates the critical advantage of the ``De-scaling $\to$ Truncation $\to$ Re-scaling'' mechanism. We compare SRGD against methods that lack scale-invariance.

We consider Matrix Trace Regression with $p=q=36$, rank $K=2$. The ground truth factors are constructed as $\bm L^\ast=(\bm U^\top, \bm 0)^\top \bm S^{1/2}$ and $\bm R^\ast=(\bm V^\top, \bm 0)^\top\bm S^{1/2}$, where $\bm U, \bm V$ are orthonormal and $\bm S=\mathrm{diag}(1, 1/\kappa)$. We vary the condition number $\kappa \in \{1, 3, 7\}$ to induce scaling imbalance.
Data are heavy-tailed: $\bm X_i \sim t_{2.5}$ and $e_i \sim 0.1 t_{1.5}$. We compare the following methods:
(1) SRGD--SHT (proposed);
(2) RGD--SHT (Regularized Gradient Descent): applies truncation directly to standard gradients (i.e., skipping the de-scaling step);
(3) Huber--GD: minimizes the Huber loss using standard gradients. 

Figure \ref{fig:exp2} illustrates the convergence trajectories.
When $\kappa=1$ (balanced scaling), SRGD and RGD perform comparably. However, as $\kappa$ increases, the performance of RGD and Huber-GD degrades drastically. This confirms the coupling failure mode described in Section \ref{sec:2.2}: without de-scaling, the optimization landscape interferes with the robust thresholding/loss, causing the algorithm to truncate valid signals along the large-norm directions or leak outliers along small-norm directions.
In contrast, SRGD exhibits \textit{condition-number independence}: its convergence profile remains virtually unchanged across all $\kappa$, validating Theorem \ref{thm:1}.

\begin{figure}[htbp]
    \centering
    \begin{subfigure}[b]{0.32\textwidth}
      \centering
      \includegraphics[width=\linewidth]{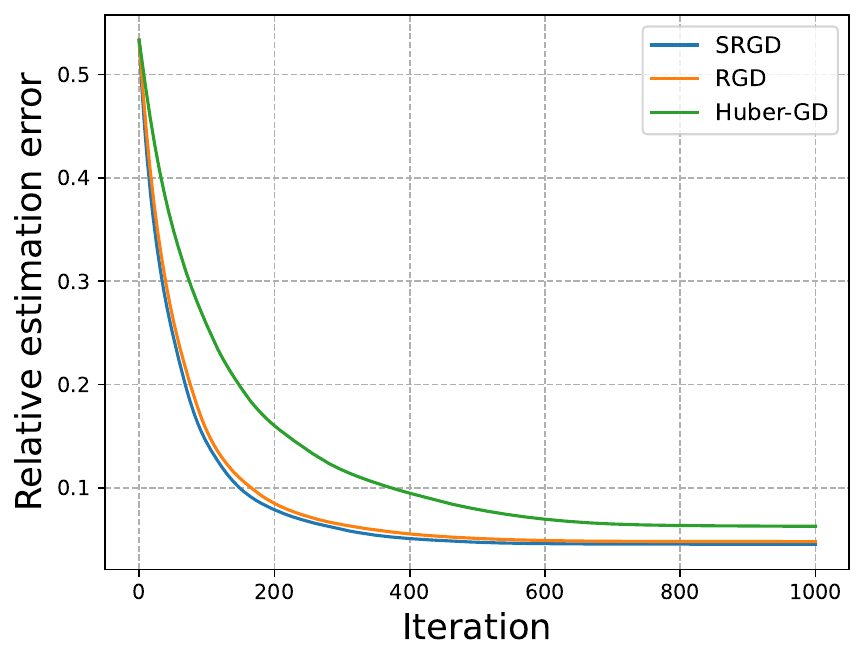}
      \subcaption{$\kappa=1$ (Well-conditioned)}
      \label{subfig:exp2_kap1}
    \end{subfigure}
    \hfill
    \begin{subfigure}[b]{0.32\textwidth}
        \centering
        \includegraphics[width=\linewidth]{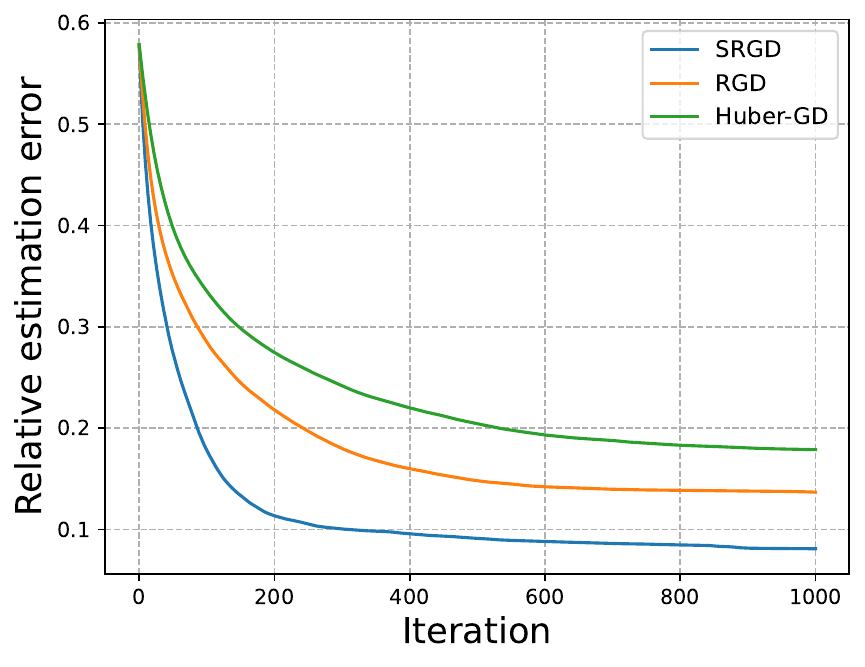}
        \subcaption{$\kappa=3$}
        \label{subfig:exp2_kap3}
    \end{subfigure}
    \hfill
    \begin{subfigure}[b]{0.32\textwidth}
      \centering
      \includegraphics[width=\linewidth]{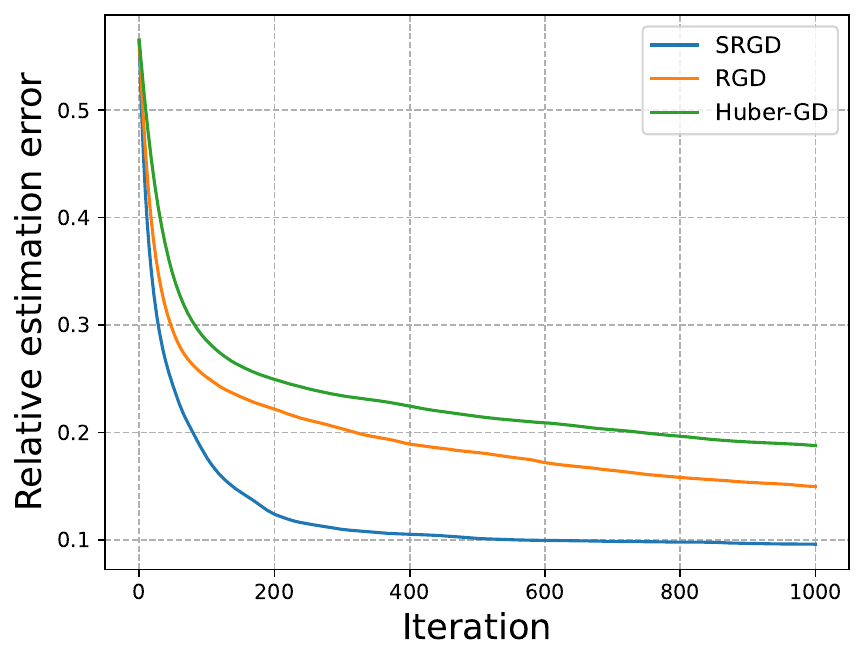}
      \subcaption{$\kappa=7$ (Ill-conditioned)}
      \label{subfig:exp2_kap7}
    \end{subfigure}
    \caption{Relative estimation error versus iterations under varying condition numbers $\kappa$.}
    \label{fig:exp2}
\end{figure}

\subsection{Experiment 3: Robustness against Heavy Tails}\label{sec:5.3}

Finally, we benchmark the statistical accuracy of the two-stage estimator against alternatives across different tail regimes.

We fix $\kappa=5$. We consider a $3 \times 2$ factorial design:
\begin{itemize}
    \item \textbf{Predictors:} Gaussian ($\bm X_i \sim N(0,1)$) vs. Heavy-tailed ($\bm X_i \sim t_{2.5}$).
    \item \textbf{Noise:} Gaussian ($e_i \sim 0.1 N(0,1)$) vs. Intermediate ($e_i \sim 0.1 t_{2.1}$) vs. Very Heavy ($e_i \sim 0.1 t_{1.5}$).
\end{itemize}
We compare: (1) SRGD (Proposed), (2) ScGD (Scaled GD without truncation, i.e., non-robust but scale-invariant), (3) RGD (Truncated but not scaled), and (4) Huber--GD. All methods use SHT for sparsity.

\begin{figure}[htbp]
  \centering
  \includegraphics[width=\textwidth]{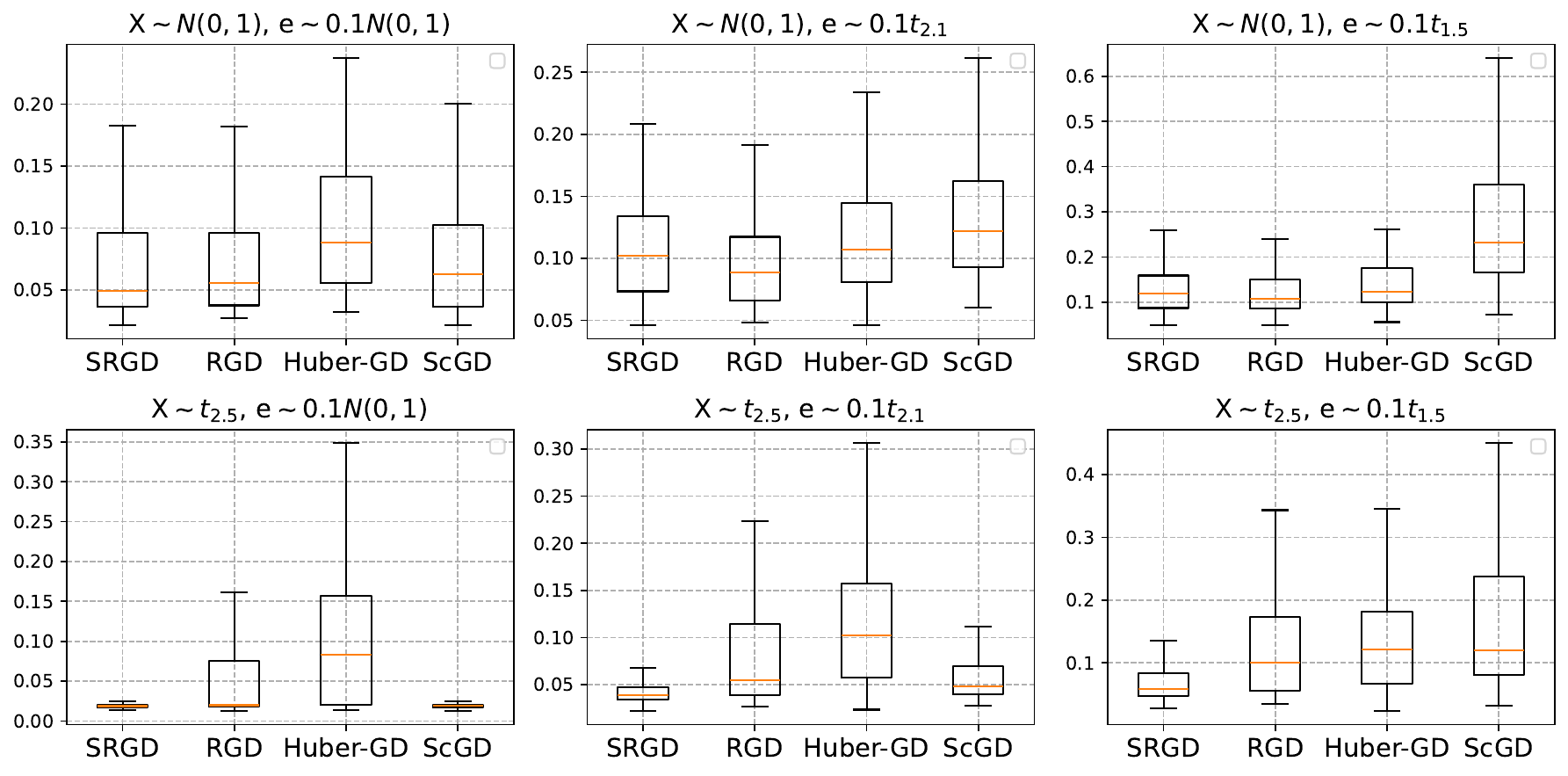}
  \caption{Boxplots of relative estimation errors (200 repetitions).}
  \label{fig:exp3_matreg}
\end{figure}

Figure \ref{fig:exp3_matreg} presents the error distributions. Key observations include:
\begin{itemize}
    \item \textbf{Benefit of Truncation (SRGD vs. ScGD):} As the noise tails get heavier (moving right), ScGD fails, while SRGD maintains low error. This confirms that scale-invariance alone (ScGD) is insufficient; explicit robustification is required.
    \item \textbf{Benefit of Scaling (SRGD vs. RGD):} In the bottom row (heavy-tailed predictors), SRGD significantly outperforms RGD. Heavy-tailed predictors induce large gradient fluctuations that exacerbate the scaling ambiguity; SRGD's de-scaling step effectively preconditions these updates.
    \item \textbf{Failure of Huber (Huber vs. SRGD):} Huber--GD performs poorly when predictors are heavy-tailed (bottom row). This is expected, as the Huber loss is robust to $Y$-outliers but not to leverage points ($\bm X$-outliers). SRGD handles both simultaneously.
\end{itemize}
Overall, SRGD achieves the lowest estimation error in the most adversarial setting (heavy-tailed $\bm{X}$ and $e$), confirming the synergy of de-scaling and statistical truncation.

\section{Real Data Examples}\label{sec:6}

In this section, we evaluate the proposed SRGD--SHT algorithm through two distinct real-world applications: the classification of electroencephalogram (EEG) data using matrix logistic regression, and the forecasting of a macroeconomic dataset via the bilinear model. These examples encompass both classification and regression tasks, demonstrating the versatility of the proposed procedure in handling high-dimensional matrix-variate data characterized by complex structures and heavy-tailed behaviors.

\subsection{EEG Data Classification}\label{sec:6.1}

We first apply our framework to the classification of EEG signals to identify alcoholism. The dataset, originally examined by \citet{hung2013matrix}, consists of recordings from 122 subjects divided into alcoholic ($y=1$) and control ($y=0$) groups. Subjects were exposed to visual stimuli under three conditions: single stimulus (S1), matching second stimulus (S2 match), and non-matching second stimulus (S2 nomatch). Signals were recorded via 64 electrodes sampled at 256 Hz. The data can be downloaded from the UCI Machine Learning Repository (\url{http://archive.ics.uci.edu/ml/datasets/EEG+Database}).

We restrict our analysis to the S1 and S2 nomatch conditions. To ensure cohort homogeneity, we analyze the dominant batch of $N=109$ subjects, excluding a minor batch of 13 subjects to avoid confounding batch effects and extreme class imbalance. For each subject, we average signals across trials to construct a single covariate matrix $\bm X_i \in \mathbb{R}^{64 \times 256}$. 

A critical step in Kronecker-structured estimation is ensuring that the matrix indices correspond to meaningful modes of variation. Standard EEG channel ordering is arbitrary. We incorporate biological priors by reordering the 64 electrodes to respect the brain's anatomical topology, following a strict anterior-to-posterior and left-to-right progression clustered into eight functional regions (e.g., Prefrontal, Frontal, Parietal; see Supplementary Materials for the full mapping). This permutation ensures that the row index of $\bm X_i$ acts as a coherent spatial coordinate. Consequently, the Kronecker factors can effectively capture local spatial correlations (anatomical clusters) distinct from temporal correlations, enhancing both estimation efficiency and physiological interpretability.

Neuroimaging data are notoriously prone to artifacts (e.g., eye blinks, muscle movement) that generate heavy-tailed distributions. Figure \ref{fig:EEG_kurt} confirms this: the sample excess kurtosis is consistently elevated across dimensions for both groups, with tails extending far beyond the Gaussian baseline (kurtosis $= 0$). These extreme values act as high-leverage points in the predictor space, which can severely skew the decision boundary of standard logistic regression models. This systematic deviation from Gaussianity validates the necessity of the proposed truncated gradient framework.

\begin{figure}[htp]
  \centering
  \includegraphics[width=0.8\textwidth]{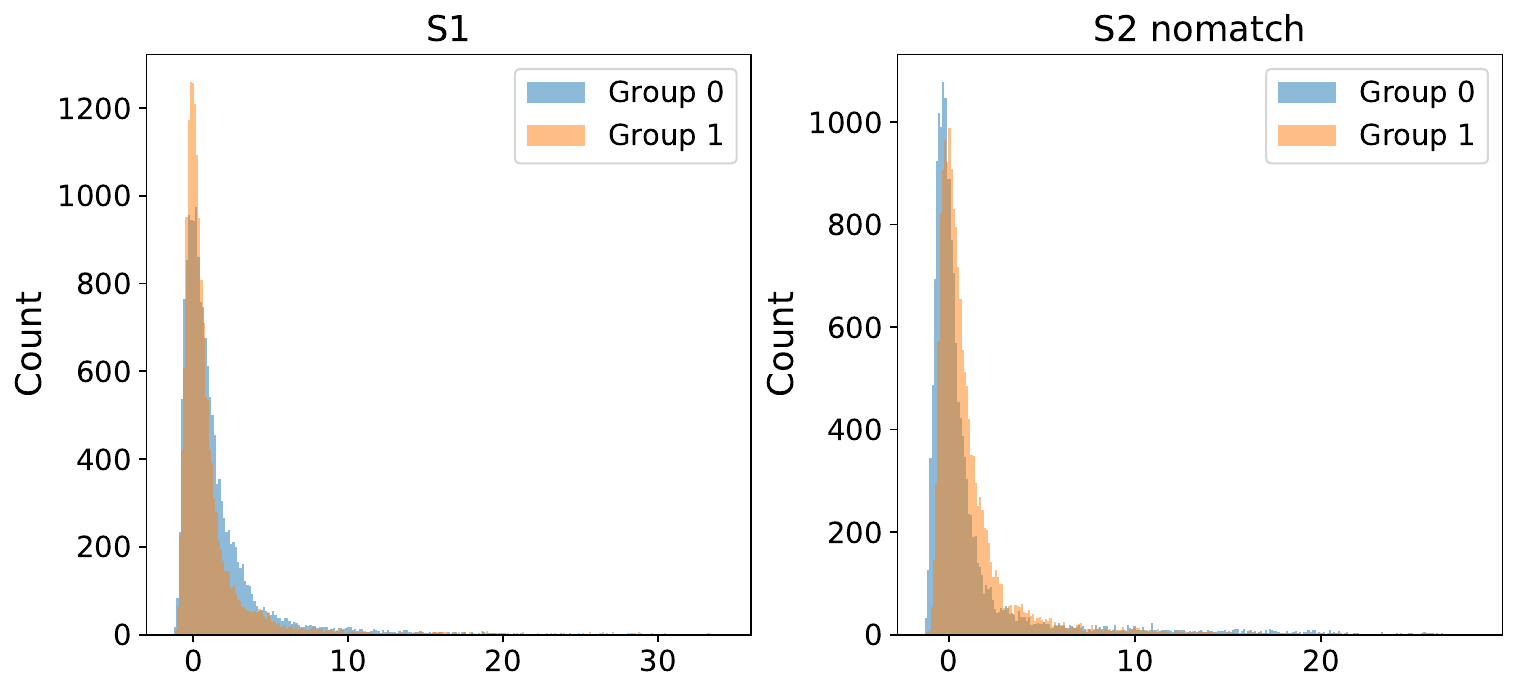}
  \caption{Evidence of heavy tails in EEG predictors.}
  \label{fig:EEG_kurt}
\end{figure}

We evaluate classification performance (AUC) on a held-out test set of 22 randomly selected subjects. Initialization is performed via the Robust Lasso (Section \ref{sec:4.2}). The truncation parameter is set to $\tau=0.25\tau_0$, where $\tau_0=\sqrt{n/\log{d}}$.
\begin{itemize}
    \item \textbf{S1 Task:} We utilize a factorization shape $(p_1,q_1,p_2,q_2)=(16,16,4,16)$ with rank $K=2$. This imposes a block-wise structure implying that the predictive signal is homogeneous within local anatomical clusters (4 electrodes) over short temporal windows ($\approx 62.5$ ms). Sparsity levels $s_L=50$ and $s_R=15$ are chosen to retain $\approx 20\%$ active parameters.
    \item \textbf{S2 Nomatch Task:} Discriminating non-matching stimuli involves higher-order cognitive conflict recruiting a broader network. We increase spatial resolution with shape $(32,16,2,16)$ and rank $K=3$, utilizing sparsity $s_L=100, s_R=7$.
\end{itemize}

Table \ref{tab:eeg_auc} benchmarks SRGD--SHT against vector-based Lasso, unstructured low-rank models, and ablation variants. Two key conclusions emerge. First, SRGD--SHT achieves the highest AUC in both tasks ($0.923$ and $0.915$), outperforming the unstructured Low-Rank baseline by a wide margin. This confirms that the Kronecker prior successfully captures the specific spatio-temporal coupling of brain activity. Second, the ablation study demonstrates the synergy of our algorithmic components. Removing scaling (RGD-SHT) or truncation (ScGD-SHT) degrades performance. Most notably in the S1 task, removing truncation causes a catastrophic drop in AUC from $0.923$ to $0.744$. This empirically proves that scale-invariance alone is insufficient; explicit robustification against heavy-tailed predictors is indispensable for reliable learning in this regime.

\begin{table}[htb]
\centering
\caption{Classification performance (AUC) on the held-out test set.}
\label{tab:eeg_auc}
\begin{tabular}{llcc}
\hline
\textbf{Method} & \textbf{Specification} & \textbf{S1 (AUC)} & \textbf{S2 nomatch (AUC)} \\ \hline
\textit{Baselines} & & & \\
LASSO & Vector Model & 0.718 & 0.795 \\
Low-Rank & Rank-1 & 0.803 & 0.632 \\
Low-Rank & Rank-3 & 0.803 & 0.778 \\ \hline
\textit{Ablation Variants} & & & \\
SRGD-Dense & No Sparsity & 0.735 & 0.863 \\
RGD-SHT & Unscaled & 0.821 & 0.906 \\
ScGD-SHT & No Truncation & 0.744 & 0.855 \\ \hline
\textit{Proposed} & & & \\
\textbf{SRGD--SHT} & \textbf{Full Model} & \textbf{0.923} & \textbf{0.915} \\ \hline
\end{tabular}
\end{table}

To visualize the impact of robustification, Figure \ref{fig:EEG_diff} displays the difference between the estimated coefficient matrices of the robust (SRGD--SHT) and non-robust (ScGD--SHT) methods for the S1 task. The heatmap reveals high-magnitude discrepancies (deep red/blue) localized to specific channels and time intervals. Since the robust model achieves significantly higher accuracy, these discrepancies indicate areas where the non-robust estimator was likely pulled by high-leverage artifacts. By truncating these outliers, SRGD recovers a more predictive signal support, leading to improved generalization.

\begin{figure}[htb]
  \centering
  \includegraphics[width=0.6\textwidth]{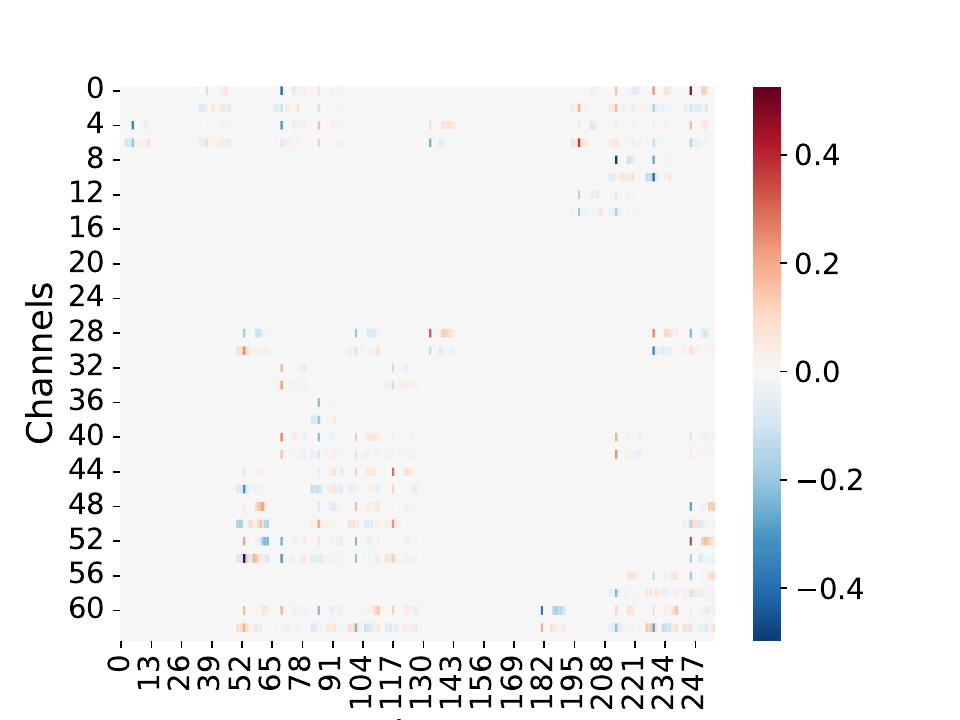}
  \caption{Heatmap of the difference $\widehat{\bbm\Theta}_\text{SRGD}-\widehat{\bbm\Theta}_\text{ScGD}$ for the S1 task.}
  \label{fig:EEG_diff}
\end{figure}

\subsection{Macroeconomic Data Forecasting}\label{sec:6.2}

To demonstrate the versatility of our framework in the context of bilinear models, we analyze a multinational macroeconomic dataset previously studied by \citet{chen2020constrained}. The data comprises 10 economic indicators (e.g., CPI, GDP, and Interest Rates) observed across 14 OECD countries (e.g., the USA, Germany, and France) over 107 quarters from 1990-Q2 to 2016-Q4. This dataset is notably challenging due to the presence of extreme market shocks (e.g., the 2008 financial crisis) which induce heavy-tailed distributions in both the predictors and the errors. At each time point $t$, the observation forms a matrix $\bm X_t \in \mathbb{R}^{14 \times 10}$. We apply the transformations described in \citet{chen2020constrained} (log-difference, seasonal adjustment, centering, and standardization) to induce stationarity. Consistent with the EEG analysis, the processed series exhibit significant excess kurtosis (see Appendix \ref{append:F} in Supplementary Materials), validating the need for robust estimation.

We model the temporal dynamics using a sparse matrix autoregressive model of order 1:
\begin{equation}
    \bm{X}_t = \bm{A}\bm{X}_{t-1}\bm{B}^\top + \bm{E}_t,
\end{equation}
which corresponds to the bilinear regression model in Section \ref{sec:4.3_bilinear} with predictors $\bm{X}_{t-1}$ and responses $\bm{X}_t$. We employ the proposed two-stage estimation procedure (Robust Dantzig initialization + SRGD--SHT refinement). The hyperparameters are set to $s_L=25$ and $s_R=15$, with truncation $\tau = 0.5\tau_0$. We conduct a rolling-window one-step-ahead forecast over the horizon 2013-Q1 to 2016-Q4 (16 quarters). We compare against three categories of baselines: (1) \textbf{RRMAR} \citep[Reduced-Rank MAR, non-robust bilinear,][]{Xiao2024RRMAR}; (2) \textbf{MFM} \citep[Matrix Factor Models,][]{wang2019factor}; and (3) \textbf{Huber-GD} (Robustness via loss function).

Table \ref{tab:TS_forecast} reports the Mean Squared Forecast Error (MSFE). The proposed SRGD--SHT achieves the lowest error (120.113), outperforming the standard RRMAR baselines. Notably, SRGD--SHT also outperforms the robust Huber-GD baseline (120.838). This reinforces the finding from the simulation study (Section \ref{subsec:exp2}) that gradient truncation is superior to Huber loss when the predictors ($\bm{X}_{t-1}$) themselves are heavy-tailed, a typical characteristic of economic time series during volatile periods. Furthermore, the significant performance drop in the non-sparse ablation variant (SRGD-Dense: 123.027) underscores that sparsity is critical for generalization in this high-dimensional regime ($p \times q = 140$ variables, $n \approx 90$ samples).

\begin{table}[htb]
\centering
\caption{One-step-ahead out-of-sample forecasting performance (MSFE). Parentheses $(r_1, r_2)$ denote the ranks of the transition matrices $\bm A$ and $\bm B$ (for RRMAR) or the number of row and column factors (for MFM).}
\label{tab:TS_forecast}
\begin{tabular}{lccccc}
\toprule
Method & Huber-GD & RRMAR(3,2) & RRMAR(6,4) & MFM(1,1) & MFM(3,2) \\ 
MSFE & 120.838 & 121.513 & 125.696 & 128.609 & 121.407 \\ \midrule
Method & \textbf{SRGD--SHT} & RGD--SHT & ScGD--SHT & SRGD--Dense & -- \\ 
MSFE & \textbf{120.113} & 120.358 & 120.709 & 123.027 & -- \\ 
\bottomrule
\end{tabular}
\end{table}

Beyond forecasting, a valid economic model must recover interpretable dynamics. In a multi-country autoregression, the transition matrix $\bm{A} \in \mathbb{R}^{14 \times 14}$ (capturing cross-country spillover) should be \textit{diagonally dominant}, reflecting the fact that a country's future state is most strongly predicted by its own immediate past.

Figure \ref{fig:TS_A} compares the estimated $\widehat{\bm{A}}$ matrices from our robust method (SRGD) and the non-robust ablation (ScGD). A striking discrepancy appears in the row corresponding to France (FRA). The non-robust estimator fails to identify the diagonal element $A_{\text{FRA, FRA}}$, estimating it as near-zero. This implies an economically implausible scenario where the French economy has no memory of its previous quarter. In contrast, the robust SRGD estimator correctly recovers a strong positive diagonal coefficient for France. This suggests that transient outliers in the French data likely skewed the least-squares gradient, causing the non-robust method to wash out the persistent signal. By truncating these artifacts, SRGD restores the expected structural logic of the economic system.

\begin{figure}[htb]
  \centering
  \includegraphics[width=0.8\textwidth]{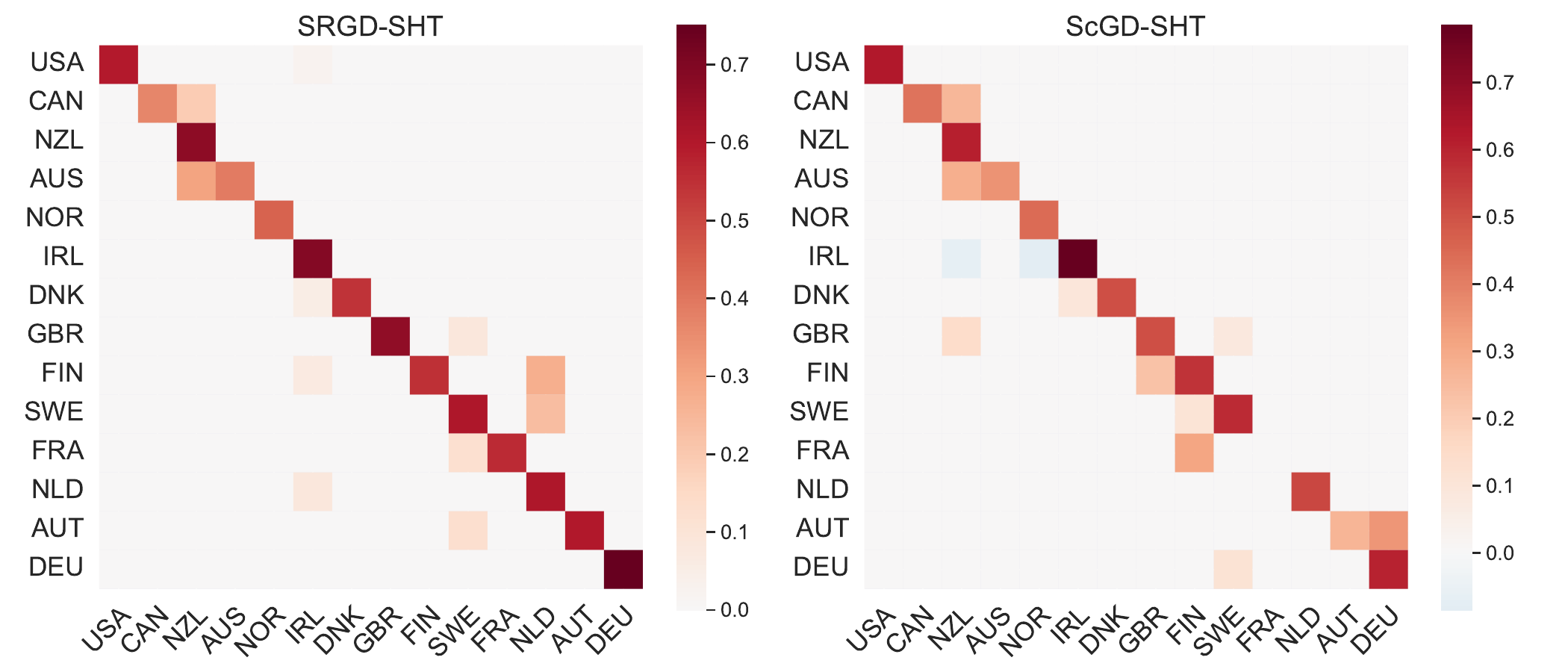}
  \caption{Estimates of the country transition matrix $\bm A$. The left panel displays the estimate from SRGD--SHT and the right panel corresponds to ScGD--SHT.}
  \label{fig:TS_A}
\end{figure}
\vspace{-1cm}

\section{Conclusion and Discussion}\label{sec:7}

In this article, we have introduced a unified framework for high-dimensional structured matrix estimation that resolves the fundamental tension between \textit{non-convex scaling ambiguity} and \textit{statistical robustness}. Our work demonstrates that standard robustification techniques, such as gradient truncation or Huber loss, are geometrically fragile in factorized models. When the factorization is ill-scaled, these fixed-threshold mechanisms fail to distinguish between heavy-tailed outliers and valid signals along large-norm directions.

Our proposed methodology, integrating Scaled Robust Gradient Descent (SRGD) and Scaled Hard Thresholding (SHT), addresses this failure mode by operating in a transformation-invariant domain. By de-scaling the gradient before truncation and performing variable selection on normalized factors, we ensure that the estimation process is stable regardless of the arbitrary parameterization of the Kronecker product.

From a theoretical perspective, our analysis establishes the convergence for a two-stage estimation procedure across three canonical models: matrix trace regression, matrix GLMs, and bilinear regression. A key insight from our results is the characterization of a smooth phase transition governed by the tail index $\epsilon$ of the noise. We prove that robustness does not incur an efficiency penalty in finite-variance regimes ($\epsilon=1$), while optimally degrading to $O(n^{-\epsilon/(1+\epsilon)})$ in heavier-tailed settings. This provides a rigorous quantification of the cost of robustness in high-dimensional structured estimation.

The principles developed here extend beyond the specific case of Kronecker products. The ``scale-then-robustify'' philosophy is broadly applicable to other non-convex factorizations plagued by scaling ambiguities, such as tensor CP/Tucker decompositions and deep neural networks. In these settings, gradient magnitudes can vary exponentially across layers or modes; our results suggest that de-scaled, element-wise truncation may offer a generic path toward stabilizing robust optimization in deep or multi-modal architectures. Future work will explore these extensions, as well as the adaptation of our two-stage framework to online and streaming data environments.

{
\setstretch{1.35}
\bibliography{mybib}

\begin{thebibliography}{}

\bibitem[Amini and Wainwright, 2008]{amini2008high}
Amini, A.~A. and Wainwright, M.~J. (2008).
\newblock High-dimensional analysis of semidefinite relaxations for sparse
  principal components.
\newblock In {\em 2008 IEEE International Symposium on Information Theory},
  pages 2454--2458. IEEE.

\bibitem[Bubeck, 2015]{bubeck2015convex}
Bubeck, S. (2015).
\newblock Convex optimization: Algorithms and complexity.
\newblock {\em Foundations and Trends{\textregistered} in Machine Learning},
  8(3-4):231--357.

\bibitem[Cai et~al., 2022]{cai2022kopa}
Cai, C., Chen, R., and Xiao, H. (2022).
\newblock Kopa: Automated kronecker product approximation.
\newblock {\em Journal of Machine Learning Research}, 23(236):1--44.

\bibitem[Cai et~al., 2023]{cai2023hybrid}
Cai, C., Chen, R., and Xiao, H. (2023).
\newblock Hybrid \uppercase{K}ronecker product decomposition and approximation.
\newblock {\em Journal of Computational and Graphical Statistics},
  32(3):838--852.

\bibitem[Cai et~al., 2013]{cai2013sparse}
Cai, T.~T., Ma, Z., and Wu, Y. (2013).
\newblock Sparse \uppercase{PCA}: Optimal rates and adaptive estimation.
\newblock {\em The Annals of Statistics}, 41(237):3074--3110.

\bibitem[Candes and Plan, 2011]{candes2011tight}
Candes, E.~J. and Plan, Y. (2011).
\newblock Tight oracle inequalities for low-rank matrix recovery from a minimal
  number of noisy random measurements.
\newblock {\em IEEE Transactions on Information Theory}, 57:2342--2359.

\bibitem[Chen et~al., 2020]{chen2020constrained}
Chen, E.~Y., Tsay, R.~S., and Chen, R. (2020).
\newblock Constrained factor models for high-dimensional matrix-variate time
  series.
\newblock {\em Journal of the American Statistical Association}.

\bibitem[Chen et~al., 2022]{chen2022fast}
Chen, K., Dong, R., Xu, W., and Zheng, Z. (2022).
\newblock Fast stagewise sparse factor regression.
\newblock {\em Journal of Machine Learning Research}, 23(271):1--45.

\bibitem[Chen et~al., 2021]{chen2021autoregressive}
Chen, R., Xiao, H., and Yang, D. (2021).
\newblock Autoregressive models for matrix-valued time series.
\newblock {\em Journal of Econometrics}, 222(1):539--560.

\bibitem[Chen and Cand{\`e}s, 2017]{chen2017solving}
Chen, Y. and Cand{\`e}s, E.~J. (2017).
\newblock Solving random quadratic systems of equations is nearly as easy as
  solving linear systems.
\newblock {\em Communications on pure and applied mathematics}, 70(5):822--883.

\bibitem[Cheng and Zhao, 2024]{cheng2024two}
Cheng, C. and Zhao, Z. (2024).
\newblock Two-way sparse reduced-rank regression via scaled gradient descent
  with hard thresholding.
\newblock In {\em 2024 IEEE 13rd Sensor Array and Multichannel Signal
  Processing Workshop (SAM)}, pages 1--5. IEEE.

\bibitem[Chi et~al., 2019]{chi2019median}
Chi, Y., Li, Y., Zhang, H., and Liang, Y. (2019).
\newblock Median-truncated gradient descent: A robust and scalable nonconvex
  approach for signal estimation.
\newblock In {\em Compressed Sensing and Its Applications: Third International
  MATHEON Conference 2017}, pages 237--261. Springer.

\bibitem[Fan et~al., 2021]{fan2021shrinkage}
Fan, J., Wang, W., and Zhu, Z. (2021).
\newblock A shrinkage principle for heavy-tailed data: High-dimensional robust
  low-rank matrix recovery.
\newblock {\em The Annals of Statistics}, 49(3):1239.

\bibitem[Fan et~al., 2025]{fan2025matrix}
Fan, Z., Zhang, X., Chen, M., and Wang, D. (2025).
\newblock Matrix time series modeling: A hybrid framework combining
  autoregression and common factors.
\newblock {\em arXiv preprint arXiv:2503.05340}.

\bibitem[Harville, 1997]{harvillematrix}
Harville, D.~A. (1997).
\newblock {\em Matrix Algebra From a Statistician's Perspective}.
\newblock Springer.

\bibitem[Hung and Wang, 2013]{hung2013matrix}
Hung, H. and Wang, C.-C. (2013).
\newblock Matrix variate logistic regression model with application to
  \uppercase{EEG} data.
\newblock {\em Biostatistics}, 14(1):189--202.

\bibitem[Li et~al., 2016]{li2016nonconvex}
Li, X., Arora, R., Liu, H., Haupt, J., and Zhao, T. (2016).
\newblock Nonconvex sparse learning via stochastic optimization with
  progressive variance reduction.
\newblock {\em arXiv preprint arXiv:1605.02711}.

\bibitem[Li et~al., 2020]{li2020non}
Li, Y., Chi, Y., Zhang, H., and Liang, Y. (2020).
\newblock Non-convex low-rank matrix recovery with arbitrary outliers via
  median-truncated gradient descent.
\newblock {\em Information and Inference: A Journal of the IMA}, 9(2):289--325.

\bibitem[Lu et~al., 2025]{lu2025robust}
Lu, Y., Tao, C., Wang, D., Uddin, G.~S., Wu, L., and Zhu, X. (2025).
\newblock Robust estimation for dynamic spatial autoregression models with
  nearly optimal rates.
\newblock {\em Journal of Econometrics}, 251:106065.

\bibitem[Ma et~al., 2018]{ma2018implicit}
Ma, C., Wang, K., Chi, Y., and Chen, Y. (2018).
\newblock Implicit regularization in nonconvex statistical estimation: Gradient
  descent converges linearly for phase retrieval and matrix completion.
\newblock In {\em International Conference on Machine Learning}, pages
  3345--3354. PMLR.

\bibitem[Negahban and Wainwright, 2011]{negahban2011estimation}
Negahban, S. and Wainwright, M.~J. (2011).
\newblock Estimation of (near) low-rank matrices with noise and
  high-dimensional scaling.
\newblock {\em Annals of Statistics}, 39:1069--1097.

\bibitem[Negahban and Wainwright, 2012]{negahban2012restricted}
Negahban, S. and Wainwright, M.~J. (2012).
\newblock Restricted strong convexity and weighted matrix completion: Optimal
  bounds with noise.
\newblock {\em Journal of Machine Learning Research}, 13:1665--1697.

\bibitem[Prasad et~al., 2020]{prasad2020robust}
Prasad, A., Suggala, A.~S., Balakrishnan, S., and Ravikumar, P. (2020).
\newblock Robust estimation via robust gradient estimation.
\newblock {\em Journal of the Royal Statistical Society Series B: Statistical
  Methodology}, 82(3):601--627.

\bibitem[Shen et~al., 2025]{shen2025computationally}
Shen, Y., Li, J., Cai, J.-F., and Xia, D. (2025).
\newblock Computationally efficient and statistically optimal robust
  high-dimensional linear regression.
\newblock {\em The Annals of Statistics}, 53(1):374--399.

\bibitem[Sun et~al., 2020]{sun2020adaptive}
Sun, Q., Zhou, W.-X., and Fan, J. (2020).
\newblock Adaptive \uppercase{H}uber regression.
\newblock {\em Journal of the American Statistical Association},
  115(529):254--265.

\bibitem[Tan et~al., 2023]{tan2023sparse}
Tan, K.~M., Sun, Q., and Witten, D. (2023).
\newblock Sparse reduced rank \uppercase{H}uber regression in high dimensions.
\newblock {\em Journal of the American Statistical Association},
  118(544):2383--2393.

\bibitem[Tong et~al., 2021a]{tong2021accelerating}
Tong, T., Ma, C., and Chi, Y. (2021a).
\newblock Accelerating ill-conditioned low-rank matrix estimation via scaled
  gradient descent.
\newblock {\em Journal of Machine Learning Research}, 22(150):1--63.

\bibitem[Tong et~al., 2021b]{tong2021low}
Tong, T., Ma, C., and Chi, Y. (2021b).
\newblock Low-rank matrix recovery with scaled subgradient methods: Fast and
  robust convergence without the condition number.
\newblock {\em IEEE Transactions on Signal Processing}, 69:2396--2409.

\bibitem[Tu et~al., 2016]{tu2016low}
Tu, S., Boczar, R., Simchowitz, M., Soltanolkotabi, M., and Recht, B. (2016).
\newblock Low-rank solutions of linear matrix equations via procrustes flow.
\newblock In {\em International Conference on Machine Learning}, pages
  964--973. PMLR.

\bibitem[Wang et~al., 2019]{wang2019factor}
Wang, D., Liu, X., and Chen, R. (2019).
\newblock Factor models for matrix-valued high-dimensional time series.
\newblock {\em Journal of Econometrics}, 208(1):231--248.

\bibitem[Wang and Tsay, 2023]{wang2023rate}
Wang, D. and Tsay, R.~S. (2023).
\newblock Rate-optimal robust estimation of high-dimensional vector
  autoregressive models.
\newblock {\em The Annals of Statistics}, 51(2):846--877.

\bibitem[Wang et~al., 2017]{wang2017unified}
Wang, L., Zhang, X., and Gu, Q. (2017).
\newblock A unified computational and statistical framework for nonconvex
  low-rank matrix estimation.
\newblock In {\em Artificial Intelligence and Statistics}, pages 981--990.
  PMLR.

\bibitem[Wu and Feng, 2023]{wu2023sparse}
Wu, S. and Feng, L. (2023).
\newblock Sparse \uppercase{K}ronecker product decomposition: a general
  framework of signal region detection in image regression.
\newblock {\em Journal of the Royal Statistical Society Series B: Statistical
  Methodology}, 85(3):783--809.

\bibitem[Xiao et~al., 2023]{Xiao2024RRMAR}
Xiao, H., Han, Y., Chen, R., and Liu, C. (2023).
\newblock Reduced-rank autoregressive models for matrix time series.
\newblock {\em Journal of Business \& Economic Statistics}.
\newblock To appear.

\bibitem[Yang et~al., 2014]{yang2014sparse}
Yang, D., Ma, Z., and Buja, A. (2014).
\newblock A sparse singular value decomposition method for high-dimensional
  data.
\newblock {\em Journal of Computational and Graphical Statistics},
  23(4):923--942.

\bibitem[Zhang et~al., 2024]{zhang2024robust}
Zhang, X., Wang, D., Li, G., and Sun, D. (2024).
\newblock Robust and optimal tensor estimation via robust gradient descent.
\newblock {\em arXiv preprint arXiv:2412.04773}.

\bibitem[Zhu and Zhou, 2021]{zhu2021taming}
Zhu, Z. and Zhou, W. (2021).
\newblock Taming heavy-tailed features by shrinkage.
\newblock In {\em International Conference on Artificial Intelligence and
  Statistics}, pages 3268--3276. PMLR.

\end{thebibliography}
}

\appendix

\section{Computational Convergence of SRGD--SHT}\label{append:A}

\subsection{Convergence of SRGD}\label{append:A1}

\begin{proof}[Proof of Theorem \ref{thm:1}]

The proof is developed by induction and consists of three steps. In the first step, we define necessary definitions and conditions used throughout the proof. In the second step, we develop the local convergence results for the proposed algorithm. Finally, in the last step, we prove that the conditions hold in an inductive manner.

First, we introduce the notations used throughout the proof. For simplicity, we denote the robust de-scaled gradient estimators $\bm{G}_L(\bm{L}_j,\bm{R}_j;\mathcal{D}_n)$ and $\bm{G}_R(\bm{L}_j,\bm{R}_j;\mathcal{D}_n)$ as $\bm{G}_{L}^{(j)}$ and $\bm{G}_{R}^{(j)}$, respectively. Similarly, we simplify the notations of the expectation of the partial gradients:
\begin{equation}
  \begin{split}
    & \mathbb{E}[\nabla_{\bm{L}}f(\bm{L},\bm{R};z_i)]|_{\bm{L}=\bm{L}_j,\bm{R}=\bm{R}_j} = \mathbb{E}[\nabla_{\bm{L}}f^{(j)}],\\
    & \mathbb{E}[\nabla_{\bm{R}}f(\bm{L},\bm{R};z_i)]|_{\bm{L}=\bm{L}_j,\bm{R}=\bm{R}_j} = \mathbb{E}[\nabla_{\bm{R}}f^{(j)}].
  \end{split}
\end{equation}

The proposed robust scaled gradient descent algorithm can be formulated as
\begin{equation}
  \begin{split} 
    \bm{L}_{j+1} & = \bm{L}_{j} - \eta\cdot \bm{G}_{L}^{(j)} (\bm{R}_j^{\top}\bm{R}_j)^{-1/2}\\
    & = \bm{L}_j - \eta\cdot\mathbb{E}[\nabla_{\bm{L}}f^{(j)}](\bm{R}_j^{\top}\bm{R}_j)^{-1} - \eta\cdot\bm{\Delta}_{L}^{(j)}(\bm{R}_j^{\top}\bm{R}_j)^{-1/2},
  \end{split}
\end{equation}
and
\begin{equation}
  \begin{split}
    \bm{R}_{j+1} & = \bm{R}_{j} - \eta\cdot \bm{G}_{R}^{(j)}(\bm{L}_j^{\top}\bm{L}_{j})^{-1/2}\\
    & = \bm{R}_{j} - \eta\cdot\mathbb{E}[\nabla_{\bm{R}}f^{(j)}](\bm{L}_j^{\top}\bm{L}_{j})^{-1} - \eta\cdot\bm{\Delta}_{R}^{(j)}(\bm{L}_j^{\top}\bm{L}_j)^{-1/2},
  \end{split}
\end{equation}
where $\bm{\Delta}_{L}^{(j)}:=\bm{G}_{L}^{(j)}-\mathbb{E}[\nabla_{\bm{L}}f^{(j)}](\bm{R}_j^{\top}\bm{R}_j)^{-1/2}$ and $\bm{\Delta}_{R}^{(j)}:=\bm{G}_{R}^{(j)}-\mathbb{E}[\nabla_{\bm{R}}f^{(j)}](\bm{L}_{j}^{\top}\bm{L}_j)^{-1/2}$ are the estimation errors of the de-scaled partial gradients. The stability of the robust de-scaled gradient estimators implies:
\begin{equation}
  \begin{split}
    \|\bm{\Delta}_{L}^{(j)}\|_\text{F}^2 & \leq \phi\|\bm{L}_{j}\bm{R}_j^{\top}-\bm{L}_*\bm{R}_*^{\top}\|_\text{F}^2 + \xi_{L}^2,\\
    \|\bm{\Delta}_{R}^{(j)}\|_\text{F}^2 & \leq \phi\|\bm{L}_{j}\bm{R}_j^{\top}-\bm{L}_*\bm{R}_*^\top\|_\text{F}^2 + \xi_{R}^2.
  \end{split}
\end{equation}

Define
\begin{equation}
   \bm{Q}_{j} = {\argmin}_{\bm{Q}\in\text{GL}(K) }\Big\{\|(\bm{L}_j\bm{Q}-\bm{L}_*)\bm{\Sigma}_*^{1/2}\|_\text{F}^2 + \|(\bm{R}_j\bm{Q}^{-\top}-\bm{R}_*)\bm{\Sigma}_*^{1/2}\|_\text{F}^2\Big\}
\end{equation} 
as the optimal transformation aligning $(\bm{L}_{j},\bm{R}_{j})$ and $(\bm{L}_*,\bm{R}_*)$. We assume
\begin{equation}\label{eq:condition1}
  d(\bm{F}_{j},\bm{F}_{*}) = \sqrt{\|(\bm{L}_j\bm{Q}_{j}-\bm{L}_*)\bm{\Sigma}_*^{1/2}\|_\text{F}^2+\|(\bm{R}_j\bm{Q}_j^{-\top}-\bm{R}_*)\bm{\Sigma}_*^{1/2}\|_\text{F}^2} \leq C\alpha^{1/2}\beta^{-1/2}\sigma_K,
\end{equation}
for some sufficiently small $C$ and for all $j=1,\dots,J$. In the last part of the proof, we will prove this condition holds for all iterations.

The upper bound for $d(\bm{F}_{j},\bm{F}_{*})$ in \eqref{eq:condition1} implies that
\begin{equation}
  \begin{split}
    & \max(\|(\bm{L}_j\bm{Q}_j-\bm{L}^*)\bm{\Sigma}_*^{-1/2}\|_\text{op},\|(\bm{R}_j\bm{Q}_j^{-\top}-\bm{R}^*)\bm{\Sigma}_*^{-1/2}\|_\text{op})\cdot\sigma_K\\
    \leq & \sqrt{\left(\|(\bm{L}_j\bm{Q}_j-\bm{L}_*)\bm{\Sigma}_*^{-1/2}\|_\text{F}^2+\|(\bm{R}_j\bm{Q}_j^{-\top}-\bm{R}_*)\bm{\Sigma}_*^{-1/2}\|_\text{F}^2\right)\sigma_K^2}\\
    \leq & \sqrt{\left(\|(\bm{L}_j\bm{Q}_j-\bm{L}^*)\bm{\Sigma}_*^{1/2}\|_\text{F}^2+\|(\bm{R}_j\bm{Q}_j^{-\top}-\bm{R}^*)\bm{\Sigma}_*^{1/2}\|_\text{F}^2\right)}\leq C\alpha^{1/2}\beta^{-1/2}\sigma_K.
  \end{split}
\end{equation}
Hence, $\max(\|(\bm{L}_j\bm{Q}_j-\bm{L}^*)\bm{\Sigma}_*^{-1/2}\|_\text{op},\|(\bm{R}_j\bm{Q}_j^{-\top}-\bm{R}^*)\bm{\Sigma}_*^{-1/2}\|_\text{op})\leq C\alpha^{1/2}\beta^{-1/2}:=B$.

To simplify notations, as we focus on the $j$-th step, denote $\bm{L}=\bm{L}_j\bm{Q}_j$ and $\bm{R}=\bm{R}_j\bm{Q}_j^{-\top}$. Then, $\bm{L}_j=\bm{L}\bm{Q}_j^{-1}$, $\bm{R}_j=\bm{R}\bm{Q}_j^\top$, $(\bm{L}_j^{\top}\bm{L}_{j})^{-1}=\bm{Q}_j(\bm{L}^\top\bm{L})^{-1}\bm{Q}_j^\top$, $(\bm{R}_j^{\top}\bm{R}_j)^{-1}=\bm{Q}_j^{-\top}(\bm{R}^\top\bm{R})^{-1}\bm{Q}_j^{-1}$, $\bm{L}_j(\bm{L}_j^{\top}\bm{L}_j)^{-1}=\bm{L}(\bm{L}^\top\bm{L})^{-1}\bm{Q}_j^\top$, and $\bm{R}_{j}(\bm{R}_j^{\top}\bm{R}_{j})^{-1}=\bm{R}(\bm{R}^\top\bm{R})^{-1}\bm{Q}_j^{-1}$.

By definition of $\bm{Q}_j$,
\begin{equation}
  \begin{split}
    & \|(\bm{L}_{j+1}\bm{Q}_{j+1}-\bm{L}_*)\bm{\Sigma}_*^{1/2}\|_\text{F}^2 + \|(\bm{R}_{j+1}\bm{Q}^{-\top}_{j+1}-\bm{R}_*)\bm{\Sigma}_*^{1/2}\|_\text{F}^2\\
    \leq & \|(\bm{L}_{j+1}\bm{Q}_{j}-\bm{L}_*)\bm{\Sigma}_*^{1/2}\|_\text{F}^2 + \|(\bm{R}_{j+1}\bm{Q}_{j}^{-\top}-\bm{R}_*)\bm{\Sigma}_*^{1/2}\|_\text{F}^2.
  \end{split}
\end{equation}
By the mean inequality, for any $\zeta>0$,
\begin{equation}
  \begin{split}
    & \|(\bm{L}_{j+1}\bm{Q}_{j}-\bm{L}_*)\bm{\Sigma}_*^{1/2}\|_\text{F}^2 \\
    = & \|(\bm{L} - \eta\cdot\mathbb{E}[\nabla f^{(j)}]\bm{R}_{j}(\bm{R}_j^{\top}\bm{R}_j)^{-1}\bm{Q}_{j} - \eta\bm{\Delta}_{L}^{(j)}(\bm{R}_j^{\top}\bm{R}_j)^{-1/2}\bm{Q}_{j}-\bm{L}_*)\bm{\Sigma}_*^{1/2}\|_\text{F}^2\\
    \leq & (1+\zeta)\|(\bm{L}-\bm{L}_*-\eta\mathbb{E}[\nabla f^{(j)}]\bm{R}_j(\bm{R}_j^{\top}\bm{R}_{j})^{-1}\bm{Q}_{j})\bm{\Sigma}_*^{1/2}\|_\text{F}^2\\
    & + (1+\zeta^{-1})\eta^2\|\bm{\Delta}_{L}^{(j)}(\bm{R}_j^{\top}\bm{R}_j)^{-1/2}\bm{Q}_{j}\bm{\Sigma}_*^{1/2}\|_\text{F}^2\\
    = & (1+\zeta)\|(\bm{L}-\bm{L}_*-\eta\mathbb{E}[\nabla f^{(j)}]\bm{R}(\bm{R}^\top\bm{R})^{-1})\bm{\Sigma}_*^{1/2}\|_\text{F}^2\\
    & + (1+\zeta^{-1})\eta^2\|\bm{\Delta}_{L}^{(j)}(\bm{R}_j^{\top}\bm{R}_j)^{-1/2}\bm{Q}_{j}\bm{\Sigma}_*^{1/2}\|_\text{F}^2.
  \end{split}
\end{equation}
For the first term in the right hand side, we have the following upper bound
\begin{equation}\label{eq:bound1}
  \begin{split}
    & \|(\bm{L}-\bm{L}_*-\eta\mathbb{E}[\nabla f^{(j)}]\bm{R}(\bm{R}^\top\bm{R})^{-1})\bm{\Sigma}_*^{1/2}\|_\text{F}^2\\
    = & \|(\bm{L}-\bm{L}_*)\bm{\Sigma}_*^{1/2}\|_\text{F}^2 + \eta^2\|\mathbb{E}[\nabla f^{(j)}]\bm{R}(\bm{R}^\top\bm{R})^{-1}\bm{\Sigma}_*^{1/2}\|_\text{F}^2 \\
    & - 2\eta\left\langle(\bm{L}-\bm{L}_*)\bm{\Sigma}_*^{1/2},\mathbb{E}[\nabla f^{(j)}]\bm{R}(\bm{R}^\top\bm{R})^{-1}\bm{\Sigma}_*^{1/2}\right\rangle\\
    \leq & \|(\bm{L}-\bm{L}_*)\bm{\Sigma}_*^{1/2}\|_\text{F}^2 + \eta^2\|\bm{R}(\bm{R}^\top\bm{R})^{-1}\bm{\Sigma}_*^{1/2}\|_\text{op}^2\cdot\|\mathbb{E}[\nabla f^{(j)}]\|_\text{F}^2\\
    &-2\eta\left\langle(\bm{L}-\bm{L}_*)\bm{R}_*^{\top},\mathbb{E}[\nabla f^{(j)}]\right\rangle\\
    &  + 2\eta\left\langle(\bm{L}-\bm{L}_*)\bm{\Sigma}_*^{1/2},\mathbb{E}[\nabla f^{(j)}](\bm{V}_*-\bm{R}(\bm{R}^\top\bm{R})^{-1}\bm{\Sigma}_*^{1/2})\right\rangle.
  \end{split}
\end{equation}
Similarly, for $\bm{R}_{j}$, we have
\begin{equation}
  \begin{split}
    & \|(\bm{R}_{j+1}\bm{Q}_{j}^{-\top}-\bm{R}_*)\bm{\Sigma}_*^{1/2}\|_\text{F}^2\\
    \leq & (1+\zeta)\|(\bm{R}-\bm{R}_*-\eta\mathbb{E}[\nabla f^{(j)}]\bm{L}_{j}(\bm{L}_j^{\top}\bm{L}_{j})^{-1}\bm{Q}_{j}^{-\top})\bm{\Sigma}_*^{1/2}\|_\text{F}^2\\
    & + (1+\zeta^{-1})\eta^2\|\bm{\Delta}_{R}^{(j)}\bm{Q}_{j}^{-\top}\bm{\Sigma}_*^{1/2}\|_\text{F}^2\\
    = & (1+\zeta)\|(\bm{R}-\bm{R}_*-\eta\mathbb{E}[\nabla f^{(j)}]\bm{L}(\bm{L}^\top\bm{L})^{-1})\bm{\Sigma}_*^{1/2}\|_\text{F}^2 + (1+\zeta^{-1})\eta^2\|\bm{\Delta}_{R}^{(j)}\bm{Q}_{j}^{-\top}\bm{\Sigma}_*^{1/2}\|_\text{F}^2,
  \end{split}
\end{equation}
where the first term can be further bounded by
\begin{equation}\label{eq:bound2}
  \begin{split}
    & \|(\bm{R}-\bm{R}_*-\eta\mathbb{E}[\nabla f^{(j)}]\bm{L}(\bm{L}^\top\bm{L})^{-1})\bm{\Sigma}_*^{1/2}\|_\text{F}^2\\
    \leq & \|(\bm{R}-\bm{R}_*)\bm{\Sigma}_*^{1/2}\|_\text{F}^2 + \eta^2\|\bm{L}(\bm{L}^\top\bm{L})^{-1}\bm{\Sigma}_*^{1/2}\|_\text{op}^2\cdot\|\mathbb{E}[\nabla f^{(j)}]\|_\text{F}^2\\
    & -2\eta\left\langle(\bm{R}-\bm{R}_*)\bm{L}_*^{\top},\mathbb{E}[\nabla f^{(j)}]^\top\right\rangle\\
    & + 2\eta\left\langle(\bm{R}-\bm{R}_*)\bm{\Sigma}_*^{1/2},\mathbb{E}[\nabla f^{(j)}]^\top(\bm{U}_*-\bm{L}(\bm{L}^\top\bm{L})^{-1}\bm{\Sigma}_*^{1/2})\right\rangle.
  \end{split}
\end{equation}

Combining these two upper bounds in \eqref{eq:bound1} and \eqref{eq:bound2}, we have
\begin{equation}
  \begin{split}
    & \|(\bm{L}-\bm{L}_*-\eta\mathbb{E}[\nabla f^{(j)}]\bm{R}(\bm{R}^\top\bm{R})^{-1})\bm{\Sigma}_*^{1/2}\|_\text{F}^2 + \|(\bm{R}-\bm{R}_*-\eta\mathbb{E}[\nabla f^{(j)}]\bm{L}(\bm{L}^\top\bm{L})^{-1})\bm{\Sigma}_*^{1/2}\|_\text{F}^2 \\
    \leq & \|(\bm{L}-\bm{L}_*)\bm{\Sigma}_*^{1/2}\|_\text{F}^2 + \|(\bm{R}-\bm{R}_*)\bm{\Sigma}_*^{1/2}\|_\text{F}^2\\
    & + \eta^2\left(\|\bm{L}(\bm{L}^\top\bm{L})^{-1}\bm{\Sigma}_*^{1/2}\|_\text{op}^2+\|\bm{R}(\bm{R}^\top\bm{R})^{-1}\bm{\Sigma}_*^{1/2}\|_\text{op}^2\right)\cdot\|\mathbb{E}[\nabla f^{(j)}]\|_\text{F}^2\\
    & - 2\eta \left\langle\mathbb{E}[\nabla f^{(j)}],(\bm{L}-\bm{L}_*)\bm{R}_*^{\top}+\bm{L}_*(\bm{R}-\bm{R}_*)^\top\right\rangle\\
    & + 2\eta \left\langle\mathbb{E}[\nabla f^{(j)}],(\bm{L}-\bm{L}_*)\bm{\Sigma}_*^{1/2}(\bm{V}_*-\bm{R}(\bm{R}^\top\bm{R})^{-1}\bm{\Sigma}_*^{1/2})^\top\right.\\
    & \left.\quad\quad\quad+(\bm{U}^*-\bm{L}(\bm{L}^\top\bm{L})^{-1}\bm{\Sigma}_*^{1/2})\bm{\Sigma}_*^{1/2}(\bm{R}-\bm{R}_*)^\top\right\rangle\\
    =: & \|(\bm{L}-\bm{L}_*)\bm{\Sigma}_*^{1/2}\|_\text{F}^2 + \|(\bm{R}-\bm{R}_*)\bm{\Sigma}_*^{1/2}\|_\text{F}^2 + \eta^2M_1 - 2\eta M_2 + 2\eta M_3,
  \end{split}
\end{equation}
where $M_1$, $M_2$, and $M_3$ are defined as above.

By Lemma \ref{lemma:1} and the bound 
\begin{equation}
  \max(\|(\bm{L}_{j}\bm{Q}_j-\bm{L}_*)\bm{\Sigma}_*^{-1/2}\|_\text{op},\|(\bm{R}_{j}\bm{Q}_j^{-\top}-\bm{R}_*)\bm{\Sigma}_*^{-1/2}\|_\text{op})\leq B,
\end{equation}
we have the upper bound for $M_1$ given as
$$M_1 \leq \frac{2}{(1-B)^2}\|\mathbb{E}[\nabla\mathcal{L}^{(j)}]\|_\text{F}^2.$$
By the permutation operator $\mathcal{P}$, it is clear that $\|\nabla f^{(j)}\|_\text{F}=\|\nabla \mathcal{L}^{(j)}\|_\text{F}$. In addition, by the RCG condition, Lemma \ref{lemma:2} and mean inequality,
\begin{equation}
  \begin{split}
    & M_2 = \langle\mathbb{E}[\nabla f^{(j)}],\bm{L}\bm{R}^\top-\bm{L}_*\bm{R}_*^{\top}-(\bm{L}-\bm{L}_*)(\bm{R}-\bm{R}_*)^\top\rangle\\
    = & \langle\mathbb{E}[\nabla f^{(j)}],\bm{L}\bm{R}^\top-\bm{L}_*\bm{R}_*^{\top}\rangle - \langle\mathbb{E}[\nabla f^{(j)}],(\bm{L}-\bm{L}_*)(\bm{R}-\bm{R}_*)^\top\rangle\\
    \geq & \frac{\alpha}{2}\|\bm{L}\bm{R}^\top-\bm{L}_*\bm{R}_*^{\top}\|_\text{F}^2+\frac{1}{2\beta}\|\mathbb{E}[\nabla f^{(j)}]\|_\text{F}^2 - \frac{1}{4\beta}\|\mathbb{E}[\nabla f^{(j)}]\|_\text{F}^2 - \beta\|(\bm{L}-\bm{L}_*)(\bm{R}-\bm{R}_*)^\top\|_\text{F}^2\\
    = & \frac{\alpha}{2}\|\bm{L}\bm{R}^\top-\bm{L}_*\bm{R}_*^{\top}\|_\text{F}^2+\frac{1}{4\beta}\|\mathbb{E}[\nabla f^{(j)}]\|_\text{F}^2 - \beta\|(\bm{L}-\bm{L}_*)(\bm{R}-\bm{R}_*)^\top\|_\text{F}^2,
  \end{split}
\end{equation}
where the last term can be bounded by
\begin{equation}
  \begin{split}
    & \|(\bm{L}-\bm{L}_*)(\bm{R}-\bm{R}_*)^\top\|_\text{F}^2\\
    \leq & \frac{1}{2}\|(\bm{L}-\bm{L}_*)\bm{\Sigma}_*^{-1/2}\|_\text{op}^2\cdot\|(\bm{R}-\bm{R}_*)\bm{\Sigma}_*^{1/2}\|_\text{F}^2 + \frac{1}{2}\|(\bm{L}-\bm{L}_*)\bm{\Sigma}_*^{1/2}\|_\text{F}^2\cdot\|(\bm{R}-\bm{R}_*)\bm{\Sigma}_*^{-1/2}\|_\text{op}^2\\
    \leq & \frac{B^2}{2}\left(\|(\bm{L}-\bm{L}_*)\bm{\Sigma}_*^{1/2}\|_\text{F}^2+\|(\bm{R}-\bm{R}_*)\bm{\Sigma}_*^{1/2}\|_\text{F}^2\right).
  \end{split}
\end{equation}

Also, by the mean inequality and Lemma \ref{lemma:1}, for any $\gamma>0$,
\begin{equation}
  \begin{split}
    M_3 \leq & \|\mathbb{E}[\nabla f^{(j)}]\|_\text{F}\cdot\|(\bm{L}-\bm{L}_*)\bm{\Sigma}_*^{1/2}\|_\text{F}\cdot\|\bm{V}_*-\bm{R}(\bm{R}^\top\bm{R})^{-1}\bm{\Sigma}_*^{1/2}\|_\text{op}\\
    & + \|\mathbb{E}[\nabla f^{(j)}]\|_\text{F}\cdot\|(\bm{R}-\bm{R}_*)\bm{\Sigma}_*^{1/2}\|_\text{F}\cdot\|\bm{U}_*-\bm{L}(\bm{L}^\top\bm{L})^{-1}\bm{\Sigma}_*^{1/2}\|_\text{op}\\
    \leq & \frac{\sqrt{2}B\gamma^{-1}}{1-B}\|\mathbb{E}[\nabla f^{(j)}]\|_\text{F}^2+\frac{\sqrt{2}\gamma}{4}\left(\|(\bm{L}-\bm{L}_*)\bm{\Sigma}_*^{1/2}\|_\text{F}^2+\|(\bm{R}-\bm{R}_*)\bm{\Sigma}_*^{1/2}\|_\text{F}^2\right).
  \end{split}
\end{equation}

Therefore, combining the results above, by Lemma \ref{lemma:2}, we have
\begin{equation}
  \begin{split}
    & \|(\bm{L}-\bm{L}_*-\eta\mathbb{E}[\nabla f^{(j)}]\bm{R}(\bm{R}^\top\bm{R})^{-1})\bm{\Sigma}_*^{1/2}\|_\text{F}^2 + \|(\bm{R}-\bm{R}_*-\eta\mathbb{E}[\nabla f^{(j)}]\bm{L}(\bm{L}^\top\bm{L})^{-1})\bm{\Sigma}_*^{1/2}\|_\text{F}^2 \\
    \leq & \left(1+B^2\beta\eta+\frac{\sqrt{2}\gamma B\eta}{2}\right)\left(\|(\bm{L}-\bm{L}_*)\bm{\Sigma}_*^{1/2}\|_\text{F}^2+\|(\bm{R}-\bm{R}_*)\bm{\Sigma}_*^{1/2}\|_\text{F}^2\right)\\
    & - \alpha\eta\|\bm{L}\bm{R}^\top-\bm{L}_*\bm{R}_*^{\top}\|_\text{F}^2 + \left(\frac{\eta^2}{(1-B)^2}+\frac{2\sqrt{2}B\gamma^{-1}\eta}{1-B}-\frac{\eta}{2\beta}\right)\|\mathbb{E}[\nabla f^{(j)}]\|_\text{F}^2\\
    \leq & \left(1+B^2\beta\eta+\frac{\sqrt{2}\gamma B\eta}{2} - \frac{\alpha\eta}{2}\right)(\|(\bm{L}-\bm{L}_*)\bm{\Sigma}_*^{1/2}\|_\text{F}^2+\|(\bm{R}-\bm{R}_*)\bm{\Sigma}_*^{1/2}\|_\text{F}^2)\\
    & + \left(\frac{\eta^2}{(1-B)^2}+2\sqrt{2}B\gamma^{-1}\eta-\frac{\eta}{2\beta}\right)\|\mathbb{E}[\nabla f^{(j)}]\|_\text{F}^2.
  \end{split}
\end{equation}
Letting $\gamma= C_1\sqrt{\alpha\beta}$, $\eta=C_2\beta^{-1}$, and $\zeta=C_3\alpha\beta^{-1}$, if $B \leq C_4\sqrt{\alpha/\beta}$, for some universal constants $C_1,\dots,C_4>0$, we have 
\begin{equation}
  \begin{split} 
    &\|(\bm{L}_{j+1}\bm{Q}_j-\bm{L}_*)\bm{\Sigma}_*^{1/2}\|_\text{F}^2 + \|(\bm{R}_{j+1}\bm{Q}_j-\bm{R}_*)\bm{\Sigma}_*^{1/2}\|_\text{F}^2\\
    \leq & (1-C\alpha\beta^{-1})(\|(\bm{L}_{j}\bm{Q}_j-\bm{L}_*)\bm{\Sigma}_*^{1/2}\|_\text{F}^2+\|(\bm{R}_{j}\bm{Q}_j-\bm{R}_*)\bm{\Sigma}_*^{1/2}\|_\text{F}^2)\\
    & + C\alpha^{-1}\beta^{-1}\|\bm{\Delta}_{L}^{(j)}(\bm{R}_j^{\top}\bm{R}_{j})^{-1/2}\bm{Q}_j\bm{\Sigma}_*^{1/2}\|_\text{F}^2 + C\alpha^{-1}\beta^{-1}\|\bm{\Delta}_{R}^{(j)}(\bm{L}_j^{\top}\bm{L}_j)^{-1/2}\bm{Q}_j^{-\top}\bm{\Sigma}_*^{1/2}\|_\text{F}^2.
  \end{split}
\end{equation}
Here, $\|\bm{\Delta}_{L}^{(j)}(\bm{R}_j^{\top}\bm{R}_j)^{-1/2}\bm{Q}_j\bm{\Sigma}_*^{1/2}\|_\text{F}^2\leq\|\bm{\Delta}_{L}^{(j)}\|_\text{F}^2\cdot\|(\bm{R}_j^{\top}\bm{R}_j)^{-1/2}\bm{Q}_j\bm{\Sigma}_*^{1/2}\|_\text{op}^2$, and
\begin{equation}
  \begin{split}
    & \|(\bm{R}_j^{\top}\bm{R}_j)^{-1/2}\bm{Q}_j\bm{\Sigma}_*^{1/2}\|_\text{op}^2 =\lambda_{\max}(\bm{\Sigma}_*^{1/2}\bm{Q}_j^\top(\bm{R}_j^{\top}\bm{R}_j)^{-1}\bm{Q}_j\bm{\Sigma}_*^{1/2})\\
    = & \lambda_{\max}(\bm{\Sigma}_*^{1/2}(\bm{R}^\top\bm{R})^{-1}\bm{\Sigma}_*^{1/2}) = \|\bm{R}(\bm{R}^\top\bm{R})^{-1}\bm{\Sigma}_*^{1/2}\|_\text{op}^2 \leq \frac{1}{(1-B)^2} \leq C.
  \end{split}
\end{equation}
Similarly, we also have $\|(\bm{L}_j^{\top}\bm{L}_j)^{-1/2}\bm{Q}_j^{-\top}\bm{\Sigma}_*^{1/2}\|_\text{op}^2\leq C$.

According to the stability of $\bm{\Delta}_{L}^{(j)}$ and $\bm{\Delta}_{R}^{(j)}$, we have
\begin{equation}
  \|\bm{\Delta}_{L}^{(j)}(\bm{R}_j^{\top}\bm{R}_j)^{-1/2}\bm{Q}_j\bm{\Sigma}_*^{1/2}\|_\text{F}^2 \lesssim \phi\|\bm{L}_j\bm{R}_j^{\top}-\bm{L}_*\bm{R}_*^{\top}\|_\text{F}^2 + \xi_{L}^2
\end{equation}
and
\begin{equation}
  \|\bm{\Delta}_{R}^{(j)}(\bm{L}_j^{\top}\bm{L}_j)^{-1/2}\bm{Q}_j^{-\top}\bm{\Sigma}_*^{1/2}\|_\text{F}^2 \lesssim \phi\|\bm{L}_{j}\bm{R}_j^{\top}-\bm{L}_*\bm{R}_*^{\top}\|_\text{F}^2 + \xi_{R}^2.
\end{equation}
Therefore, with $\phi\lesssim\alpha^2$, we have that
\begin{equation}
  \begin{split} 
    d(\bm{F}_{j+1},\bm{F}_*)^2 \leq &~ (1-C\alpha\beta^{-1})d(\bm{F}_{j},\bm{F}_*)^2 + C\alpha^{-1}\beta^{-1}(\xi_{L}^2+\xi_{R}^2)\\
    \leq &~ \cdots\\
    \leq &~ (1-C\alpha\beta^{-1})^t d(\bm{F}_{0},\bm{F}_*)^2 + C\alpha^{-2}(\xi_{L}^2+\xi_{R}^2).
  \end{split}
\end{equation}
Finally, since $\xi_{L}^2 + \xi_{R}^2 \lesssim \alpha^3\beta^{-1}\sigma_K^2$, we have
\begin{equation}
  \begin{split}
    d(\bm{F}_{j+1},\bm{F}_*)^2 &\leq (1-C\alpha\beta^{-1})C\alpha\beta^{-1}\sigma_K^2 + C\alpha^{-1}\beta^{-1}(\xi_{L}^2+\xi_{R}^2)\\
    & \leq C\alpha\beta^{-1}\sigma_K^2+C\alpha^{-1}\beta^{-1}(\xi^2_L+\xi^2_R-C\alpha^3\beta^{-1}\sigma_K^2)\\
    & \leq C\alpha\beta^{-1}\sigma_K^2.
  \end{split}
\end{equation}
Then, the whole proof can be completed inductively. In addition, it naturally implies that $\bm{L}_{j+1}$ and $\bm{R}_{j+1}$ are not singular.
\end{proof}

\subsection{Convergence of SRGD--SHT}
\begin{proof}[Proof of Theorem \ref{thm:2}]

The proof is developed by induction. The proposed projected robust gradient descent algorithm is formulated 
as
\begin{equation}
  \begin{split}
    \widetilde{\bm{L}}_{j+1} & = \bm{L}_{j} - \eta\cdot \bm{G}_{L}^{(j)} (\bm{R}_j^{\top}\bm{R}_{j})^{-1/2},\\
    \widetilde{\bm{R}}_{j+1} & = \bm{R}_{j} - \eta\cdot \bm{G}_{R}^{(j)}(\bm{L}_j^{\top}\bm{L}_j)^{-1/2},\\
    \bm{L}_{j+1} & = \text{HT}(\widetilde{\bm{L}}_{j+1}(\widetilde{\bm{R}}_{j+1}^{\top}\widetilde{\bm{R}}_{j+1})^{1/2},s_L)(\widetilde{\bm{R}}_{j+1}^{\top}\widetilde{\bm{R}}_{j+1})^{-1/2},\\
    \bm{R}_{j+1} & = \text{HT}(\widetilde{\bm{R}}_{j+1}(\widetilde{\bm{L}}_{j+1}^{\top}\widetilde{\bm{L}}_{j+1})^{1/2},s_R)(\widetilde{\bm{L}}_{j+1}^{\top}\widetilde{\bm{L}}_{j+1})^{-1/2}.
  \end{split}
\end{equation}
Denote the nonzero active sets of $\bm{L}_j$, $\bm{R}_j$, $\bm{L}_{*}$, and $\bm{R}_{*}$ as $S_{L,j}$, $S_{R,j}$, $S_{L,*}$ and $S_{R,*}$, respectively. Let $\bar{S}_{L,j}=S_{L,j}\cup S_{L,*}$ and $\bar{S}_{R,j}=S_{R,j}\cup S_{R,*}$, for all $j=1,2,\dots,J$. Let $s_1 = s_L + s_{L,*}$ and $s_2 = s_R + s_{R,*}$.

First, we analyze the scaled hard thresholding step. For any $j=1,2,\dots,J$,
\begin{equation}
  \begin{split}
    & \|(\bm{L}_{j}\bm{Q}_{j}-\bm{L}_*)\bm{\Sigma}_*^{1/2}\|_\text{F}^2 + \|(\bm{R}_{j}\bm{Q}_{j}^{-\top}-\bm{R}_*)\bm{\Sigma}_*^{1/2}\|_\text{F}^2 \\
    = & \|(\bm{L}_{j}\bm{Q}_{j}-\bm{L}_*)_{\bar{S}_{L,j}}\bm{\Sigma}_*^{1/2}\|_\text{F}^2 + \|(\bm{R}_{j}\bm{Q}_{j}^{-\top}-\bm{R}_*)_{\bar{S}_{R,j}}\bm{\Sigma}_*^{1/2}\|_\text{F}^2 \\
    \leq & \|(\bm{L}_{j}\widetilde{\bm{Q}}_{j}-\bm{L}_*)_{\bar{S}_{L,j}}\bm{\Sigma}_*^{1/2}\|_\text{F}^2 + \|(\bm{R}_{j}\widetilde{\bm{Q}}_{j}^{-\top}-\bm{R}_*)_{\bar{S}_{R,j}}\bm{\Sigma}_*^{1/2}\|_\text{F}^2,
  \end{split}
\end{equation}
where 
\begin{equation}
  \widetilde{\bm{Q}}_j = \arg\min_{\bm{Q}}\{\|(\widetilde{\bm{L}}_{j}\bm{Q}-\bm{L}_*)\bm{\Sigma}_*^{1/2}\|_\text{F}^2 + \|(\widetilde{\bm{R}}_{j}\bm{Q}^{-\top}-\bm{R}_*)\bm{\Sigma}_*^{1/2}\|_\text{F}^2\}
\end{equation}
Then, as $\bm{L}_{j}=\text{HT}(\widetilde{\bm{L}}_{j}(\widetilde{\bm{R}}_j^\top\widetilde{\bm{R}}_{j})^{1/2},s_L)(\widetilde{\bm{R}}_j^{\top}\widetilde{\bm{R}}_j)^{-1/2}$, we have the upper bound
\begin{equation}
  \begin{split}
    & \|(\bm{L}_{j}\widetilde{\bm{Q}}_{j}-\bm{L}_*)\bm{\Sigma}_*^{1/2}\|_\text{F}^2 \\
    = & \|(\text{HT}(\widetilde{\bm{L}}_{j}(\widetilde{\bm{R}}_j^{\top}\widetilde{\bm{R}}_{j})^{1/2},s_L)(\widetilde{\bm{R}}_j^{\top}\widetilde{\bm{R}}_j)^{-1/2}\widetilde{\bm{Q}}_{j}\bm{\Sigma}_*^{1/2}-\bm{L}_*\bm{\Sigma}_*^{1/2}\|_\text{F}^2 \\
    = & \|(\text{HT}((\widetilde{\bm{L}}_{j})_{\bar{S}_{L,j}}(\widetilde{\bm{R}}_{j}^\top\widetilde{\bm{R}}_{j})^{1/2},s_L)(\widetilde{\bm{R}}_{j}^\top\widetilde{\bm{R}}_{j})^{-1/2}\widetilde{\bm{Q}}_{j}\bm{\Sigma}_*^{1/2}-(\bm{L}_*)_{\bar{S}_{L,j}}\bm{\Sigma}_*^{1/2}\|_\text{F}^2 \\
    \leq & \|(\text{HT}((\widetilde{\bm{L}}_{j})_{\bar{S}_{L,j}}(\widetilde{\bm{R}}_{j}^\top\widetilde{\bm{R}}_{j})^{1/2},s_L)-(\bm{L}_*)_{\bar{S}_{L,j}}\widetilde{\bm{Q}}_{j}^{-1}(\widetilde{\bm{R}}_{j}^\top\widetilde{\bm{R}}_{j})^{1/2}\|_\text{F}^2\cdot\|(\widetilde{\bm{R}}_{j}^\top\widetilde{\bm{R}}_{j})^{-1/2}\widetilde{\bm{Q}}_{j}\bm{\Sigma}_*^{1/2}\|_\text{op}^2 \\
    \leq & (1+q_1)\|(\widetilde{\bm{L}}_{j})_{\bar{S}_{L,j}}(\widetilde{\bm{R}}_{j}^\top\widetilde{\bm{R}}_{j})^{1/2} - (\bm{L}_*)_{\bar{S}_{L,j}}\widetilde{\bm{Q}}_{j}^{-1}(\widetilde{\bm{R}}_{j}^\top\widetilde{\bm{R}}_{j})^{1/2}\|_\text{F}^2\cdot\|(\widetilde{\bm{R}}_{j}^\top\widetilde{\bm{R}}_{j})^{-1/2}\widetilde{\bm{Q}}_{j}\bm{\Sigma}_*^{1/2}\|_\text{op}^2\\
    \leq & (1+q_1)\|(\widetilde{\bm{L}}_{j}\widetilde{\bm{Q}}_{j}-\bm{L}_*)_{\bar{S}_{L,j}}\bm{\Sigma}_*^{1/2}\|_\text{F}^2\cdot\|(\widetilde{\bm{R}}_{j}^\top\widetilde{\bm{R}}_{j})^{1/2}\widetilde{\bm{Q}}_{j}^{-\top}\bm{\Sigma}_*^{-1/2}\|_\text{op}^2\cdot\|(\widetilde{\bm{R}}_{j}^\top\widetilde{\bm{R}}_{j})^{-1/2}\widetilde{\bm{Q}}_{j}\bm{\Sigma}_*^{1/2}\|_\text{op}^2,
  \end{split}
\end{equation}
where $q_1=2\sqrt{s_{L,*}/(s_L-s_{L,*})}$. Note that by Lemma \ref{lemma:1},
\begin{equation}
  \|(\widetilde{\bm{R}}_{j}^\top\widetilde{\bm{R}}_{j})^{1/2}\widetilde{\bm{Q}}_{j}^{-\top}\bm{\Sigma}_*^{-1/2}\|_\text{op}^2 = \|\widetilde{\bm{R}}\bm{\Sigma}_*^{-1/2}\|_\text{op}^2 \leq (1+B)^2,
\end{equation}
and
\begin{equation}
  \|(\widetilde{\bm{R}}_{j}^\top\widetilde{\bm{R}}_{j})^{-1/2}\widetilde{\bm{Q}}_{j}\bm{\Sigma}_*^{1/2}\|_\text{op}^2 = \|\widetilde{\bm{R}}(\widetilde{\bm{R}}^\top\widetilde{\bm{R}})^{-1}\bm{\Sigma}_*^{1/2}\|_\text{op}^2 \leq (1-B)^{-2}.
\end{equation}
These imply the upper bound
\begin{equation}
  \|(\bm{L}_{j}\bm{Q}_{j}-\bm{L}_*)\bm{\Sigma}_*^{1/2}\|_\text{F}^2 \leq (1+q_1)(1+B)^2(1-B)^{-2} \|(\widetilde{\bm{L}}_{j}\widetilde{\bm{Q}}_{j}-\bm{L}_*)_{\bar{S}_{L,j}}\bm{\Sigma}_*^{1/2}\|_\text{F}^2.
\end{equation}
Similarly, we have
\begin{equation}
  \|(\bm{R}_{j}\bm{Q}_{j}^{-\top}-\bm{R}_*)\bm{\Sigma}_*^{1/2}\|_\text{F}^2 \leq (1+q_2)(1+B)^2(1-B)^{-2} \|(\widetilde{\bm{R}}_{j}\widetilde{\bm{Q}}_{j}^{-\top}-\bm{R}_*)_{\bar{S}_{R,j}}\bm{\Sigma}_*^{1/2}\|_\text{F}^2.
\end{equation}
Given $\max(q_1,q_2)\lesssim\alpha\beta^{-1}$ and $B=C\alpha^{1/2}\beta^{-1/2}$, by Lemma \ref{lemma:3}, we have
\begin{equation}
  \begin{split}
    & \|(\bm{L}_{j+1}\bm{Q}_{j+1}-\bm{L}_*)\bm{\Sigma}_*^{1/2}\|_\text{F}^2 + \|(\bm{R}_{j+1}\bm{Q}_{j+1}^{-\top}-\bm{R}_*)\bm{\Sigma}_*^{1/2}\|_\text{F}^2 \\
    \leq & (1+C\alpha\beta^{-1}) \left[\|(\widetilde{\bm{L}}_{j+1}\widetilde{\bm{Q}}_{j+1}-\bm{L}_*)_{\bar{S}_{L,j+1}}\bm{\Sigma}_*^{1/2}\|_\text{F}^2 + \|(\bm{R}_{j+1}\widetilde{\bm{Q}}_{j+1}^{-\top}-\bm{R}_*)_{\bar{S}_{R,j+1}}\bm{\Sigma}_*^{1/2}\|_\text{F}^2\right].
  \end{split}
\end{equation}

Second, we derive the upper bounds for the gradient descent step. By the optimality of $\widetilde{\bm{Q}}_{j+1}$, we have
\begin{equation}
  \begin{split}
    & \|(\widetilde{\bm{L}}_{j+1}\widetilde{\bm{Q}}_{j+1}-\bm{L}_*)_{\bar{S}_{L,j+1}}\bm{\Sigma}_*^{1/2}\|_\text{F}^2 + \|(\bm{R}_{j+1}\widetilde{\bm{Q}}_{j+1}^{-\top}-\bm{R}_*)_{\bar{S}_{R,j+1}}\bm{\Sigma}_*^{1/2}\|_\text{F}^2\\
    \leq & \|(\widetilde{\bm{L}}_{j+1}\bm{Q}_{j}-\bm{L}_*)_{\bar{S}_{L,j+1}}\bm{\Sigma}_*^{1/2}\|_\text{F}^2 + \|(\bm{R}_{j+1}\bm{Q}_{j}^{-\top}-\bm{R}_*)_{\bar{S}_{R,j+1}}\bm{\Sigma}_*^{1/2}\|_\text{F}^2.
  \end{split}
\end{equation}
Here, for any $\zeta>0$,
\begin{equation}
  \begin{split}
    & \|(\widetilde{\bm{L}}_{j+1}\bm{Q}_{j}-\bm{L}_*)_{\bar{S}_{L,j+1}}\bm{\Sigma}_*^{1/2}\|_\text{F}^2 \\
    \leq & (1+\zeta) \|(\bm{L}-\bm{L}_*-\eta\mathbb{E}[\nabla\mathcal{L}_j]\bm{R}_j(\bm{R}_j^\top\bm{R}_j)^{-1})_{\bar{S}_{L,j+1}}\bm{\Sigma}_*^{1/2}\|_\text{F}^2\\
    & + (1-\zeta^{-1})\eta^2\|(\bm{\Delta}_{L,j})_{\bar{S}_{L,j+1}}(\bm{R}_j^\top\bm{R}_j)^{-1/2}\bm{Q}_j\bm{\Sigma}_*^{1/2}\|_\text{F}^2.
  \end{split}
\end{equation}

Following the proof of Theorem \ref{thm:1}, they can be further bounded by
\begin{equation}
  \begin{split}
    & \|(\widetilde{\bm{L}}_{j+1}\widetilde{\bm{Q}}_{j+1}-\bm{L}_*)_{\bar{S}_{L,j+1}}\bm{\Sigma}_*^{1/2}\|_\text{F}^2 + \|(\bm{R}_{j+1}\widetilde{\bm{Q}}_{j+1}^{-\top}-\bm{R}_*)_{\bar{S}_{R,j+1}}\bm{\Sigma}_*^{1/2}\|_\text{F}^2\\
    \leq & (1-C\alpha\beta^{-1})(\|\bm{L}_{j}\bm{Q}_{j}-\bm{L}_*)\bm{\Sigma}_*^{1/2}\|_\text{F}^2 + \|(\bm{R}_{j}\bm{Q}_{j}^{-\top}-\bm{R}_*)\bm{\Sigma}_*^{1/2}\|_\text{F}^2)\\
    & + C\alpha^{-1}\beta^{-1}(\|(\bm{\Delta}_{L,j})_{\bar{S}_{L,j+1}}\|_\text{F}^2 + \|(\bm{\Delta}_{R,j})_{\bar{S}_{R,j+2}}\|_\text{F}^2).
  \end{split}
\end{equation}

Then, by the stability of the robust de-scaled gradient functions,
\begin{equation}
  \begin{split}
    \|(\bm{\Delta}_{L,j})_{\bar{S}_{L,j+1}}\|_\text{F}^2 & \leq \phi \|\bm{L}_j\bm{R}_j^\top - \bm{L}_*\bm{R}_*^\top\|_\text{F}^2 + \xi_{L,s_1}^2,\\
    \text{and }\|(\bm{\Delta}_{R,j})_{\bar{S}_{R,j+1}}\|_\text{F}^2 & \leq \phi \|\bm{L}_j\bm{R}_j^\top - \bm{L}_*\bm{R}_*^\top\|_\text{F}^2 + \xi_{R,s_2}^2.
  \end{split}
\end{equation}
Hence, given $\phi\lesssim\alpha^2$, we have the following upper bound
\begin{equation}
  \begin{split}
    & \|(\bm{L}_{j}\bm{Q}_{j}-\bm{L}_*)\bm{\Sigma}_*^{1/2}\|_\text{F}^2 + \|(\bm{R}_{j}\bm{Q}_{j}^{-\top}-\bm{R}_*)\bm{\Sigma}_*^{1/2}\|_\text{F}^2\\
    \leq & (1-C\alpha\beta^{-1})\left\{\|(\bm{L}_{t-1}\bm{Q}_{t-1}-\bm{L}_*)\bm{\Sigma}_*^{1/2}\|_\text{F}^2 + \|(\bm{R}_{t-1}\bm{Q}_{t-1}^{-\top}-\bm{R}_*)\bm{\Sigma}_*^{1/2}\|_\text{F}^2\right\}\\ 
    & + C\alpha^{-1}\beta^{-1}(\xi_{L,s_1}^2+\xi_{R,s_2}^2)\\
    \leq & \cdots \\
    \leq & (1-C\alpha\beta^{-1})^t\left\{\|(\bm{L}_{0}\bm{Q}_{0}-\bm{L}_*)\bm{\Sigma}_*^{1/2}\|_\text{F}^2 + \|(\bm{R}_{0}\bm{Q}_{0}^{-\top}-\bm{R}_*)\bm{\Sigma}_*^{1/2}\|_\text{F}^2\right\}\\ 
    & + C\alpha^{-2}(\xi_{L,s_1}^2+\xi_{R,s_2}^2).
  \end{split}
\end{equation}
The remaining proof follows a similar argument to that in Appendix \ref{append:A1}.
\end{proof}

\subsection{Auxiliary Lemmas}

We state some auxiliary lemmas. The first lemma is Lemma 12 in \citet{tong2021accelerating}, which presents the perturbation bounds for matrix decomposition.

\begin{lemma}\label{lemma:1}
  For any $\bm{L}\in\mathbb{R}^{p\times r}$ and $\bm{R}\in\mathbb{R}^{q\times r}$, suppose that
  \begin{equation}
    \max\left\{\|(\bm{L}-\bm{L}_*)\bm{\Sigma}_*^{-1/2}\|_\textup{op},\|(\bm{R}-\bm{R}_*)\bm{\Sigma}_*^{-1/2}\|_\textup{op}\right\}<1,
  \end{equation}
  then we have
  \begin{equation}
    \begin{split}
      \|\bm{L}\bm{\Sigma}_*^{-1/2}\|_\textup{op} & \leq 1+\|(\bm{L}-\bm{L}_*)\bm{\Sigma}_*^{-1/2}\|_\textup{op},\\
      \|\bm{R}\bm{\Sigma}_*^{-1/2}\|_\textup{op} & \leq 1+\|(\bm{R}-\bm{R}_*)\bm{\Sigma}_*^{-1/2}\|_\textup{op},\\
      \|\bm{L}(\bm{L}^\top\bm{L})^{-1}\bm{\Sigma}_*^{1/2}\|_\textup{op} & \leq \frac{1}{1-\|(\bm{L}-\bm{L}_*)\bm{\Sigma}_*^{-1/2}\|_\textup{op}},\\
      \|\bm{R}(\bm{R}^\top\bm{R})^{-1}\bm{\Sigma}_*^{1/2}\|_\textup{op} & \leq \frac{1}{1-\|(\bm{R}-\bm{R}_*)\bm{\Sigma}_*^{-1/2}\|_\textup{op}},\\
      \|\bm{L}(\bm{L}^\top\bm{L})^{-1}\bm{\Sigma}_*^{1/2}-\bm{U}_*\|_\textup{op} & \leq \frac{\sqrt{2}\|(\bm{L}-\bm{L}_*)\bm{\Sigma}_*^{-1/2}\|_\textup{op}}{1-\|(\bm{L}-\bm{L}_*)\bm{\Sigma}_*^{-1/2}\|_\textup{op}},\\
      \text{and }\|\bm{R}(\bm{R}^\top\bm{R})^{-1}\bm{\Sigma}_*^{1/2}-\bm{V}_*\|_\textup{op} & \leq \frac{\sqrt{2}\|(\bm{R}-\bm{R}_*)\bm{\Sigma}_*^{-1/2}\|_\textup{op}}{1-\|(\bm{R}-\bm{R}_*)\bm{\Sigma}_*^{-1/2}\|_\textup{op}},
    \end{split}
  \end{equation}
  where $\bm{U}_*=\bm{L}_*\bm{\Sigma}_*^{-1/2}$ and $\bm{V}_*=\bm{R}_*\bm{\Sigma}_*^{-1/2}$.
\end{lemma}

The second lemma combines Lemmas 11 and 13 in \citet{tong2021accelerating}, constructing an approximate equivalence between the distance $\text{dist}(\bm{F},\bm{F}_*)$ and $\|\bm{L}\bm{R}^\top - \bm{L}_*\bm{R}_*^\top\|_\text{F}$.

\begin{lemma}\label{lemma:2}
  For any $\bm{L}\in\mathbb{R}^{p\times r}$ and $\bm{R}\in\mathbb{R}^{q\times r}$,
  \begin{equation}
    \begin{split}
      \|\bm{L}\bm{R}^\top-\bm{L}_*\bm{R}_*^\top\|_\textup{F}\leq&\left(1+\frac{1}{2}\max\left\{\|(\bm{L}-\bm{L}_*)\bm{\Sigma}_*^{-1/2}\|_\textup{op},\|(\bm{R}-\bm{R}_*)\bm{\Sigma}_*^{-1/2}\|_\textup{op}\right\}\right)\\
      &\times[\|(\bm{L}-\bm{L}_*)\bm{\Sigma}_*^{1/2}\|_\textup{F}+\|(\bm{L}-\bm{L}_*)\bm{\Sigma}_*^{1/2}\|_\textup{F}].
    \end{split}
  \end{equation}
  and
  \begin{equation}
    d(\bm{F},\bm{F}_*)^2 \leq 2\|\bm{L}\bm{R}^\top-\bm{L}_*\bm{R}_*^\top\|_\textup{F}^2.
  \end{equation}
\end{lemma}

The third lemma develops an upper bound for hard thresholding operation, as in Lemma 3.3 of \citet{li2016nonconvex}.

\begin{lemma}\label{lemma:3}
  For $\bbm{\theta}_*\in\mathbb{R}^d$ such that $\|\bbm{\theta}_*\|_0\leq k_*$ and hard thresholding operator $\textup{HT}(\cdot;k)$ with $k>k_*$, we have
  \begin{equation}
    \|\textup{HT}(\bbm{\theta};k) - \bbm{\theta}\|_2^2 \leq \left(1 + \frac{2\sqrt{k_*}}{\sqrt{k-k_*}}\right)\|\bbm{\theta}-\bbm{\theta}_*\|_2^2.
  \end{equation}
\end{lemma}

\section{Matrix Trace Regression}\label{append:B}

\subsection{Initialization Guarantees}\label{append:B.1}

We first present the statistical guarantees of initialization.

\begin{proposition}[Rates of Initialization]
  \label{prop:matreg_initial}
  For the robust initialization in \eqref{eq:DS_init}, if the tuning parameters are selected as 
  \begin{equation}
    \begin{split}
      \tau_x \asymp & \left(\frac{nM_{x,2+2\lambda,1}}{\log d}\right)^{\frac{1}{1+\lambda}}, \quad 
      \tau_{yx}\asymp\left(\frac{nM_{yx,1+\delta}}{\log d}\right)^{\frac{1}{1+\delta}},\\
      \text{and}\quad R\asymp & \|\bbm{\theta}_\ast\|_1 \cdot \left(\frac{M_{x,2+2\lambda,1}^{1/\lambda}\log d}{n}\right)^{\frac{\lambda}{1+\lambda}} + \left(\frac{M_{yx,1+\delta}^{1/\delta}\log d}{n}\right)^{\frac{\delta}{1+\delta}},
    \end{split}
  \end{equation}
  and the sample size satisfies 
  \begin{equation}
    \begin{split}
      n & \gtrsim \|\bm{\Sigma}_x^{-1}\|_{1,\infty}^{\frac{1+\delta}{\delta}}\sigma_K^{-\frac{1+\delta}{\delta}}\left(\beta_x/\alpha_x\right)^{\frac{1+\delta}{2\delta}}s_\ast^{\frac{1+\delta}{2\delta}}\left(\|\bbm{\theta}_\ast\|_1^{\frac{1+\delta}{\delta}}M_{x,2+2\lambda,1}^\frac{1+\delta}{\delta(1+\lambda)}+M_{yx,1+\delta}^\delta\right)\log d \\
      & =: \mathfrak{C}_1 \cdot \left(s K \right)^{\frac{1+\delta}{2\delta}}\log d,
    \end{split}
  \end{equation}
  where $\delta=\min(\lambda,\epsilon)$, and $M_{yx,1+\delta}=2M_{\mathrm{eff},1+\delta,1}+2M_{x,2+2\delta}\|\bbm{\theta}_\ast\|_2^{1+\delta}$. Then, with probability at least $1-C\exp(-C\log d)$, the initial estimator satisfies
  \begin{equation}
    \|\bm{\Theta}_{0} - \bm{\Theta}_\ast\|_{\textup{F}}^2 \lesssim \alpha_x\beta_x^{-1} \sigma^2_K.
  \end{equation}
\end{proposition}

\begin{proof}[Proof of Proposition \ref{prop:matreg_initial}]

The proof consists of two steps. In the first step, we develop the theoretical results for the initialization via the robust Dantzig selector in \eqref{eq:DS_init}. In the second step, we establish the estimation accuracy of $\wh{\bbm\Sigma}_x(\tau_x)$ and $\wh{\bbm\sigma}_{yx}(\tau_{yx})$. Then, we use these conclusions to verify the initialization condition.\\

\noindent\textit{Step 1.} (Error bounds of robust Dantzig Selector)

\noindent To begin with, we show that under certain accuracy conditions of $\widehat{\bm{\Sigma}}_x$ and $\widehat{\bbm{\sigma}}_{yx}$, the estimation error of $\widehat{\bbm{\theta}}_{\text{DS}}(R,\tau_x,\tau_{yx})$ is well-behaved with properly selected $R$. Specifically, when
\begin{equation}
  \|\widehat{\bm{\Sigma}}_x - \bm{\Sigma}_x\|_\infty\leq\zeta_0\quad\text{and}\quad \|\widehat{\bbm{\sigma}}_{yx} - \bm{\sigma}_{yx}\|_\infty\leq\zeta_1,
\end{equation}
and $R\geq\zeta_0\|\bbm{\theta}_*\|_1+\zeta_1$, then the error $\widehat{\bbm{\theta}}-\bbm{\theta}_*$ can be bounded in terms of multiple vector norms.

Note that
\begin{equation}
  \begin{split}
    & \|\widehat{\bm{\Sigma}}_x\bbm{\theta}_* - \widehat{\bbm{\sigma}}_{yx}\|_\infty = \|\widehat{\bm{\Sigma}}_x\bbm{\theta}_* - \bbm{\sigma}_{yx} + \bbm{\sigma}_{yx} - \widehat{\bbm{\sigma}}_{yx}\|_\infty\\
    = & \|\widehat{\bm{\Sigma}}_x\bbm{\theta}_* - \bm{\Sigma}_x\bbm{\theta}_* + \bbm{\sigma}_{yx} - \widehat{\bbm{\sigma}}_{yx}\|_\infty \\
    \leq & \|(\widehat{\bm{\Sigma}}_x - \bm{\Sigma}_x)\bbm{\theta}_*\|_\infty + \|\bbm{\sigma}_{yx} - \widehat{\bbm{\sigma}}_{yx}\|_\infty \\ 
    \leq & \|(\widehat{\bm{\Sigma}}_x - \bm{\Sigma}_x)\|_\infty \|\bbm{\theta}_*\|_1 + \|\bbm{\sigma}_{yx} - \widehat{\bbm{\sigma}}_{yx}\|_\infty \leq \zeta_0\|\bbm{\theta}_*\|_1 + \zeta_1 \leq R.
  \end{split}
\end{equation}
Therefore, $\bbm{\theta}_*$ is feasible in the optimization constraint, and hence $\|\widehat{\bbm{\theta}}\|_1 \leq \|\bbm{\theta}_*\|_1$. Then, by triangle inequality, we have
\begin{equation}
  \begin{split}
    & \|\widehat{\bbm{\theta}} - \bm{\Sigma}_x^{-1}\bbm{\sigma}_{yx}\|_\infty = \|\bm{\Sigma}_x^{-1}[\bm{\Sigma}_x\widehat{\bbm{\theta}} - \bbm{\sigma}_{yx}]\|_\infty \\
    = & \|\bm{\Sigma}_x^{-1}[\bm{\Sigma}_x\widehat{\bbm{\theta}} - \widehat{\bm{\Sigma}}_x\widehat{\bbm{\theta}} + \widehat{\bm{\Sigma}}_x\widehat{\bbm{\theta}} - \widehat{\bbm{\sigma}}_{yx} + \widehat{\bbm{\sigma}}_{yx} - \bbm{\sigma}_{yx}]\|_\infty\\
    \leq & \|\bm{\Sigma}_x^{-1}\|_{1,\infty}\cdot\|(\bm{\Sigma}_x - \widehat{\bm{\Sigma}}_x)\widehat{\bbm{\theta}}\|_\infty + \|\bm{\Sigma}_x^{-1}\|_{1,\infty}\cdot\|\widehat{\bm{\Sigma}}_x\widehat{\bbm{\theta}} - \widehat{\bm{\sigma}}_{yx}\|_\infty + \|\bm{\Sigma}_x^{-1}\|_{1,\infty}\cdot\|\widehat{\bbm{\sigma}}_{yx} - \bbm{\sigma}_{yx}\|_\infty\\
    \leq & \|\bm{\Sigma}_x^{-1}\|_{1,\infty}\cdot(\zeta_0\|\widehat{\bbm{\theta}}\|_1 + R + \zeta_1) \leq 2R\|\bm{\Sigma}_x^{-1}\|_{1,\infty}.
  \end{split}
\end{equation}

As $\bbm{\theta}_*$ is a $s_*$-sparse vector with nonzero support $S$ and $s_*=s_{1,*}s_{2,*}$, we have $\|\widehat{\bbm{\theta}}\|_1 \leq \|\bbm{\theta}_*\|_1 = \|(\bbm{\theta}_*)_S\|_1$. Moreover, we have
\begin{equation}\label{eq:theta1_decmopose}
  \begin{split}
    \|\widehat{\bbm{\theta}}\|_1 = & \|\bbm{\theta}_* + (\widehat{\bbm{\theta}} - \bbm{\theta}_*)\|_1 = \|(\bbm{\theta}_*)_S + (\widehat{\bbm{\theta}} - \bbm{\theta}_*)_S + (\widehat{\bbm{\theta}} - \bbm{\theta}_*)_{S^\perp}\|_1 \\
    \geq & \|(\bbm{\theta}_*)_S + (\widehat{\bbm{\theta}} - \bbm{\theta}_*)_{S^\perp}\|_1 - \|(\widehat{\bbm{\theta}} - \bbm{\theta}_*)_S\|_1 \\
    \geq & \|(\bbm{\theta}_*)_S\|_1 + \|(\widehat{\bbm{\theta}} - \bbm{\theta}_*)_{S^\perp}\|_1 - \|(\widehat{\bbm{\theta}} - \bbm{\theta}_*)_S\|_1,
  \end{split}
\end{equation}
where the first inequality follows from the triangle inequality and the second one follows from the decomposability of $\|\cdot\|_1$. Therefore, we have
\begin{equation}
  \|(\widehat{\bbm{\theta}} - \bbm{\theta}_*)_{S^\perp}\|_1 \leq \|(\widehat{\bbm{\theta}} - \bbm{\theta}_*)_S\|_1.
\end{equation}

Hence, the upper bound in terms of $\|\cdot\|_1$ can be bounded as
\begin{equation}
  \begin{split}
    & \|\widehat{\bbm{\theta}} - \bbm{\theta}_*\|_1 = \|(\widehat{\bbm{\theta}} - \bbm{\theta}_*)_S\|_1 + \|(\widehat{\bbm{\theta}} - \bbm{\theta}_*)_{S^\perp}\|_1\\
    \leq & 2\|(\widehat{\bbm{\theta}} - \bbm{\theta}_*)_S\|_1 \leq 2s_*\|\widehat{\bbm{\theta}} - \bbm{\theta}_*\|_\infty \leq 4s_*R\|\bm{\Sigma}_x^{-1}\|_{1,\infty}.
  \end{split}
\end{equation}
Finally, by the duality between $\|\cdot\|_1$ and $\|\cdot\|_\infty$ and $R\asymp \zeta_0\|\bbm\theta_*\|_1+\zeta_1$, we have
\begin{equation}\label{eq:init_bound}
  \begin{split}
    \|\widehat{\bbm{\theta}} - \bbm{\theta}_*\|_2 \leq & \|\widehat{\bbm{\theta}} - \bbm{\theta}_*\|_1^{1/2}\|\widehat{\bbm{\theta}} - \bbm{\theta}_*\|_\infty^{1/2} \leq \sqrt{8s_*}R\|\bm{\Sigma}_x^{-1}\|_{1,\infty}\\
    \lesssim & \|\bm{\Sigma}_x^{-1}\|_{1,\infty}\sqrt{8s_*}(\zeta_0\|\bbm{\theta_*}\|_1+\zeta_1).
  \end{split}
\end{equation}~

\noindent\textit{Step 2.} (Accuracy of robust covariance estimation)

\noindent Next, we establish the estimation accuracy $\zeta_0$ and $\zeta_1$. Denote $\bm x_i=\text{Vec}(\bm X_i)$. For each entry of $\bm{x}_i\bm{x}_i^\top$, it has a finite $(1+\lambda)$-th moment $M_{x,2+2\lambda,1}$. Similarly, the $j$-th entry in $Y_i\bm{x}_i$ is
\begin{equation}
  (E_i+\bm{x}_i^\top\bbm{\theta}_*)x_{ij} = E_ix_{ij} + x_{ij}\bm{x}_i^\top\bbm{\theta}_*,
\end{equation}
where $x_{ij}$ represents the $j$-th entry of $\bm x_i$.
For $\delta=\min(\lambda,\epsilon)$, the first term $E_ix_{ij}$ has a finite $(1+\delta)$-th moment $M_{e,1+\delta}M_{x,1+\delta}$ and the second term $x_{ij}\bm{x}_i^\top\bbm{\theta}_*$ has a finite $(1+\delta)$-th moment
\begin{equation}
  \mathbb{E}[|x_{ij}\bm{x}_i^\top\bbm{\theta}_*|^{1+\delta}] \leq \mathbb{E}[|\bm{x}_i^\top\bm{c}_j|^{2+2\delta}]^{1/2}\cdot\mathbb{E}[|\bm{x}_i^\top\bbm{\theta}_*|^{2+2\delta}]^{1/2} \leq M_{x,2+2\delta}\cdot\|\bbm{\theta}_*\|_2^{1+\delta}.
\end{equation}
Combining these two moment bounds, each entry of $Y_ix_i$ has a finite $(1+\delta)$-th moment $M_{yx,1+\delta}=2M_{\mathrm{eff},1+\delta,1}+2M_{x,2+2\delta}\|\bbm{\theta}_*\|_2^{1+\delta}$.

To analyze $\zeta_0$, we establish an upper bound for each entry of $\widehat{\bm{\Sigma}}_x(\tau_x)-\bm{\Sigma}_x$. The $(j,k)$-th entry can be upper bounded as
\begin{equation}\label{eq:zeta_0}
  \left|\left(\widehat{\bm{\Sigma}}_x(\tau_x)-\bm{\Sigma}_x\right)_{jk}\right|\leq \left|\left(\widehat{\bm{\Sigma}}_x(\tau_x)-\mathbb E[\widehat{\bm{\Sigma}}_x(\tau_x)]\right)_{jk}\right|+\left|\left(\mathbb E[\widehat{\bm{\Sigma}}_x(\tau_x)]-\bm{\Sigma}_x\right)_{jk}\right|.
\end{equation}
Similar to the proof of Theorem \ref{thm:matrix_trace_reg}, when $\tau_x\asymp ( n M_{x,2+2\lambda}/\log d)^{1/(1+\lambda)}$, the second term can be bounded as
\begin{equation}
\begin{aligned}
  \left|\left(\bb E[\widehat{\bm{\Sigma}}_x(\tau_x)]-\bm{\Sigma}_x\right)_{jk}\right| & =\left|\bb E\left[\text{T}\left(\bm c_j^\top \bm x_i\bm x_i^\top \bm c_k,\tau_x\right)-\bm c_j^\top \bm x_i\bm x_i^\top \bm c_k\right]\right|\\
  & \leq \bb E\left[|\bm c_j^\top \bm x_i\bm x_i^\top\bm c_k| \cdot 1\{|\bm c_j^\top \bm x_i\bm x_i^\top\bm c_k|\geq \tau_x\}\right]\\
  & \leq \bb E\left[|\bm c_j^\top \bm x_i\bm x_i^\top\bm c_k|^{1+\lambda}\right]^\frac{1}{1+\lambda}\cdot \bb P\left(|\bm c_j^\top \bm x_i\bm x_i^\top\bm c_k|\geq \tau_x\right)^\frac{\lambda}{1+\lambda}\\
  & \leq \bb E\left[|\bm c_j^\top \bm x_i\bm x_i^\top\bm c_k|^{1+\lambda}\right]\cdot \tau_x^{-\lambda}\\
  & \leq \bb E[|\bm c_j^\top \bm x_i|^{2+2\lambda}]^\frac{1}{2}\cdot \bb E[|\bm x_i^\top\bm c_k|^{2+2\lambda}]^\frac{1}{2}\cdot \tau_x^{-\lambda}\\
  & \leq M_{x,2+2\lambda,1}\cdot \tau_x^{-\lambda}\\
  & \asymp \left(\frac{M_{x,2+2\lambda,1}^{1/\lambda}\log d}{n}\right)^\frac{\lambda}{1+\lambda}.
\end{aligned}
\end{equation}
For the first term of \eqref{eq:zeta_0}, we have
\begin{equation}
\begin{aligned}
  \text{Var}\left(\text{T}\left(\bm c_j^\top \bm x_i\bm x_i^\top \bm c_k,\tau_x\right)\right) & \leq \bb E\left[\left(\text{T}\left(\bm c_j^\top \bm x_i\bm x_i^\top \bm c_k,\tau_x\right)\right)^2\right]\\
  & \leq \tau_x^{1-\lambda}\bb E[|\bm c_j^\top \bm x_i\bm x_i^\top \bm c_k|^{1+\lambda}]\\
  & \leq \tau_x^{1-\lambda} M_{x,2+2\lambda,1}.
\end{aligned}
\end{equation}
By Bernstein's inequality, for any $t\geq 0$,
\begin{equation}
  \bb P\left(\left|\frac{1}{n}\sum_{i=1}^n \text{T}\left(\bm c_j^\top \bm x_i\bm x_i^\top \bm c_k,\tau_x\right)-\bb E\left[\text{T}\left(\bm c_j^\top \bm x_i\bm x_i^\top \bm c_k,\tau_x\right)\right]\right|\geq t\right) \leq 2\exp \left(-\frac{nt^2/2}{\tau_x^{1-\lambda} M_{x,2+2\lambda,1}+\tau_x t}\right).
\end{equation}
Letting $t=CM_{x,2+2\lambda}\cdot\tau_x^{-\lambda}$, we have that
$$\begin{aligned}
  &\mathbb{P}\left(\norm{\widehat{\bm{\Sigma}}_x(\tau_x)-\bb E[\widehat{\bm{\Sigma}}_x(\tau_x)]}_\infty\gtrsim \left(\frac{M_{x,2+2\lambda,1}^{1/\lambda}\log d}{n}\right)^\frac{\lambda}{1+\lambda}\right)\\
  \leq& 2d_1^2d_2^2\exp(-C\log d)\leq \exp(-C\log d).
\end{aligned}
$$
Combining the two upper bounds, we have that with probability at least $1-\exp(-C\log d)$,
$$
\left\|\widehat{\bm{\Sigma}}_x(\tau_x)-\bm{\Sigma}_x\right\|_\infty \lesssim  \left(\frac{M_{x,2+2\lambda,1}^{1/\lambda}\log d}{n}\right)^\frac{\lambda}{1+\lambda}=\zeta_0.
$$

Next, we analyze the estimation accuracy of $\widehat{\bbm \sigma}_{yx}$, quantified by $\zeta_1$. Similarly, the $j$-th entry of $(\widehat{\bbm{\sigma}}_{yx}(\tau_{yx})-\bbm{\sigma}_{yx})$ can be upper bounded as
\begin{equation}\label{eq:zeta_1}
\begin{aligned}
  \left|(\widehat{\bbm{\sigma}}_{yx}(\tau_{yx})-\bbm{\sigma}_{yx})_{j}\right|\leq &\left|\frac{1}{n}\sum_{i=1}^n\text{T}(Y_i\bm x_{ij},\tau_{yx})-\bb E[\text{T}(Y_i\bm x_{ij},\tau_{yx})]\right|\\
  & +\left|\bb E[\text{T}(Y_i\bm x_{ij},\tau_{yx})]-\bb E[Y_i\bm x_{ij}]\right|.
\end{aligned}
\end{equation}
When $\tau_{yx}\asymp \left(n M_{yx,1+\delta}/\log d\right)^{1/(1+\delta)}$, the second term can be bounded as
\begin{equation}
\begin{aligned}
  \left|\bb E[\text{T}(Y_i\bm x_{ij},\tau_{yx})-Y_i\bm x_{ij}]\right| \leq & \bb E[|Y_i\bm x_{ij}|\cdot 1\{|Y_i\bm x_{ij}|\geq \tau_{yx}\}]\\
  \leq &\bb E[|Y_i\bm x_{ij}|^{1+\delta}]^\frac{1}{1+\delta}\cdot \bb P\left(|Y_i\bm x_{ij}|\geq \tau_{yx}\right)^\frac{\delta}{1+\delta}\\
  \leq & E[|Y_i\bm x_{ij}|^{1+\delta}]\cdot \tau_{yx}^{-\delta}\\
  \lesssim & \left(\frac{M_{yx,1+\delta}^{1/\delta}\log d}{n}\right)^{\frac{\delta}{1+\delta}}.
\end{aligned}
\end{equation}
Meanwhile, for the first term of \eqref{eq:zeta_1}, by Bernstein's inequality and the fact that
$$
  \text{Var}(\text{T}(Y_i\bm x_{ij},\tau_{yx}))\leq \bb E[(\text{T}(Y_i\bm x_{ij},\tau_{yx}))^2]\leq \tau_{yx}^{1-\delta}M_{yx,1+\delta},
$$
we have that for any $t\geq 0$,
$$
\bb P\left(\left|\frac{1}{n}\sum_{i=1}^n\text{T}(Y_i\bm x_{ij},\tau_{yx})-\bb E[\text{T}(Y_i\bm x_{ij},\tau_{yx})]\right|\geq t\right)\leq 2\exp\left(-\frac{nt^2/2}{\tau_{yx}^{1-\delta}M_{yx,1+\delta}+\tau_{yx}t}\right).
$$
Letting $t=CM_{yx,1+\delta}\cdot\tau_{yx}^{-\delta}$, we obtain that
$$
\begin{aligned}
  &\mathbb{P}\left(\left\|\widehat{\bbm{\sigma}}_{yx}(\tau_{yx})-\mathbb{E}[\wh{\bbm{\sigma}}_{yx}(\tau_{yx})]\right\|_\infty \gtrsim\left(\frac{M_{yx,1+\delta}^{1/\delta}\log d}{n}\right)^{\frac{\delta}{1+\delta}}\right)\\
\leq &2d_1d_2\exp (-C\log d)\leq \exp (-C\log d)
\end{aligned}
$$
Combining the two pieces, we have that with probability at least $1-\exp(-C\log d)$,
$$
\left\|\widehat{\bbm{\sigma}}_{yx}(\tau_{yx})-\bbm{\sigma}_{yx}\right\|_\infty \lesssim \left(\frac{M_{yx,1+\delta}^{1/\delta}\log d}{n}\right)^{\frac{\delta}{1+\delta}}=\zeta_1.
$$

Plugging $\zeta_0$ and $\zeta_1$ into \eqref{eq:init_bound}, we have
\begin{equation}
  \begin{split}
    & \|\widehat{\bbm{\theta}}(R,\tau_x,\tau_{yx})-\bbm{\theta}_*\|_2 \\
    \lesssim & \|\bm{\Sigma}_x^{-1}\|_{1,\infty}\sqrt{s_*} \left(\left[\frac{M_{x,2+2\lambda,1}^{1/\lambda}\log d}{n}\right]^{\frac{\lambda}{1+\lambda}}\|\bbm{\theta}_*\|_1 + \left[\frac{M_{yx,1+\delta}^{1/\delta}\log d}{n}\right]^{\frac{\delta}{1+\delta}}\right) \\
    \lesssim & \|\bm{\Sigma}_x^{-1}\|_{1,\infty}\sqrt{s_*}\left(\|\bbm{\theta}_*\|_1M_{x,2+2\lambda,1}^{1/(1+\lambda)}+M_{yx,1+\delta}^{1/(1+\delta)}\right)\left[\frac{\log d}{n}\right]^{\frac{\delta}{1+\delta}}.
  \end{split}
\end{equation}
Hence, when
\begin{equation}
  n \gtrsim \|\bm{\Sigma}_x^{-1}\|_{1,\infty}^{\frac{1+\delta}{\delta}}\sigma_K^{-\frac{1+\delta}{\delta}}\left(\beta_x/\alpha_x\right)^{\frac{1+\delta}{2\delta}}s_*^{\frac{1+\delta}{2\delta}}\left(\|\bbm{\theta}_*\|_1^{\frac{1+\delta}{\delta}}M_{x,2+2\lambda,1}^\frac{1+\delta}{\delta(1+\lambda)}+M_{yx,1+\delta}^\delta\right)\log d,
\end{equation}
we have $\|\widehat{\bbm{\theta}}(R,\tau_x,\tau_{yx})-\bbm{\theta}_*\|_2 \lesssim (\alpha_x/\beta_x)^{1/2}\sigma_K$. Then, by Lemma \ref{lemma:3} and the Young-Eckart-Mirsky Theorem,
$$
d(\bm F_0,\bm F_*)\leq 2\norm{\bbm\Theta_0-\bbm\Theta_*}_\mathrm{F}\lesssim(\alpha_x/\beta_x)^{1/2}\sigma_K.
$$
\end{proof}

\subsection{Stability Conditions}\label{append:B.2}

We first verify that the stability condition in Definition \ref{def:4} holds for the matrix trace regression and provide explicit forms for $\xi^2_{L,s_1}$, $\xi^2_{R,s_2}$, and $\phi$. 

\begin{proposition}[Stability of Robust De-scaled Gradient Estimators]\label{prop:stability_linear_reg}
  Under the conditions in Theorem \ref{thm:matrix_trace_reg}, with a probability at least $1-C\exp(-C\log d)$, the robust de-scaled gradient estimators satisfy
  \begin{equation}
    \begin{split}
      & \|\bm{G}_L(\bm{L},\bm{R};\tau) - \mathbb{E}[\nabla_{\bm{L}}f(\bm{L},\bm{R};z_i)](\bm{R}^\top\bm{R})^{-1/2}\|_\textup{F}^2 \lesssim \xi^2_L + \phi_{\lambda,\epsilon}\|\bm{L}\bm{R}^\top-\bm{L}_*\bm{R}_*^\top\|_\textup{F}^2,\\
      & \|\bm{G}_R(\bm{L},\bm{R};\tau) - \mathbb{E}[\nabla_{\bm{R}}f(\bm{L},\bm{R};z_i)](\bm{R}^\top\bm{R})^{-1/2}\|_\textup{F}^2 \lesssim \xi^2_L + \phi_{\lambda,\epsilon}\|\bm{L}\bm{R}^\top-\bm{L}_*\bm{R}_*^\top\|_\textup{F}^2,
    \end{split}
  \end{equation}
  where 
  $$
    \phi_{\lambda,\epsilon}=s_1K\bar{\sigma}_K^{2\lambda}\overline{M}_{x,\textup{eff}}^2\left(\frac{\log d}{n}\right)^\frac{2\min(\epsilon,\lambda)}{1+\epsilon},
  $$ 
  $$
    \overline{M}_{x,\mathrm{eff}}^2=M_{x,2+2\lambda}^2M_{\mathrm{eff},1+\epsilon,s}^\frac{-2\lambda}{1+\epsilon}+M_{x,2+2\lambda}^\frac{2}{1+\lambda}M_{\mathrm{eff},1+\epsilon,s}^\frac{2\epsilon(\lambda-1)}{(1+\epsilon)(1+\lambda)}+M_{x,2+2\lambda}^\frac{2}{1+\lambda},
  $$
  $$
    \xi^2_L=s_1K\left[\frac{M_{\textup{eff},1+\epsilon,s}^{1/\epsilon}\log d}{n}\right]^{\frac{2\epsilon}{1+\epsilon}},\quad\text{and}\quad
    \xi^2_R=s_2K\left[\frac{M_{\textup{eff},1+\epsilon,s}^{1/\epsilon}\log d}{n}\right]^{\frac{2\epsilon}{1+\epsilon}}.
  $$
\end{proposition}

\begin{proof}[Proof of Proposition \ref{prop:stability_linear_reg}]

Denote $\bm{Z}_i=\mathcal{P}(\bm{X}_i)$ and $\widetilde{\bm{R}}=\bm{R}(\bm{R}^\top\bm{R})^{-1/2}$. Note that $\wt{\bm R}$ is a matrix with orthonormal columns. By definition,
\begin{equation}
  \begin{aligned}
    & \bm{G}_L(\bm{L},\bm{R};\tau)-\mathbb{E}[\nabla_{\bm{L}}f(\bm{L},\bm{R};z_i)](\bm{R}^\top\bm{R})^{-1/2}\\
    = & \frac{1}{n}\sum_{i=1}^n\text{T}\left(\left\{\langle\bm Z_i,\bm{L}\bm{R}^\top\rangle-Y_i\right\}\bm Z_i\wt{\bm{R}},\tau\right)-\bb{E}\left[\left\{\langle\bm Z_i,\bm{L}\bm{R}^\top\rangle-Y_i\right\}\bm Z_i\wt{\bm{R}}\right] \\
    = & \Bigg\{ \mathbb{E}\left[\text{T}\left(\left\{\inner{\bm Z_i}{\bm L_*\bm R_*^\top}-Y_i\right\}\bm{Z}_i\widetilde{\bm{R}},\tau\right)\right] - \mathbb{E}\left[\left\{\inner{\bm Z_i}{\bm L_*\bm R_*^\top}-Y_i\right\}\bm{Z}_i\widetilde{\bm{R}}\right]\Bigg\}\\
    + & \Bigg\{\frac{1}{n}\sum_{i=1}^n\text{T}\left(\left\{\inner{\bm Z_i}{\bm L_*\bm R_*^\top}-Y_i\right\}\bm{Z}_i\widetilde{\bm{R}},\tau\right) - \mathbb{E}\left[\text{T}\left(\left\{\inner{\bm Z_i}{\bm L_*\bm R_*^\top}-Y_i\right\}\bm{Z}_i\widetilde{\bm{R}},\tau\right)\right] \Bigg\}\\
    + & \Bigg\{ \mathbb{E}\left[\left\{\inner{\bm Z_i}{\bm L_*\bm R_*^\top}-Y_i\right\}\bm{Z}_i\widetilde{\bm{R}}\right] - \mathbb{E}\left[\left\{\inner{\bm Z_i}{\bm L\bm R^\top}-Y_i\right\}\bm{Z}_i\widetilde{\bm{R}}\right] \\
    & + \mathbb{E}\left[\text{T}\left(\left\{\inner{\bm Z_i}{\bm L\bm R^\top}-Y_i\right\}\bm{Z}_i\widetilde{\bm{R}};\tau\right)\right] - \mathbb{E}\left[\text{T}\left(\left\{\inner{\bm Z_i}{\bm L_*\bm R_*^\top}-Y_i\right\}\bm{Z}_i\widetilde{\bm{R}};\tau\right)\right]\Bigg\}\\
    + & \Bigg\{ \frac{1}{n}\sum_{i=1}^n\text{T}\left(\left\{\inner{\bm Z_i}{\bm L\bm R^\top}-Y_i\right\}\bm{Z}_i\widetilde{\bm{R}};\tau\right) - \frac{1}{n}\sum_{i=1}^n\text{T}\left(\left\{\inner{\bm Z_i}{\bm L_*\bm R_*^\top}-Y_i\right\}\bm{Z}_i\widetilde{\bm{R}};\tau\right)\\
    & - \mathbb{E}\left[\text{T}\left(\left\{\inner{\bm Z_i}{\bm L\bm R^\top}-Y_i\right\}\bm{Z}_i\widetilde{\bm{R}};\tau\right)\right] + \mathbb{E}\left[\text{T}\left(\left\{\inner{\bm Z_i}{\bm L_*\bm R_*^\top}-Y_i\right\}\bm{Z}_i\widetilde{\bm{R}};\tau\right)\right]\Bigg\}\\
    =: &~ \bm{M}_{L,1} + \bm{M}_{L,2} + \bm{M}_{L,3} + \bm{M}_{L,4}.
  \end{aligned}
\end{equation}~

\noindent\textit{Step 1.} (Moment bounds)

\noindent Denote $\widetilde{\bm{R}}=[\widetilde{\bm{r}}_1,\dots,\widetilde{\bm{r}}_K]$ where each $\widetilde{\bm{r}}_k$ is a $s_R$-sparse vector with unit Euclidean norm. Since $\inner{\bm Z_i}{\bm L_*\bm R_*^\top}-Y_i=-E_i$, denoting the $(j,k)$-th entry of $\left\{\inner{\bm Z_i}{\bm L_*\bm R_*^\top}-Y_i\right\}\bm{Z}_i\widetilde{\bm{R}}$ as $v_{ijk}=-E_i\bm{c}_j^\top\bm{Z}_i\widetilde{\bm{r}}_k$, its $(1+\epsilon)$-th moment is
\begin{equation}\label{eq:moment1}
  \begin{split}
    & \mathbb{E}\left[|v_{ijk}|^{1+\epsilon}\right] = \mathbb{E}\left[\mathbb{E}\left[|E_i\bm{c}_j^\top\bm{Z}_i\widetilde{\bm{r}}_k|^{1+\epsilon}|\bm{X}_i\right]\right]\\
    = & \mathbb{E}\left[\mathbb{E}\left[\left|E_i\right|^{1+\epsilon}|\bm{X}_i\right]\cdot\left|\bm{c}_j^\top\bm{Z}_i\widetilde{\bm{r}}_k\right|^{1+\epsilon}\right]\\
    \leq & M_{e,1+\epsilon}\cdot M_{x,1+\epsilon,s}=M_{\text{eff},1+\epsilon,s}.
  \end{split}
\end{equation}
In addition, the $(j,k)$-th entry of $\bm M_{L,3}$ is
\begin{equation}
  \begin{aligned}
    & \mathbb{E}[\left\{\inner{\bm Z_i}{\bm L_*\bm R_*^\top-\bm L\bm R^\top}\right\}\bm{c}_j^\top\bm{Z}_i\widetilde{\bm{r}}_k]\\
    & + \mathbb{E}[\text{T}\left(\left\{\inner{\bm Z_i}{\bm L\bm R^\top-\bm L_*\bm R_*^\top}-E_i\right\}\bm{c}_j^\top\bm{Z}_i\widetilde{\bm{r}}_k;\tau\right)] - \mathbb{E}[\text{T}(-E_i\bm{c}_j^\top\bm{Z}_i\widetilde{\bm{r}}_k;\tau)]\\
    =: &~ \mathbb{E}[q_{ijk}] + \mathbb{E}[\text{T}(-q_{ijk}+v_{ijk};\tau)] - \mathbb{E}[\text{T}(v_{ijk};\tau)].
  \end{aligned}
\end{equation}
Next, we develop an upper bound for the $(1+\lambda)$-th moment for $q_{ijk}$. By Holder inequality,
\begin{equation}\label{eq:moment2}
  \begin{split}
    & \mathbb{E}\left[|q_{ijk}|^{1+\lambda}\right] \\
    = & \mathbb{E}\left[\left|\text{vec}(\bm Z_i)^\top\text{vec}(\bm L_*\bm R_*^\top-\bm L\bm R^\top)\right|^{1+\lambda}\cdot\left|\text{vec}(\bm{Z}_i)^\top(\widetilde{\bm{r}}_k\otimes\bm{c}_j)\right|^{1+\lambda}\right]\\
    \leq & \mathbb{E}\left[\left|\text{vec}(\bm{Z}_i)^\top\frac{\text{vec}(\bm{L}\bm{R}^\top-\bm{L}_*\bm{R}_*^\top)}{\|\bm{L}\bm{R}^\top-\bm{L}_*\bm{R}_*^\top\|_\text{F}}\right|^{2+2\lambda}\right]^{1/2} \cdot \mathbb{E}\left[|\text{vec}(\bm{Z}_i)^\top(\widetilde{\bm{r}}_k\otimes\bm{c}_j)|^{2+2\lambda}\right]^{1/2}\\
    & \cdot\|\bm{L}\bm{R}^\top-\bm{L}_*\bm{R}_*^\top\|_\text{F}^{1+\lambda}\\
    \leq & M_{x,2+2\lambda}\cdot\|\bm{L}\bm{R}^\top-\bm{L}_*\bm{R}_*^\top\|_\text{F}^{1+\lambda}.
  \end{split}
\end{equation}~

\noindent\textit{Step 2.} (Bound $\|(\bm{M}_{L,1})_{S_1}\|_\text{F}^2$)

\noindent Next, we bound $M_{L,1}$. For any $S_1\subset \{1,2,\dots,d_1\}$, we have
\begin{equation}
  \begin{split} 
    & \|\bm M_{L,1}\|_\text{F}^2\\
    = & \sum_{j\in S_1}\sum_{k=1}^K \left|\mathbb{E}\left[\text{T}\left(\left\{\inner{\bm Z_i}{\bm L_*\bm R_*^\top}-Y_i\right\}\bm{c}_j^\top\bm{X}_i\widetilde{\bm{r}}_k;\tau\right)\right] - \mathbb{E}\left[\left\{\inner{\bm Z_i}{\bm L_*\bm R_*^\top}-Y_i\right\}\bm{c}_j^\top\bm{X}_i\widetilde{\bm{r}}_k\right]\right|^2\\
    = & \sum_{j\in S_1}\sum_{k=1}^K \Big|\mathbb{E}[\text{T}(v_{ijk};\tau)]-\mathbb{E}[v_{ijk}]\Big|^2.
  \end{split}
\end{equation}
For any $\ell\in S_1$ and $k\in[K]$, by the definition of the truncation operator, moment bound in \eqref{eq:moment1}, and Markov's inequality,
\begin{equation}\label{eq:matreg_M1}
  \begin{split}
    & \Big|\mathbb{E}[\text{T}(v_{ijk};\tau)]-\mathbb{E}[v_{ijk}]\Big| \leq \mathbb{E}\Big[|v_{ijk}|\cdot1\{|v_{ijk}|\geq\tau\}\Big]\\
    \leq & \mathbb{E}\Big[|v_{ijk}|^{1+\epsilon}\Big]^{1/(1+\epsilon)}\cdot\mathbb{P}\Big(|v_{ijk}|\geq\tau\Big)^{\epsilon/(1+\epsilon)}\\
    \leq & \mathbb{E}\Big[|v_{ijk}|^{1+\epsilon}\Big]^{1/(1+\epsilon)}\cdot\left(\frac{\mathbb{E}\Big[|v_{ijk}|^{1+\epsilon}\Big]^{1/(1+\epsilon)}}{\tau^{1+\epsilon}}\right)^{\epsilon/(1+\epsilon)} \\
    \leq & M_{\text{eff},1+\epsilon,s}\cdot\tau^{-\epsilon} \asymp \left[\frac{M_{\text{eff},1+\epsilon,s}^{1/\epsilon}\log d}{n}\right]^{\epsilon/(1+\epsilon)},
  \end{split}
\end{equation}
with truncation parameter $\tau\asymp(nM_{\text{eff},1+\epsilon,s}/\log d)^{1/(1+\epsilon)}$. Hence, for any $S_1$ with $|S_1|\leq s_1$,
\begin{equation}
  \|(\bm{M}_{L,1})_{S_1}\|_\text{F}^2 \lesssim s_1K \left[\frac{M_{\text{eff},1+\epsilon,s}^{1/\epsilon}\log d}{n}\right]^{2\epsilon/(1+\epsilon)}.
\end{equation}~

\noindent\textit{Step 3.} (Bound $\|(\bm{M}_{L,2})_{S_1}\|_\text{F}^2$)

\noindent By definition,
\begin{equation}
  \|(\bm{M}_{L,2})_{S_1}\|_\text{F}^2 = \sum_{j\in S_1}\sum_{k=1}^K \left|\frac{1}{n}\sum_{i=1}^n \text{T}(v_{ijk};\tau) - \mathbb{E}[\text{T}(v_{ijk};\tau)]\right|^2.
\end{equation}
For each $i=1,\dots,n$, by the nature of truncation and moment bound in \eqref{eq:moment1}, we have the upper bound for the variance
\begin{equation}
  \text{var}(\text{T}(v_{ijk};\tau)) \leq\mathbb{E}\left[\text{T}(v_{ijk};\tau)^2\right] \leq \tau^{1-\epsilon}\cdot\mathbb{E}\left[|v_{ijk}|^{1+\epsilon}\right]\leq \tau^{1-\epsilon} M_{\text{eff},1+\epsilon,s}.
\end{equation}
Also, for any $q=3,4,\dots$, the higher-order moments satisfy that
\begin{equation}
  \mathbb{E}\Big[\left|\text{T}(v_{ijk};\tau) - \mathbb{E}[\text{T}(v_{ijk};\tau)]\right|^q\Big] \leq (2\tau)^{q-2}\cdot\text{var}(\text{T}(v_{ijk};\tau)).
\end{equation}
By Bernstein's inequality, for any $j\in S_1$ and $k\in[K]$, and $0<t\lesssim \tau^{-\epsilon}M_{\text{eff},1+\epsilon,s}$,
\begin{equation}
  \mathbb{P}\left(\left|\frac{1}{n}\sum_{i=1}^n\text{T}(v_{ijk};\tau) - \mathbb{E}[\text{T}(v_{ijk};\tau)]\right|\geq t\right)\leq 2\exp\left(-\frac{nt^2}{4\tau^{1-\epsilon}M_{\text{eff}}}\right).
\end{equation}
Letting $t=CM_{\text{eff},1+\epsilon,s}^{1/(1+\epsilon)}\log(d)^{\epsilon/(1+\epsilon)}n^{-\epsilon/(1+\epsilon)}$, we have
\begin{equation}\label{eq:matreg_M2_bern}
  \begin{aligned}
      &\mathbb{P}\left(\max_{\substack{j\in[d_1]\\ k\in[K]}}\left|\frac{1}{n}\sum_{i=1}^n\text{T}(v_{ijk};\tau) - \mathbb{E}[\text{T}(v_{ijk};\tau)]\right|\gtrsim \left[\frac{M_{\text{eff},1+\epsilon,s}^{1/\epsilon}\log d}{n}\right]^{\frac{\epsilon}{1+\epsilon}} \right)\\
      \leq & \sum_{j\in S_1}\sum_{k=1}^{K}\mathbb{P}\left(\left|\frac{1}{n}\sum_{i=1}^n\text{T}(v_{ijk};\tau) - \mathbb{E}[\text{T}(v_{ijk};\tau)]\right|\gtrsim \left[\frac{M_{\text{eff},1+\epsilon,s}^{1/\epsilon}\log d}{n}\right]^{\frac{\epsilon}{1+\epsilon}} \right)\\
      \leq & 2s_1K\exp\left(-C\log d\right)\leq C\exp\left(-C\log d\right).
  \end{aligned}
\end{equation}
Hence, with probability at least $1-C\exp(-C\log d)$, for any $S_1$ with $|S_1|\leq s_1$,
\begin{equation}
  \|(\bm{M}_{L,2})_{S_1}\|_\text{F}^2 \lesssim s_1K\left[\frac{M_{\text{eff},1+\epsilon,s}\log d}{n}\right]^{2\epsilon/(1+\epsilon)}.
\end{equation}~

\noindent\textit{Step 4.} (Bound $\|(\bm{M}_{L,3})_{S_1}\|_\text{F}^2$)

\noindent By definition, for any $S_1$,
\begin{equation}
  \|(\bm{M}_{L,3})_{S_1}\|_\text{F}^2 = \sum_{j\in S_1}\sum_{k=1}^K\Big|\mathbb{E}[q_{ijk}] - \mathbb{E}\Big[\text{T}(-q_{ijk}+v_{ijk};\tau) - \text{T}(v_{ijk};\tau)\Big]\Big|^2.
\end{equation}
By the nature of truncation operator, moment conditions in \eqref{eq:moment1} and \eqref{eq:moment2}, and Markov's inequality, we have
\begin{equation}\label{eq:matreg_M3}
  \begin{split}
    & \Big|\mathbb{E}[q_{ijk}] - \mathbb{E}\Big[\text{T}(-q_{ijk}+v_{ijk};\tau) - \text{T}(v_{ijk};\tau)\Big]\Big| \\
    \leq & \Big|\mathbb{E}\Big[q_{ijk}\cdot1\{(|v_{ijk}|\geq\tau)\cup(|-q_{ijk}+v_{ijk}|\geq\tau)\}\Big]\Big|\\
    \leq & \Big|\mathbb{E}\Big[q_{ijk}\cdot1\{(|v_{ijk}|\geq\tau)\cup(|q_{ijk}|\geq\tau/2)\cup(|v_{ijk}|\geq\tau/2)\}\Big]\Big|\\
    \leq & \mathbb{E}\Big[|q_{ijk}|\cdot1\{|q_{ijk}\geq\tau/2|\}\Big] + \mathbb{E}\Big[|q_{ijk}|\cdot1\{|v_{ijk}\geq\tau/2|\}\Big]\\
    \leq &~ \mathbb{E}\left[|q_{ijk}|^{1+\lambda}\right]^{\frac{1}{1+\lambda}}\cdot\left(\frac{\mathbb{E}\Big[|q_{ijk}|^{1+\lambda}\Big]}{\tau^{1+\lambda}}\right)^{\frac{\lambda}{1+\lambda}} + \mathbb{E}\left[|q_{ijk}|^{1+\lambda}\right]^{\frac{1}{1+\lambda}}\cdot\left(\frac{\mathbb{E}\Big[|v_{ijk}|^{1+\epsilon}\Big]}{\tau^{1+\epsilon}}\right)^{\frac{\lambda}{1+\lambda}}\\
    = &~ \mathbb{E}\left[|q_{ijk}|^{1+\lambda}\right]\cdot\tau^{-\lambda}~+~\mathbb{E}\left[|q_{ijk}|^{1+\lambda}\right]^{\frac{1}{1+\lambda}}\cdot\mathbb{E}\left[|v_{ijk}|^{1+\epsilon}\right]^{\frac{\lambda}{1+\lambda}}\cdot\tau^{-\frac{\lambda(1+\epsilon)}{1+\lambda}}\\
    \leq &~ M_{x,2+2\lambda}\cdot\left(\frac{\log d}{nM_{\text{eff},1+\epsilon,s}}\right)^{\frac{\lambda}{1+\epsilon}}\cdot\|\bm{L}\bm{R}^\top-\bm{L}_*\bm{R}_*^\top\|_\text{F}^{1+\lambda}\\
    & + M_{x,2+2\lambda}^{1/(1+\lambda)}\cdot M_{\text{eff},1+\epsilon,s}^{\lambda/(1+\lambda)}\cdot\left(\frac{\log d}{nM_{\text{eff},1+\epsilon,s}}\right)^{\frac{\lambda}{1+\lambda}}\cdot\|\bm{L}\bm{R}^\top-\bm{L}_*\bm{R}_*^\top\|_\text{F}\\
    \lesssim & ~ \bar{\sigma}_K^{\lambda}\left(M_{x,2+2\lambda}M_{\text{eff},1+\epsilon,s}^{-\lambda/(1+\epsilon)}+M_{x,2+2\lambda}^{1/(1+\lambda)}\right)\left(\frac{\log d}{n}\right)^{\min\left(\frac{\lambda}{1+\lambda},\frac{\lambda}{1+\epsilon}\right)}\|\bm{L}\bm{R}^\top-\bm{L}_*\bm{R}_*^\top\|_\text{F},
  \end{split}
\end{equation}
where $\bar{\sigma}_K^{\lambda}:=\max (\sigma_K^\lambda(\bm L_*\bm R_*^\top),1)$.

Therefore, we have
\begin{equation}
  \begin{split}  
    \|(\bm{M}_{L,3})_{S_1}\|_\text{F}^2 & \lesssim s_1K\bar{\sigma}_K^{2\lambda}\left(M_{x,2+2\lambda}^2M_{\text{eff},1+\epsilon,s}^{-2\lambda/(1+\epsilon)}+M_{x,2+2\lambda}^{2/(1+\lambda)}\right)\left(\frac{\log d}{n}\right)^{\min\left(\frac{2\lambda}{1+\lambda},\frac{2\lambda}{1+\epsilon}\right)}\cdot\\
    & \|\bm{L}\bm{R}^\top-\bm{L}_*\bm{R}_*^\top\|_\text{F}^2.
  \end{split}
\end{equation}~

\noindent\textit{Step 5.} (Bound $\|(\bm{M}_{L,4})_{S_1}\|_\text{F}^2$)

\noindent By definition, for any $S_1$,
\begin{equation}
  \|(\bm{M}_{L,4})_{S_1}\|_\text{F}^2 = \sum_{j\in S_1}\sum_{k=1}^K\Bigg|\frac{1}{n}\sum_{i=1}^n\Big[\text{T}(-q_{ijk}+v_{ijk};\tau)-\text{T}(v_{ijk};\tau)\Big] -\mathbb{E}\Big[\text{T}(-q_{ijk}+v_{ijk};\tau)-\text{T}(v_{ijk};\tau)\Big] \Bigg|^2.
\end{equation}
For each $i=1,\dots,n$, we have $|\text{T}(-q_{ijk}+v_{ijk};\tau)-\text{T}(v_{ijk};\tau)|\leq 2\tau$. Meanwhile, by the fact that $|\text{T}(a+b,\tau)-\text{T}(b,\tau)|\leq |a|$ for any $a$ and $b$, we obtain
\begin{equation}
  \begin{split}
    \mathbb{E}\left[\Big(\text{T}(-q_{ijk}+v_{ijk};\tau)-\text{T}(v_{ijk};\tau)\Big)^2\right] & \leq (2\tau)^{1-\lambda}\cdot\mathbb{E}\left[|q_{ijk}|^{1+\lambda}\right]\\
    & \lesssim \tau^{1-\lambda}M_{x,2+2\lambda}\|\bm{L}\bm{R}^\top-\bm{L}_*\bm{R}_*^\top\|_\text{F}^{1+\lambda}=: V^2.
  \end{split}
\end{equation}
In addition, for any $q=3,4,\dots$,
\begin{equation}
  \mathbb{E}\left[\Big(\text{T}(q_{ijk}+v_{ijk};\tau)-\text{T}(v_{ijk};\tau)\Big)^q\right] \leq (2\tau)^{q-2}\mathbb{E}\left[\Big(\text{T}(q_{ijk}+v_{ijk};\tau)-\text{T}(v_{ijk};\tau)\Big)^2\right].
\end{equation}
By Bernstein's inequality, for any $j\in S_1$ and $k\in[K]$,
\begin{equation}\label{eq:M4_bern}
  \begin{split}
    & \mathbb{P}\left(\left|\frac{1}{n}\sum_{i=1}^n\Big[\text{T}(q_{ijk}+v_{ijk};\tau)-\text{T}(v_{ijk};\tau)\Big]-\mathbb{E}\Big[\text{T}(q_{ijk}+v_{ijk};\tau)-\text{T}(v_{ijk};\tau)\Big]\right|\geq t\right)\\
    & \leq 2\exp\left(-\frac{Cnt^2}{V^2+\tau t}\right).
  \end{split}
\end{equation}
Let $t=C_1V\sqrt{\log d/n}+C_2\tau\log d/n$, where $C_1$ and $C_2$ are two positive constants. If $V^2\lesssim \tau t$, the RHS of \eqref{eq:M4_bern} is upper bounded by
$$\begin{aligned}
  \text{RHS} &\leq 2\exp\left(-\frac{Cnt^2}{\tau t}\right)\\
  &=2\exp\left(-\frac{Cn(C_1V\sqrt{\log d/n}+C_2\tau\log d/n)}{\tau}\right)\\
  &\leq 2\exp\left(-\frac{CnC_2\tau\log d/n}{\tau}\right)
  =2\exp\left(-C\log d\right).
\end{aligned}
$$
Conversely, if $V^2\gtrsim \tau t$, the RHS of \eqref{eq:M4_bern} is upper bounded by
$$\begin{aligned}
  \text{RHS} &\leq 2\exp\left(-\frac{Cnt^2}{V^2}\right)\\
  &=2\exp\left(-\frac{Cn(C_1V\sqrt{\log d/n}+C_2\tau\log d/n)^2}{V^2}\right)\\
  &\leq 2\exp\left(-\frac{Cn(C_1V\sqrt{\log d/n})^2}{V^2}\right)=2\exp\left(-C\log d\right).
\end{aligned}
$$
Then, we have
\begin{equation}
  \begin{split} 
    \mathbb{P}\Bigg(&\max_{j\in S_1,k\in[K]}\left|\frac{1}{n}\sum_{i=1}^n\Big[\text{T}(q_{ijk}+v_{ijk};\tau)-\text{T}(v_{ijk};\tau)\Big]-\mathbb{E}\Big[\text{T}(q_{ijk}+v_{ijk};\tau)-\text{T}(v_{ijk};\tau)\Big]\right|\geq t\Bigg) \\ 
    & \lesssim d_1K\exp(-C\log d) \leq C\exp(-C\log d).
  \end{split}
\end{equation}
For the upper bound $t$, plugging in the values of $V^2$ and $\tau$, we have
$$
\begin{aligned}
  t&=C_1\sqrt{\frac{V^2\log d}{n}}+C_2\tau\frac{\log d}{n}\\
  &\asymp M_{x,2+2\lambda}^{1/2}M_{\text{eff},1+\epsilon,s}^{(1-\lambda)/(2+2\epsilon)}\left(\frac{\log d}{n}\right)^\frac{\epsilon+\lambda}{2+2\epsilon}\|\bm L\bm R^\top-\bm L_*\bm R_*^\top\|_\mathrm{F}^\frac{1+\lambda}{2}+M_{\text{eff},1+\epsilon,s}\left(\frac{\log d}{n}\right)^\frac{\epsilon}{1+\epsilon}.
\end{aligned}
$$
Therefore, $\|(\bm{M}_{L,4})_{S_1}\|_\text{F}^2$ is upper bounded as
\begin{equation}\label{eq:matreg_M4}
  \begin{aligned}
      \|(\bm{M}_{L,4})_{S_1}\|_\text{F}^2&\lesssim s_1KM_{x,2+2\lambda}M_{\text{eff},1+\epsilon,s}^{(1-\lambda)/(1+\epsilon)}\left(\frac{\log d}{n}\right)^\frac{\epsilon+\lambda}{1+\epsilon}\|\bm L\bm R^\top-\bm L_*\bm R_*^\top\|_\mathrm{F}^{1+\lambda}\\
      &+s_1KM_{\text{eff},1+\epsilon,s}^2\left(\frac{\log d}{n}\right)^\frac{2\epsilon}{1+\epsilon}.
  \end{aligned}
\end{equation}
Next, we derive the expression of $\phi_{\lambda,\epsilon}$ by partitioning the analysis into two regimes, where the estimation error $\|\bm L\bm R^\top-\bm L_*\bm R_*^\top\|_\mathrm{F}^{1+\lambda}$ is relatively large and small. For
\begin{equation}
  \psi \asymp M_{x,2+2\lambda}^{-1/(1+\lambda)}M_{\text{eff},1+\epsilon,s}^{(1+\lambda+2\epsilon)/[(1+\epsilon)(1+\lambda)]}\left(\frac{\log d}{n}\right)^\frac{\epsilon-\lambda}{(1+\epsilon)(1+\lambda)},
\end{equation}
when $\|\bm L\bm R^\top-\bm L_*\bm R_*^\top\|_\mathrm{F}\lesssim \psi$, i.e., the estimation error is small, \eqref{eq:matreg_M4} is dominated by the second term. By plugging in the expression of $\psi$, \eqref{eq:matreg_M4} gives
\begin{equation}
  \|(\bm{M}_{L,4})_{S_1}\|_\text{F}^2\lesssim s_1KM_{\text{eff},1+\epsilon,s}^2\left(\frac{\log d}{n}\right)^\frac{2\epsilon}{1+\epsilon}.
\end{equation}
Conversely, when $\|\bm L\bm R^\top-\bm L_*\bm R_*^\top\|_\mathrm{F}\gtrsim \psi$, i.e., the estimation error is large, we have that
\begin{equation}
  \|\bm L\bm R^\top-\bm L_*\bm R_*^\top\|_\mathrm{F}^{1+\lambda}=\frac{\|\bm L\bm R^\top-\bm L_*\bm R_*^\top\|_\mathrm{F}^2}{\|\bm L\bm R^\top-\bm L_*\bm R_*^\top\|_\mathrm{F}^{1-\lambda}}\lesssim \psi^{\lambda-1}\|\bm L\bm R^\top-\bm L_*\bm R_*^\top\|_\mathrm{F}^2.
\end{equation}
Then, plugging in the expression of $\psi$, the first term of \eqref{eq:matreg_M4} is upper bounded by
\begin{equation}
  s_1KM_{x,2+2\lambda}^\frac{2}{1+\lambda}M_{\text{eff},1+\epsilon,s}^\frac{2\epsilon(\lambda-1)}{(1+\epsilon)(1+\lambda)}\left(\frac{\log d}{n}\right)^\frac{2\lambda}{1+\lambda}\|\bm L\bm R^\top-\bm L_*\bm R_*^\top\|_\mathrm{F}^2.
\end{equation}
Conbining the two cases, we have the following upper bound for $\|(\bm{M}_{L,4})_{S_1}\|_\text{F}^2$:
\begin{equation}\label{eq:matreg_m4_final}
  \begin{split}
      \|(\bm{M}_{L,4})_{S_1}\|_\text{F}^2&\lesssim s_1KM_{x,2+2\lambda}^\frac{2}{1+\lambda}M_{\text{eff},1+\epsilon,s}^\frac{2\epsilon(\lambda-1)}{(1+\epsilon)(1+\lambda)}\left(\frac{\log d}{n}\right)^\frac{2\lambda}{1+\lambda}\|\bm L\bm R^\top-\bm L_*\bm R_*^\top\|_\mathrm{F}^2\\
      &+s_1KM_{\text{eff},1+\epsilon,s}^2\left(\frac{\log d}{n}\right)^\frac{2\epsilon}{1+\epsilon}.
  \end{split}
\end{equation}

Finally, combining the upper bound of this step and Step 4, and using the fact that 
$$\min\left(\frac{\epsilon}{1+\epsilon},\frac{\lambda}{1+\epsilon},\frac{\lambda}{1+\lambda}\right)=\frac{\min(\epsilon,\lambda)}{1+\epsilon},$$ 
we have
\begin{equation}
\begin{aligned}
    \|(\bm{M}_{L,3})_{S_1}\|_\text{F}^2+\|(\bm{M}_{L,4})_{S_1}\|_\text{F}^2
    \lesssim \phi_{\lambda,\epsilon}\|\bm{L}\bm{R}^\top-\bm{L}_*\bm{R}_*^\top\|_\text{F}^2+ s_1K\left[\frac{M_{\text{eff},1+\epsilon,s}^{1/\epsilon}\log d}{n}\right]^{\frac{2\epsilon}{1+\epsilon}},
\end{aligned}
\end{equation}
where 
$$
\phi_{\lambda,\epsilon}=s_1K\bar{\sigma}_K^{2\lambda}\overline{M}_{x,\text{eff}}^2\left(\frac{\log d}{n}\right)^\frac{2\min(\epsilon,\lambda)}{1+\epsilon},
$$
and 
$$
\overline{M}_{x,\text{eff}}^2=M_{x,2+2\lambda}^2M_{\text{eff},1+\epsilon,s}^\frac{-2\lambda}{1+\epsilon}+M_{x,2+2\lambda}^\frac{2}{1+\lambda}M_{\text{eff},1+\epsilon,s}^\frac{2\epsilon(\lambda-1)}{(1+\epsilon)(1+\lambda)}+M_{x,2+2\lambda}^\frac{2}{1+\lambda}.
$$

In summary, based on the upper bounds in steps 2 to 5, we have
\begin{equation}
  \sum_{j=1}^4\|(\bm{M}_{L,j})_{S_1}\|_\text{F}^2 \lesssim \xi^2_L + \phi_{\lambda,\epsilon}\|\bm{L}\bm{R}^\top-\bm{L}_*\bm{R}_*^\top\|_\text{F}^2,
\end{equation}
where 
$$
  \xi^2_L=s_1K\left[\frac{M_{\text{eff},1+\epsilon,s}^{1/\epsilon}\log d}{n}\right]^{\frac{2\epsilon}{1+\epsilon}}.
$$
\end{proof}

\subsection{Proof of Convergence for Trace Regression}

\begin{proof}[Proof of Theorem \ref{thm:matrix_trace_reg}]

  The main building blocks of the proof, initialization guarantees and robust gradient stability, have been developed in Appendices \ref{append:B.1} and \ref{append:B.2}. Specifically, Proposition \ref{prop:matreg_initial} states that if
  \begin{equation}
    n\gtrsim \mathfrak{C}_1 \cdot \left(s K \right)^{\frac{1+\min(\epsilon,\lambda)}{2\min(\epsilon,\lambda)}}\log d,
  \end{equation}
  where $\mathfrak{C}_1$ is defined in Proposition \ref{prop:matreg_initial}, then the Robust Dantzig Selector satisfies
  \begin{equation}
    \|\bm{\Theta}_{0} - \bm{\Theta}_\ast\|_{\textup{F}}^2 \lesssim \alpha_x\beta_x^{-1} \sigma^2_K,
  \end{equation}
  such that the initialization requirement in Theorem \ref{thm:2} is satisfied.

  In addition, Proposition \ref{prop:stability_linear_reg} states that the robust de-scaled gradient estimators satisfy the following stability conditions:
  \begin{equation}
    \begin{split}
      & \|\bm{G}_L(\bm{L},\bm{R};\tau) - \mathbb{E}[\nabla_{\bm{L}}f(\bm{L},\bm{R};z_i)](\bm{R}^\top\bm{R})^{-1/2}\|_\textup{F}^2 \lesssim \xi^2_L + \phi_{\lambda,\epsilon}\|\bm{L}\bm{R}^\top-\bm{L}_*\bm{R}_*^\top\|_\textup{F}^2,\\
      & \|\bm{G}_R(\bm{L},\bm{R};\tau) - \mathbb{E}[\nabla_{\bm{R}}f(\bm{L},\bm{R};z_i)](\bm{R}^\top\bm{R})^{-1/2}\|_\textup{F}^2 \lesssim \xi^2_L + \phi_{\lambda,\epsilon}\|\bm{L}\bm{R}^\top-\bm{L}_*\bm{R}_*^\top\|_\textup{F}^2,
    \end{split}
  \end{equation}
  where 
  $$
    \phi_{\lambda,\epsilon}=s_1K\bar{\sigma}_K^{2\lambda}\overline{M}_{x,\textup{eff}}^2\left(\frac{\log d}{n}\right)^\frac{2\min(\epsilon,\lambda)}{1+\epsilon},
  $$ 
  $$
    \overline{M}_{x,\text{eff}}^2=M_{x,2+2\lambda}^2M_{\text{eff},1+\epsilon,s}^\frac{-2\lambda}{1+\epsilon}+M_{x,2+2\lambda}^\frac{2}{1+\lambda}M_{\text{eff},1+\epsilon,s}^\frac{2\epsilon(\lambda-1)}{(1+\epsilon)(1+\lambda)}+M_{x,2+2\lambda}^\frac{2}{1+\lambda},
  $$
  $$
    \xi^2_L=s_1K\left[\frac{M_{\textup{eff},1+\epsilon,s}^{1/\epsilon}\log d}{n}\right]^{\frac{2\epsilon}{1+\epsilon}},\quad\text{and}\quad
    \xi^2_R=s_2K\left[\frac{M_{\textup{eff},1+\epsilon,s}^{1/\epsilon}\log d}{n}\right]^{\frac{2\epsilon}{1+\epsilon}}.
  $$

  In the last step, we apply the results above to Theorem \ref{thm:2}. First, we examine the conditions in Theorem \ref{thm:2} hold. Under Assumption \ref{asmp:moment}, by Lemma 3.11 in \citet{bubeck2015convex}, we can show that the RCG condition in Definition \ref{def:1} is directly implied by the restricted strong convexity and smoothness with $\alpha=\alpha_x$ and $\beta=\beta_x$.

  Next, we can show that the de-scaled robust gradient estimators are stable: i.e., with probability at least $1-C\exp(-\log d)$,
  \begin{equation}
    \|\bm{G}_L(\bm{L},\bm{R};\tau)_{S_1}\|_\text{F}^2\lesssim \xi^2_L + \phi_{\lambda,\epsilon}\|\bm{L}\bm{R}^\top-\bm{L}_*\bm{R}_*^\top\|_\text{F}^2.
  \end{equation}
  Similarly, the similar bound holds for $\bm{G}_R(\bm{L},\bm{R};\tau)_{S_2}$,
  \begin{equation}
    \|\bm{G}_R(\bm{L},\bm{R};\tau)_{S_2}\|_\text{F}^2\lesssim \xi^2_R + \phi_{\lambda,\epsilon}\|\bm{L}\bm{R}^\top-\bm{L}_*\bm{R}_*^\top\|_\text{F}^2,
  \end{equation}
  where 
  $$
    \xi^2_R=s_2K\left[\frac{M_{\text{eff},1+\epsilon,s}^{1/\epsilon}\log d}{n}\right]^{\frac{2\epsilon}{1+\epsilon}}.
  $$
  As the sample size satisfies that
  \begin{equation}\begin{aligned}
    &n\gtrsim \left(sK\alpha_x^{-2}\bar{\sigma}^{2\lambda}\overline{M}_{x,\text{eff}}^2\right)^\frac{1+\epsilon}{2\min(\epsilon,\lambda)}\log d,\\
    \text{and}~& n\gtrsim \left(sK\sigma_K^{-2}\alpha_x^{-3}\beta_x\right)^\frac{1+\epsilon}{2\epsilon}M_{\text{eff},1+\epsilon,s}^{1/\epsilon}\log d,
  \end{aligned}
  \end{equation}
  we have $\phi_{\lambda,\epsilon}\lesssim\alpha_x^2$ and $\xi_L^2+\xi_R^2\lesssim \alpha_x^3\beta_x^{-1}\sigma_K^2$, respectively. Hence, the conditions of Theorem \ref{thm:2} hold, which implies that
  \begin{equation}
    d(\bm{F}^{(j)},\bm{F}_*)^2 \leq (1-C\alpha_x\beta_x^{-1})^t d(\bm{F}^{(0)},\bm{F}_*)^2 + C\alpha_x^{-2}sK\left[\frac{M_{\text{eff},1+\epsilon,s}^{1/\epsilon}\log d}{n}\right]^{\frac{2\epsilon}{1+\epsilon}},
  \end{equation}
  and
  \begin{equation}
    \|\bm{\Theta}^{(j)}-\bm{\Theta}_*\|_\text{F}^2 \lesssim (1-C\alpha_x\beta_x^{-1})^t \|\bm{\Theta}^{(0)}-\bm{\Theta}_*\|_\text{F}^2 + C\alpha_x^{-2}sK\left[\frac{M_{\text{eff},1+\epsilon,s}^{1/\epsilon}\log d}{n}\right]^{\frac{2\epsilon}{1+\epsilon}}.
  \end{equation}
  After sufficient iterations, the computational error term is dominated by the statistical error term, implying that
  $$
    \|\wh{\bbm\Theta}-\bm{\Theta}_*\|_\text{F}^2 \lesssim\alpha_x^{-2}sK\left[\frac{M_{\text{eff},1+\epsilon,s}^{1/\epsilon}\log d}{n}\right]^{\frac{2\epsilon}{1+\epsilon}}.
  $$

\end{proof}

\section{Matrix GLMs}\label{append:C}

\subsection{Initialization Guarantees and Gradient Stability}

First, we present the stability of the robust gradient estimators for matrix GLMs. 

\begin{proposition}
  [Stability of Robust De-scaled Gradient Estimators]
  \label{prop:stability_GLM}
  Under the conditions in Theorem \ref{thm:matrix_trace_reg}, with a probability at least $1-C\exp(-C\log d)$, the robust de-scaled gradient estimators satisfy
  \begin{equation}
    \begin{split}
      & \|\bm{G}_L(\bm{L},\bm{R};\tau) - \mathbb{E}[\nabla_{\bm{L}}f(\bm{L},\bm{R};z_i)](\bm{R}^\top\bm{R})^{-1/2}\|_\textup{F}^2 \lesssim \xi^2_L + \phi_{\lambda,\epsilon}\|\bm{L}\bm{R}^\top-\bm{L}_*\bm{R}_*^\top\|_\textup{F}^2,\\
      & \|\bm{G}_R(\bm{L},\bm{R};\tau) - \mathbb{E}[\nabla_{\bm{R}}f(\bm{L},\bm{R};z_i)](\bm{R}^\top\bm{R})^{-1/2}\|_\textup{F}^2 \lesssim \xi^2_L + \phi_{\lambda,\epsilon}\|\bm{L}\bm{R}^\top-\bm{L}_*\bm{R}_*^\top\|_\textup{F}^2,
    \end{split}
  \end{equation}
  where 
  $$
    \phi_{\lambda,\epsilon} = sK\bar{\sigma}_K^{2\lambda}\overline{M}_{x,\textup{eff}}^2\left(\frac{\log d}{n}\right)^\frac{2\min(\epsilon,\lambda)}{1+\epsilon},
  $$ 
  $$
    \overline{M}_{x,\mathrm{eff}}^2=M_{x,2+2\lambda}^2M_{\mathrm{eff},1+\epsilon,s}^\frac{-2\lambda}{1+\epsilon}+M_{x,2+2\lambda}^\frac{2}{1+\lambda}M_{\mathrm{eff},1+\epsilon,s}^\frac{2\epsilon(\lambda-1)}{(1+\epsilon)(1+\lambda)}+M_{x,2+2\lambda}^\frac{2}{1+\lambda},
  $$
  $$
    \xi^2_L=s_1K\left[\frac{M_{\textup{eff},1+\epsilon,s}^{1/\epsilon}\log d}{n}\right]^{\frac{2\epsilon}{1+\epsilon}},\quad\text{and}\quad
    \xi^2_R=s_2K\left[\frac{M_{\textup{eff},1+\epsilon,s}^{1/\epsilon}\log d}{n}\right]^{\frac{2\epsilon}{1+\epsilon}}.
  $$
\end{proposition}

\begin{proof}[Proof of Proposition \ref{prop:stability_GLM}]

For matrix GLMs, we use similar techniques as \ref{append:B.2} to prove that the stability condition holds. Denoting $\bm{Z}_i=\mathcal{P}(\bm{X}_i)$ and $\widetilde{\bm{R}}=\bm{R}(\bm{R}^\top\bm{R})^{-1/2}$, by definition, we have
\begin{equation}
  \begin{split}
    & \bm{G}_L(\bm{L},\bm{R};\tau) -\mathbb{E}[\nabla_{\bm{L}}f(\bm{L},\bm{R};z_i)](\bm{R}^\top\bm{R})^{-1/2}\\
    = & \frac{1}{n}\sum_{i=1}^n\text{T}\left(\left\{g'(\langle\bm{Z}_i,\bm{L}\bm{R}^\top\rangle)-Y_i\right\}\bm{Z}_i\bm{R}(\bm{R}^\top\bm{R})^{-1/2};\tau\right) - \mathbb{E}\left[\{g'(\langle\bm{Z}_i,\bm{L}\bm{R}^\top\rangle)-Y_i\}\bm{Z}_i\widetilde{\bm{R}}\right]\\
    = & \Bigg\{ \mathbb{E}\left[\text{T}\left(\{g'(\langle\bm{Z}_i,\bm{L}_*\bm{R}_*^\top\rangle)-Y_i\}\bm{Z}_i\widetilde{\bm{R}},\tau\right)\right] - \mathbb{E}\left[\{g'(\langle\bm{Z}_i,\bm{L}_*\bm{R}_*^\top\rangle)-Y_i\}\bm{Z}_i\widetilde{\bm{R}}\right]\Bigg\}\\
    + & \Bigg\{\frac{1}{n}\sum_{i=1}^n\text{T}\left(\{g'(\langle\bm{Z}_i,\bm{L}_*\bm{R}_*^\top\rangle)-Y_i\}\bm{Z}_i\widetilde{\bm{R}},\tau\right) - \mathbb{E}\left[\text{T}\left(\{g'(\langle\bm{Z}_i,\bm{L}_*\bm{R}_*^\top\rangle)-Y_i\}\bm{Z}_i\widetilde{\bm{R}},\tau\right)\right] \Bigg\}\\
    + & \Bigg\{ \mathbb{E}\left[\{g'(\langle\bm{Z}_i,\bm{L}_*\bm{R}_*^\top\rangle)-Y_i\}\bm{Z}_i\widetilde{\bm{R}}\right] - \mathbb{E}\left[\{g'(\langle\bm{Z}_i,\bm{L}\bm{R}^\top\rangle)-Y_i\}\bm{Z}_i\widetilde{\bm{R}}\right] \\
    & + \mathbb{E}\left[\text{T}\left(\{g'(\langle\bm{Z}_i,\bm{L}\bm{R}^\top\rangle)-Y_i\}\bm{Z}_i\widetilde{\bm{R}};\tau\right)\right] - \mathbb{E}\left[\text{T}\left(\{g'(\langle\bm{Z}_i,\bm{L}_*\bm{R}_*^\top\rangle)-Y_i\}\bm{Z}_i\widetilde{\bm{R}};\tau\right)\right]\Bigg\}\\
    + & \Bigg\{ \frac{1}{n}\sum_{i=1}^n\text{T}\left(\{g'(\langle\bm{Z}_i,\bm{L}\bm{R}^\top\rangle)-Y_i\}\bm{Z}_i\widetilde{\bm{R}};\tau\right) - \frac{1}{n}\sum_{i=1}^n\text{T}\left(\{g'(\langle\bm{Z}_i,\bm{L}_*\bm{R}_*^\top\rangle)-Y_i\}\bm{Z}_i\widetilde{\bm{R}};\tau\right)\\
    & - \mathbb{E}\left[\text{T}\left(\{g'(\langle\bm{Z}_i,\bm{L}\bm{R}^\top\rangle)-Y_i\}\bm{Z}_i\widetilde{\bm{R}};\tau\right)\right] + \mathbb{E}\left[\text{T}\left(\{g'(\langle\bm{Z}_i,\bm{L}_*\bm{R}_*^\top\rangle)-Y_i\}\bm{Z}_i\widetilde{\bm{R}};\tau\right)\right]\Bigg\}\\
    =: &~ \bm{M}_{L,1} + \bm{M}_{L,2} + \bm{M}_{L,3} + \bm{M}_{L,4}.
  \end{split}
\end{equation}

\noindent\textit{Step 1.} (Moment bounds)

\noindent Denote $\widetilde{\bm{R}}=[\widetilde{\bm{r}}_1,\dots,\widetilde{\bm{r}}_K]$, where each $\widetilde{\bm{r}}_k$ is a $s_R$-sparse vector with unit Euclidean norm. Denoting the $(j,k)$-th entry of $\{g'(\langle\bm{Z}_i,\bm{L}_*\bm{R}_*^\top\rangle)-Y_i\}\bm{Z}_i\widetilde{\bm{R}}$ as $v_{ijk}=\{g'(\langle\bm{Z}_i,\bm{L}_*\bm{R}_*^\top\rangle)-Y_i\}\bm{c}_j^\top\bm{Z}_i\widetilde{\bm{r}}_k$, its $(1+\epsilon)$-th moment is
\begin{equation}\label{eq:GLM_moment1}
  \begin{split}
    & \mathbb{E}\left[|v_{ijk}|^{1+\epsilon}\right] = \mathbb{E}\left[\mathbb{E}\left[|\{g'(\langle\bm{Z}_i,\bm{L}_*\bm{R}_*^\top\rangle)-Y_i\}\bm{c}_j^\top\bm{Z}_i\widetilde{\bm{r}}_k|^{1+\epsilon}|\bm{X}_i\right]\right]\\
    = & \mathbb{E}\left[\mathbb{E}\left[\left|\{g'(\langle\bm{Z}_i,\bm{L}_*\bm{R}_*^\top\rangle)-Y_i\}\right|^{1+\epsilon}|\bm{X}_i\right]\cdot\left|\bm{c}_j^\top\bm{Z}_i\widetilde{\bm{r}}_j\right|^{1+\epsilon}\right]\\
    \leq & M_{e,1+\epsilon}\cdot M_{x,1+\epsilon,s}=M_{\text{eff},1+\epsilon,s}.
  \end{split}
\end{equation}
In addition, the $(j,k)$-th entry of $\bm M_{L,3}$ is
\begin{equation}
  \begin{split}
    & \mathbb{E}[\{g'(\langle\bm{Z}_i,\bm{L}_*\bm{R}_*^\top\rangle) - g'(\langle\bm{Z}_i,\bm{L}\bm{R}^\top\rangle)\}\bm{c}_j^\top\bm{Z}_i\widetilde{\bm{r}}_k]\\
    & + \mathbb{E}[\text{T}(\{g'(\langle\bm{Z}_i,\bm{L}\bm{R}^\top\rangle) - Y_i\rangle)\}\bm{c}_j^\top\bm{Z}_i\widetilde{\bm{r}}_k;\tau)] - \mathbb{E}[\text{T}(\{g'(\langle\bm{Z}_i,\bm{L}_*\bm{R}_*^\top\rangle) - Y_i\rangle)\}\bm{c}_j^\top\bm{Z}_i\widetilde{\bm{r}}_k;\tau)]\\
    =: &~ \mathbb{E}[q_{ijk}] + \mathbb{E}[\text{T}(-q_{ijk}+v_{ijk};\tau)] - \mathbb{E}[\text{T}(v_{ijk};\tau)].
  \end{split}
\end{equation}
Next, we develop an upper bound for $(1+\lambda)$-th moment for $q_{ijk}$. Note that
\begin{equation}\label{eq:GLM_moment2}
  \begin{split}
    & \mathbb{E}\left[|q_{ijk}|^{1+\lambda}\right] = \mathbb{E}\left[\left|g'(\langle\bm{Z}_i,\bm{L}\bm{R}^\top\rangle)-g'(\langle\bm{Z}_i,\bm{L}_*\bm{R}_*^\top\rangle)\right|^{1+\lambda}\cdot\left|\text{vec}(\bm{Z}_i)^\top(\widetilde{\bm{r}}_k\otimes\bm{c}_j)\right|^{1+\lambda}\right]\\
    \leq & \mathbb{E}\left[\left|L\langle\bm{Z}_i,\bm{L}\bm{R}^\top-\bm{L}_*\bm{R}_*^\top\rangle\right|^{1+\lambda}\cdot\left|\text{vec}(\bm{Z}_i)^\top(\widetilde{\bm{r}}_k\otimes\bm{c}_j)\right|^{1+\lambda}\right]\\
    \leq & L^{1+\lambda}\cdot\mathbb{E}\left[|\text{vec}(\bm{Z}_i)^\top(\widetilde{\bm{r}}_k\otimes\bm{c}_j)|^{2+2\lambda}\right]^{1/2}\cdot\mathbb{E}\left[\left|\text{vec}(\bm{Z}_i)^\top\frac{\text{vec}(\bm{L}\bm{R}^\top-\bm{L}_*\bm{R}_*^\top)}{\|\bm{L}\bm{R}^\top-\bm{L}_*\bm{R}_*^\top\|_\text{F}}\right|^{2+2\lambda}\right]^{1/2}\\
    & \cdot\|\bm{L}\bm{R}^\top-\bm{L}_*\bm{R}_*^\top\|_\text{F}^{1+\lambda}\\
    \lesssim & M_{x,2+2\lambda}\cdot\|\bm{L}\bm{R}^\top-\bm{L}_*\bm{R}_*^\top\|_\text{F}^{1+\lambda}.
  \end{split}
\end{equation}~

\noindent\textit{Step 2.} (Bound $\|(\bm{M}_{L,1})_{S_1}\|_\text{F}^2$)

\noindent For $\bm M_{L,1}$ and any $S_1$, we have
\begin{equation}
  \begin{split} 
    \|\bm M_{L,1}\|_\text{F}^2=  \sum_{j\in S_1}\sum_{k=1}^K \Big|\mathbb{E}[\text{T}(v_{ijk};\tau)]-\mathbb{E}[v_{ijk}]\Big|^2.
  \end{split}
\end{equation}
For any $\ell\in S_1$ and $k\in[K]$, by similar argument as \eqref{eq:matreg_M1} and the moment bound in \eqref{eq:GLM_moment1}, we have
\begin{equation}
  \begin{split}
    \Big|\mathbb{E}[\text{T}(v_{ijk};\tau)]-\mathbb{E}[v_{ijk}]\Big| \lesssim  M_{\text{eff},1+\epsilon,s}\cdot\tau^{-\epsilon} \asymp \left[\frac{M_{\text{eff},1+\epsilon,s}^{1/\epsilon}\log d}{n}\right]^{\epsilon/(1+\epsilon)},
  \end{split}
\end{equation}
with truncation parameter $\tau\asymp(nM_{\text{eff},1+\epsilon,s}/\log d)^{1/(1+\epsilon)}$. Hence, for any $S_1$ with $|S_1|\leq s_1$,
\begin{equation}
  \|(\bm{M}_{L,1})_{S_1}\|_\text{F}^2 \lesssim s_1K \left[\frac{M_{\text{eff},1+\epsilon,s}^{1/\epsilon}\log d}{n}\right]^{2\epsilon/(1+\epsilon)}.
\end{equation}~

\noindent\textit{Step 3.} (Bound $\|(\bm{M}_{L,2})_{S_1}\|_\text{F}^2$)

\noindent By definition,
\begin{equation}
  \|(\bm{M}_{L,2})_{S_1}\|_\text{F}^2 = \sum_{j\in S_1}\sum_{k=1}^K \left|\frac{1}{n}\sum_{i=1}^n \text{T}(v_{ijk};\tau) - \mathbb{E}[\text{T}(v_{ijk};\tau)]\right|^2.
\end{equation}
For each $i=1,\dots,n$, by the nature of truncation and moment bound in \eqref{eq:GLM_moment1}, we have the upper bound for the variance
\begin{equation}
  \text{var}(\text{T}(v_{ijk};\tau)) \leq \mathbb{E}\left[\text{T}(v_{ijk};\tau)^2\right] \leq \tau^{(1-\epsilon)}\cdot\mathbb{E}\left[|v_{ijk}|^{1+\epsilon}\right]\lesssim\tau^{1-\epsilon}M_{\text{eff},1+\epsilon,s}.
\end{equation}
By Bernstein's inequality, for any $j\in S_1$ and $k\in[K]$, and $0<t\lesssim \tau^{-\epsilon}M_{\text{eff},1+\epsilon,s}$,
\begin{equation}
  \mathbb{P}\left(\left|\frac{1}{n}\sum_{i=1}^n\text{T}(v_{ijk};\tau) - \mathbb{E}[\text{T}(v_{ijk};\tau)]\right|\geq t\right)\leq 2\exp\left(-\frac{nt^2}{4\tau^{1-\epsilon}M_{\text{eff}}}\right).
\end{equation}
Letting $t=CM_{\text{eff},1+\epsilon,s}^{1/(1+\epsilon)}\log(d)^{\epsilon/(1+\epsilon)}n^{-\epsilon/(1+\epsilon)}$, similar to \eqref{eq:matreg_M2_bern}, we have
\begin{equation}
  \mathbb{P}\left(\max_{\substack{j\in[d_1]\\ k\in[K]}}\left|\frac{1}{n}\sum_{i=1}^n\text{T}(v_{ijk};\tau) - \mathbb{E}[\text{T}(v_{ijk};\tau)]\right|\gtrsim \left[\frac{M_{\text{eff},1+\epsilon,s}^{1/\epsilon}\log d}{n}\right]^{\frac{\epsilon}{1+\epsilon}} \right)\leq C\exp\left(-C\log d\right).
\end{equation}
Hence, with probability at least $1-C\exp(-C\log d)$, for any $S_1$ with $|S_1|\leq s_1$,
\begin{equation}
  \|(\bm{M}_{L,2})_{S_1}\|_\text{F}^2 \lesssim s_1K\left[\frac{M_{\text{eff},1+\epsilon,s}^{1/\epsilon}\log d}{n}\right]^{2\epsilon/(1+\epsilon)}.
\end{equation}~

\noindent\textit{Step 4.} (Bound $\|(\bm{M}_{L,3})_{S_1}\|_\text{F}^2$)

\noindent By definition, for any $S_1$,
\begin{equation}
  \|(\bm{M}_{L,3})_{S_1}\|_\text{F}^2 = \sum_{j\in S_1}\sum_{k=1}^K\Big|\mathbb{E}[q_{ijk}] - \mathbb{E}\Big[\text{T}(q_{ijk}+v_{ijk};\tau) - \text{T}(v_{ijk};\tau)\Big]\Big|^2.
\end{equation}
Similar to \eqref{eq:matreg_M3}, by the nature of truncation operator, moment conditions in \eqref{eq:GLM_moment1} and \eqref{eq:GLM_moment2}, and Markov's inequality,  we have
\begin{equation}
  \begin{split}
    & \Big|\mathbb{E}[q_{ijk}] - \mathbb{E}\Big[\text{T}(q_{ijk}+v_{ijk};\tau) - \text{T}(v_{ijk};\tau)\Big]\Big| \\
    \leq &~ \mathbb{E}\left[|q_{ijk}|^{1+\lambda}\right]\cdot\tau^{-\lambda}~+~\mathbb{E}\left[|q_{ijk}|^{1+\lambda}\right]^{\frac{1}{1+\lambda}}\cdot\mathbb{E}\left[|v_{ijk}|^{1+\epsilon}\right]^{\frac{\lambda}{1+\lambda}}\cdot\tau^{-\frac{\lambda(1+\epsilon)}{1+\lambda}}\\
    \leq &~ M_{x,2+2\lambda}\cdot\left(\frac{\log d}{nM_{\text{eff},1+\epsilon,s}}\right)^{\frac{\lambda}{1+\epsilon}}\cdot\|\bm{L}\bm{R}^\top-\bm{L}_*\bm{R}_*^\top\|_\text{F}^{1+\lambda}\\
    & + M_{x,2+2\lambda}^{1/(1+\lambda)}\cdot M_{\text{eff},1+\epsilon,s}^{\lambda/(1+\lambda)}\cdot\left(\frac{\log d}{nM_{\text{eff},1+\epsilon,s}}\right)^{\frac{\lambda}{1+\lambda}}\cdot\|\bm{L}\bm{R}^\top-\bm{L}_*\bm{R}_*^\top\|_\text{F}\\
    \lesssim & ~ \bar{\sigma}_K^{\lambda}\left(M_{x,2+2\lambda}M_{\text{eff},1+\epsilon,s}^{-\lambda/(1+\epsilon)}+M_{x,2+2\lambda}^{1/(1+\lambda)}\right)\left(\frac{\log d}{n}\right)^{\min\left(\frac{\lambda}{1+\lambda},\frac{\lambda}{1+\epsilon}\right)}\|\bm{L}\bm{R}^\top-\bm{L}_*\bm{R}_*^\top\|_\text{F},
  \end{split}
\end{equation}
where $\bar{\sigma}_K=\max (\sigma_K(\bm L_*\bm R_*^\top),1)$.

Therefore, we have
\begin{equation}
  \begin{split}  
    \|(\bm{M}_{L,3})_{S_1}\|_\text{F}^2 & \lesssim s_1K\bar{\sigma}_K^{2\lambda}\left(M_{x,2+2\lambda}^2M_{\text{eff},1+\epsilon,s}^{-2\lambda/(1+\epsilon)}+M_{x,2+2\lambda}^{2/(1+\lambda)}\right)\left(\frac{\log d}{n}\right)^{\min\left(\frac{2\lambda}{1+\lambda},\frac{2\lambda}{1+\epsilon}\right)}\|\bm{L}\bm{R}^\top-\bm{L}_*\bm{R}_*^\top\|_\text{F}^2.
  \end{split}
\end{equation}~

\noindent\textit{Step 5.} (Bound $\|(\bm{M}_{L,4})_{S_1}\|_\text{F}^2$)

\noindent By definition, for any $S_1$,
\begin{equation}
  \|(\bm{M}_{L,4})_{S_1}\|_\text{F}^2 = \sum_{j\in S_1}\sum_{k=1}^K\Bigg|\frac{1}{n}\sum_{i=1}^n\Big[\text{T}(-q_{ijk}+v_{ijk};\tau)-\text{T}(v_{ijk};\tau)\Big] -\mathbb{E}\Big[\text{T}(-q_{ijk}+v_{ijk};\tau)-\text{T}(v_{ijk};\tau)\Big] \Bigg|^2.
\end{equation}
We derive the upper bound using the same techniques as in Step 5 of Appendix \ref{append:B.2}. For each $i=1,\dots,n$, we have $|\text{T}(-q_{ijk}+v_{ijk};\tau)-\text{T}(v_{ijk};\tau)|\leq 2\tau$. Meanwhile, for the variance, we obtain
\begin{equation}
  \begin{split}
    \mathbb{E}\left[\Big(\text{T}(-q_{ijk}+v_{ijk};\tau)-\text{T}(v_{ijk};\tau)\Big)^2\right] & \leq (2\tau)^{1-\lambda}\cdot\mathbb{E}\left[|q_{ijk}|^{1+\lambda}\right]\\
    & \lesssim \tau^{1-\lambda}M_{x,2+2\lambda}\|\bm{L}\bm{R}^\top-\bm{L}_*\bm{R}_*^\top\|_\text{F}^{1+\lambda}=: V^2.
  \end{split}
\end{equation}
By Bernstein's inequality, for any $j\in S_1$ and $k\in[K]$,
\begin{equation}\label{eq:GLM_M4_bern}
  \begin{split}
    & \mathbb{P}\left(\left|\frac{1}{n}\sum_{i=1}^n\Big[\text{T}(q_{ijk}+v_{ijk};\tau)-\text{T}(v_{ijk};\tau)\Big]-\mathbb{E}\Big[\text{T}(q_{ijk}+v_{ijk};\tau)-\text{T}(v_{ijk};\tau)\Big]\right|\geq t\right)\\
    & \leq 2\exp\left(-\frac{Cnt^2}{V^2+\tau t}\right).
  \end{split}
\end{equation}
Let $t=C_1V\sqrt{\log d/n}+C_2\tau\log d/n$, where $C_1$ and $C_2$ are two positive constants. If $V^2\lesssim \tau t$, the RHS of \eqref{eq:GLM_M4_bern} is upper bounded by
$$\begin{aligned}
  \text{RHS} &\leq 2\exp\left(-\frac{Cnt^2}{\tau t}\right)\\
  &=2\exp\left(-\frac{Cn(C_1V\sqrt{\log d/n}+C_2\tau\log d/n)}{\tau}\right) \leq 2\exp\left(-C\log d\right).
\end{aligned}
$$
Conversely, if $V^2\gtrsim \tau t$, the RHS of \eqref{eq:GLM_M4_bern} is upper bounded by
$$\begin{aligned}
  \text{RHS} &\leq 2\exp\left(-\frac{Cnt^2}{V^2}\right)\\
  &=2\exp\left(-\frac{Cn(C_1V\sqrt{\log d/n}+C_2\tau\log d/n)^2}{V^2}\right)=2\exp\left(-C\log d\right).
\end{aligned}
$$
Then, we have
\begin{equation}
  \begin{split} 
    \mathbb{P}\Bigg(&\max_{j\in S_1,k\in[K]}\left|\frac{1}{n}\sum_{i=1}^n\Big[\text{T}(q_{ijk}+v_{ijk};\tau)-\text{T}(v_{ijk};\tau)\Big]-\mathbb{E}\Big[\text{T}(q_{ijk}+v_{ijk};\tau)-\text{T}(v_{ijk};\tau)\Big]\right|\geq t\Bigg) \\ 
    & \lesssim d_1K\exp(-C\log d) \leq C\exp(-C\log d).
  \end{split}
\end{equation}
For the upper bound $t$, plugging in the values of $V^2$ and $\tau$, we have
\begin{equation}
  \begin{aligned}
    t&=C_1\sqrt{\frac{V^2\log d}{n}}+C_2\tau\frac{\log d}{n}\\
    & \lesssim M_{x,2+2\lambda}^{1/2}M_{\text{eff},1+\epsilon,s}^{(1-\lambda)/(2+2\epsilon)}\left(\frac{\log d}{n}\right)^\frac{\min(\epsilon,\lambda)}{1+\epsilon}\|\bm L\bm R^\top-\bm L_*\bm R_*^\top\|_\mathrm{F}^\frac{1+\lambda}{2}+M_{\text{eff},1+\epsilon,s}\left(\frac{\log d}{n}\right)^\frac{\epsilon}{1+\epsilon}.
  \end{aligned}
\end{equation}
By employing an argument analogous to the derivation of \eqref{eq:matreg_m4_final}, we have
\begin{equation}
  \begin{split}
      \|(\bm{M}_{L,4})_{S_1}\|_\text{F}^2&\lesssim s_1KM_{x,2+2\lambda}^\frac{2}{1+\lambda}M_{\text{eff},1+\epsilon,s}^\frac{2\epsilon(\lambda-1)}{(1+\epsilon)(1+\lambda)}\left(\frac{\log d}{n}\right)^\frac{2\lambda}{1+\lambda}\|\bm L\bm R^\top-\bm L_*\bm R_*^\top\|_\mathrm{F}^2\\
      &+s_1KM_{\text{eff},1+\epsilon,s}^2\left(\frac{\log d}{n}\right)^\frac{2\epsilon}{1+\epsilon}.
  \end{split}
\end{equation}

Finally, combining the upper bounds of this step and Steps 2, 3, and 4, we have
\begin{equation}
\begin{aligned}
    \sum_{j=1}^4\|(\bm{M}_{L,j})_{S_1}\|_\text{F}^2
    \lesssim \phi_{\lambda,\epsilon}\|\bm{L}\bm{R}^\top-\bm{L}_*\bm{R}_*^\top\|_\text{F}^2+ s_1K\left[\frac{M_{\text{eff},1+\epsilon,s}^{1/\epsilon}\log d}{n}\right]^{\frac{2\epsilon}{1+\epsilon}},
\end{aligned}
\end{equation}
where 
$$\begin{aligned}
\phi_{\lambda,\epsilon}&=s_1K\bar{\sigma}_K^{2\lambda}\overline{M}_{x,\text{eff}}^2\left(\frac{\log d}{n}\right)^\frac{2\min(\epsilon,\lambda)}{1+\epsilon},\\
\xi^2_L&=s_1K\left[\frac{M_{\text{eff},1+\epsilon,s}^{1/\epsilon}\log d}{n}\right]^{\frac{2\epsilon}{1+\epsilon}},
\end{aligned}
$$
and 
$$
  \overline{M}_{x,\mathrm{eff}}^2=M_{x,2+2\lambda}^2M_{\mathrm{eff},1+\epsilon,s}^\frac{-2\lambda}{1+\epsilon}+M_{x,2+2\lambda}^\frac{2}{1+\lambda}M_{\mathrm{eff},1+\epsilon,s}^\frac{2\epsilon(\lambda-1)}{(1+\epsilon)(1+\lambda)}+M_{x,2+2\lambda}^\frac{2}{1+\lambda}.
$$

Analogously, we can obtain the same result for $\bm M_{R,j}$ except that $\xi^2_R$ is defined with respect to $s_2$.

\end{proof}

Finally, applying the same techniques in Appendix \ref{append:B.1} and utilizing the results in Theorem 2 of \citet{zhu2021taming}, we have the following initialization guarantees. For brevity, the proof is omitted.

\begin{proposition}\label{prop:logistic_initial}
  For the robust initialization via robust lasso, if the tuning parameters are selected as 
  \begin{equation}
      \tau_x \asymp \left(\frac{nM_{x,2+2\lambda,1}}{\log d}\right)^{\frac{1}{1+\lambda}},\quad
      \text{and}\quad R\asymp \left(\frac{M_{x,2+2\lambda,1}^{1/\lambda}\log d}{n}\right)^{\frac{1}{2}},
  \end{equation}
  and the sample size satisfies 
  \begin{equation}
    \begin{split}
      n & \gtrsim \sigma_K^{-\frac{1+\lambda}{\lambda}}\left(\beta_x/\alpha_x\right)^{\frac{1+\lambda}{2\lambda}}(s_\ast K)^{\frac{1+\lambda}{2\lambda}}M_{x,2+2\lambda,1}^\frac{1}{\lambda}\log d =: \mathfrak{C}_2 \cdot \left(s K \right)^{\frac{1+\lambda}{2\lambda}}\log d,
    \end{split}
  \end{equation}
  where $\delta=\min(\lambda,\epsilon)$, and $M_{yx,1+\delta}=2M_{\mathrm{eff},1+\delta,1}+2M_{x,2+2\delta}\|\bbm{\theta}_\ast\|_2^{1+\delta}$. Then, with probability at least $1-C\exp(-C\log d)$, the initial estimator satisfies
  \begin{equation}
    \|\bm{\Theta}_{0} - \bm{\Theta}_\ast\|_{\textup{F}}^2 \lesssim \alpha_x\beta_x^{-1} \sigma^2_K.
  \end{equation}
\end{proposition}

\subsection{Statistical Rates of GLM and Logistic Regression}

\begin{proof}[Proof of Theorem \ref{thm:matrixGLM}]
  In the last step, we examine the conditions and apply Theorem \ref{thm:2}. Similarly, under Assumption \ref{asmp:moment}, by Lemma 3.11 in \citet{bubeck2015convex}, we can show that the RCG condition in Definition \ref{def:1} is implied with $\alpha=\alpha_x$ and $\beta=\beta_x$.

  The initialization guarantees are provided in Proposition \ref{prop:logistic_initial}. In addition, by Proposition \ref{prop:stability_GLM}, the de-scaled robust gradient estimators are stable, that is, with probability at least $1-C\exp(-\log d)$,
\begin{equation}
  \begin{aligned}
    \|\bm{G}_L(\bm{L},\bm{R};\tau)_{S_1}\|_\text{F}^2\lesssim s_1K\left[\frac{M_{\text{eff},1+\epsilon,s}^{1/\epsilon}\log d}{n}\right]^{\frac{2\epsilon}{1+\epsilon}} + \phi_{\lambda,\epsilon}\|\bm{L}\bm{R}^\top-\bm{L}_*\bm{R}_*^\top\|_\text{F}^2,\\
  \text{and} \quad\|\bm{G}_R(\bm{L},\bm{R};\tau)_{S_2}\|_\text{F}^2\lesssim s_2K\left[\frac{M_{\text{eff},1+\epsilon,s}^{1/\epsilon}\log d}{n}\right]^{\frac{2\epsilon}{1+\epsilon}} + \phi_{\lambda,\epsilon}\|\bm{L}\bm{R}^\top-\bm{L}_*\bm{R}_*^\top\|_\text{F}^2.
  \end{aligned}
\end{equation}

As the sample size satisfies that
\begin{equation}
  \begin{aligned}
    &n\gtrsim \left(sK\alpha_x^{-2}\bar{\sigma}^{2\lambda}\overline{M}_{x,\text{eff}}^2\right)^\frac{1+\epsilon}{2\min(\epsilon,\lambda)}\log d,\\
   \text{and}~& n\gtrsim \left(sK\sigma_K^{-2}\alpha_x^{-3}\beta_x\right)^\frac{1+\epsilon}{2\epsilon}M_{\text{eff},1+\epsilon,s}^{1/\epsilon}\log d,
  \end{aligned}
\end{equation}
we have $\phi_{\lambda,\epsilon}\lesssim\alpha_x^2$ and $\xi^2_L+\xi^2_R\leq\alpha_x^3\beta_x^{-1}\sigma_K^2$, respectively. Therefore, applying Theorem \ref{thm:2}, we have
\begin{equation}
  \|\bm{\Theta}^{(j)}-\bm{\Theta}_*\|_\text{F}^2\asymp d(\bm{F}^{(j)},\bm{F}_*)^2 \leq (1-C\alpha_x\beta_x^{-1})^t d(\bm{F}^{(0)},\bm{F}_*)^2 + C\alpha_x^{-2}sK\left[\frac{M_{\text{eff},1+\epsilon,s}^{1/\epsilon}\log d}{n}\right]^{\frac{2\epsilon}{1+\epsilon}}.
\end{equation}
After sufficient iterations, we finally obtain that
$$
  \|\bm{\Theta}^{(j)}-\bm{\Theta}_*\|_\text{F}^2 \lesssim  \alpha_x^{-2}sK\left[\frac{M_{\text{eff},1+\epsilon,s}^{1/\epsilon}\log d}{n}\right]^{\frac{2\epsilon}{1+\epsilon}}.
$$
\end{proof}

Next, we are ready to extend the proof to Corollary \ref{cor:logistic_rates}

\begin{proof}[Proof of Corollary \ref{cor:logistic_rates}]

  The proof of Corollary \ref{cor:logistic_rates} directly follows Theorem \ref{thm:matrixGLM} with $L=1/4$ in Assumption \ref{asmp:Lipschitz} and $\epsilon=1$ in Assumption \ref{asmp:moment}, as $Y_i$ is bounded in matrix logistic regression. Hence, the details are omitted.

\end{proof}

\section{Bilinear Model}\label{append:D}

\subsection{Initialization Guarantees}\label{append:D.1}

\begin{proposition}[Rates of Initialization]
  \label{prop:bilinear_init}
  For the robust Dantzig selector, if the tuning parameters are selected as 
  \begin{equation}
    \begin{split}
      \tau_x \asymp & \left(\frac{nM_{x,2+2\lambda,1}}{\log d}\right)^{\frac{1}{1+\lambda}}, \quad 
      \tau_{yx}\asymp\left(\frac{nM_{yx,1+\delta}}{\log d}\right)^{\frac{1}{1+\delta}},\\
      \text{and}\quad R\asymp & \|\bbm{\Theta}_\ast\|_1 \cdot \left(\frac{M_{x,2+2\lambda,1}^{1/\lambda}\log d}{n}\right)^{\frac{\lambda}{1+\lambda}} + \left(\frac{M_{yx,1+\delta}^{1/\delta}\log d}{n}\right)^{\frac{\delta}{1+\delta}},
    \end{split}
  \end{equation}
  and the sample size satisfies 
\begin{equation}
  n \gtrsim \|(\bm{\Sigma}_x^{-1})^\top\|_{1,\infty}^{\frac{1+\delta}{\delta}}\sigma_1^{-\frac{1+\delta}{2\delta}}\left(\beta_x/\alpha_x\right)^{\frac{1+\delta}{2\delta}}s_\ast^{\frac{1+\delta}{2\delta}}\left(M_{x,2+2\lambda,1}^{1/(1+\lambda)}\|\bbm{\Theta}_\ast\|_1+M_{yx,1+\delta}^{1/(1+\lambda)}\right)^{\frac{1+\delta}{\delta}}\log d,
\end{equation}
  where $\delta=\min(\lambda,\epsilon)$, and $M_{yx,1+\delta}=2M_{\mathrm{eff},1+\delta,1}+2M_{x,2+2\delta}\|\bbm{\Theta}_\ast\|_{2,\infty}^{1+\delta}$. Then, with probability at least $1-C\exp(-C\log d)$, the initial estimator satisfies
  \begin{equation}
    \|\bm{\Theta}_{0} - \bm{\Theta}_\ast\|_{\textup{F}}^2 \lesssim \alpha_x\beta_x^{-1} \sigma^2_K.
  \end{equation}
\end{proposition}

\begin{proof}[Proof of Proposition \ref{prop:bilinear_init}]

\noindent\textit{Step 1.} (Error bounds of robust Dantzig Selector)

\noindent Similar to Appendix \ref{append:B.1}, we firstly show that when
\begin{equation}
  \|\widehat{\bm{\Sigma}}_x - \bm{\Sigma}_x\|_\infty\leq\zeta_0\quad\text{and}\quad \|\widehat{\bm{\Sigma}}_{yx} - \bm{\Sigma}_{yx}\|_\infty\leq\zeta_1,
\end{equation}
and $R\geq\|\bbm{\Theta}_*\|_1\zeta_0+\zeta_1$, then the error $\|\widehat{\bbm{\Theta}}-\bbm{\Theta}_*\|_\text{F}$ can be bounded in terms of multiple matrix norms.

Note that
\begin{equation}
  \begin{split}
    & \|\widehat{\bm \Sigma}_{yx}-\bbm\Theta_*\widehat{\bm\Sigma}_x\|_\infty\\
    = & \|\widehat{\bm \Sigma}_{yx} - \bm \Sigma_{yx} + \bbm\Theta_*\bm\Sigma_x -\bbm\Theta_*\widehat{\bm\Sigma}_x\|_\infty \\
    \leq & \|\widehat{\bm \Sigma}_{yx} - \bm \Sigma_{yx}\|_\infty + \|\bbm\Theta_*\|_{1,\infty}\|\bm\Sigma_x -\widehat{\bm\Sigma}_x\|_\infty \\ 
    \leq & \zeta_1 + \|\bbm{\Theta}_*\|_1\zeta_0 \leq R.
  \end{split}
\end{equation}
Therefore, $\bbm{\Theta}_*$ is feasible in the optimization constraint, and $\|\widehat{\bbm\Theta}\|_1\leq \|\bbm{\Theta}_*\|_1$. Then, by triangle inequality, we have
\begin{equation}
  \begin{split}
    & \|\widehat{\bbm{\Theta}} - \bbm\Theta_*\|_\infty \\
    = & \left\|\left(\widehat{\bbm\Theta}\bbm\Sigma_{x}-\widehat{\bbm\Theta}\widehat{\bbm\Sigma}_{x}+\widehat{\bbm\Theta}\widehat{\bbm\Sigma}_{x}-\widehat{\bbm\Sigma}_{yx}+\widehat{\bbm\Sigma}_{yx}-\bbm\Sigma_{yx}\right)\bbm\Sigma_{x}^{-1}\right\|_\infty\\
    \leq & \|\widehat{\bbm\Theta}\|_{1,\infty}\|\bbm\Sigma_{x}-\widehat{\bbm\Sigma}_{x}\|_\infty\|(\bbm\Sigma_{x}^{-1})^\top\|_{1,\infty}+\|\widehat{\bbm\Theta}\widehat{\bbm\Sigma}_{x}-\widehat{\bbm\Sigma}_{yx}\|_\infty\|(\bbm\Sigma_{x}^{-1})^\top\|_{1,\infty}\\
    &+\|\widehat{\bbm\Sigma}_{yx}-\bbm\Sigma_{yx}\|_\infty\|(\bbm\Sigma_{x}^{-1})^\top\|_{1,\infty}\\
    \leq &\left(\|\widehat{\bbm\Theta}\|_1\zeta_0+R+\zeta_1\right)\|(\bbm\Sigma_{x}^{-1})^\top\|_{1,\infty}\\
    \leq & \left(\|\bbm\Theta_*\|_1\zeta_0+R+\zeta_1\right)\|(\bbm\Sigma_{x}^{-1})^\top\|_{1,\infty}\leq 2R\|(\bbm\Sigma_{x}^{-1})^\top\|_{1,\infty}.
  \end{split}
\end{equation}

Denote the nonzero support of $\bbm\Theta_*$ to be
$$S=\{(j,k)\in\{1,\dots,p_1p_2\}\times \{1,\dots,q_1q_2\}: (\bbm\Theta_*)_{jk}\neq 0\}.$$
Then, $|S|=s_{L,*}\cdot s_{R,*}:=s_*$. By similar argument as \eqref{eq:theta1_decmopose}, we have
\begin{equation}
  \|(\widehat{\bbm{\Theta}} - \bbm{\Theta}_*)_{S^\perp}\|_1 \leq \|(\widehat{\bbm{\Theta}} - \bbm{\Theta}_*)_S\|_1.
\end{equation}

Hence, the upper bound in terms of $\|\cdot\|_{L_1}$ can be bounded as
\begin{equation}
  \begin{split}
    & \|\widehat{\bbm{\Theta}} - \bbm{\Theta}_*\|_1 = \|(\widehat{\bbm{\Theta}} - \bbm{\Theta}_*)_S\|_1 + \|(\widehat{\bbm{\Theta}} - \bbm{\Theta}_*)_{S^\perp}\|_1\\
    \leq & 2\|(\widehat{\bbm{\Theta}} - \bbm{\Theta}_*)_S\|_1 \leq 2s_*\|\widehat{\bbm{\Theta}} - \bbm{\Theta}_*\|_\infty.
  \end{split}
\end{equation}
Finally, by duality and $R\asymp \zeta_0\|\bbm\Theta_*\|_1+\zeta_1$, we have
\begin{equation}\label{eq:Bilinear_init_bound}
  \begin{split}
    \|\widehat{\bbm{\Theta}} - \bbm{\Theta}_*\|_\text{F} \leq & \|\widehat{\bbm{\Theta}} - \bbm{\Theta}_*\|_1^{1/2}\|\widehat{\bbm{\Theta}} - \bbm{\Theta}_*\|_\infty^{1/2}\\
    \leq& \sqrt{8s_*}R\|(\bm{\Sigma}_x^{-1})^\top\|_{1,\infty}\\
    \lesssim & \|(\bm{\Sigma}_x^{-1})^\top\|_{1,\infty}\sqrt{8s_*}(\zeta_0\|\bbm{\Theta_*}\|_1+\zeta_1).
  \end{split}
\end{equation}~

\noindent\textit{Step 2.} (Accuracy of robust covariance estimation)

\noindent In this step, we derive the error bounds $\zeta_0$ and $\zeta_1$. Denote $\bm x_i=\text{Vec}(\bm X_i^\top)$, $\bm y_i=\text{Vec}(\bm Y_i^\top)$, and $\bm e_i=\text{Vec}(\bm E_i^\top)$, respectively. For each entry of $\bm{x}_i\bm{x}_i^\top$, it has a finite $(1+\lambda)$-th moment $M_{x,2+2\lambda,1}$. For $\bm y_i\bm{x}_i^\top$, note that the $(j,k)$-th entry is 
\begin{equation}
  y_{ij}x_{ik}=\bm c_j^\top\bm e_i \bm x_i^\top \bm c_k+\bm c_j^\top\bbm\Theta_*\bm x_i\bm x_{i}^\top\bm c_k.
\end{equation}
Denoting $\delta=\min (\epsilon,\lambda)$, the first term has a finite $(1+\delta)$-th moment $M_{\text{eff},1+\delta,1}$. For the second term, we have
\begin{equation}
  \mathbb{E}[|\bm c_j^\top\bbm\Theta_*\bm x_i\bm x_{i}^\top\bm c_k|^{1+\delta}] \leq \mathbb{E}[|\bm c_j^\top\bbm\Theta_*\bm x_i|^{2+2\delta}]^{1/2}\cdot\mathbb{E}[|\bm x_{i}^\top\bm c_k|^{2+2\delta}]^{1/2} \leq M_{x,2+2\delta}\cdot\|\bbm{\Theta}_*\|_{2,\infty}^{1+\delta},
\end{equation}
indicating that it also has a finite $(1+\delta)$-th moment. Therefore, combining these two moment bounds, each entry of $\bm y_i\bm x_i^\top$ has a finite $(1+\delta)$-th moment $M_{yx,1+\delta}=2M_{\text{eff},1+\delta,1}+2M_{x,2+2\delta}\|\bbm{\Theta}_*\|_{2,\infty}^{1+\delta}$.

Next, we analyze $\zeta_0$, the estimation accuracy of $\widehat{\bbm\Sigma}_x$. By the same argument as in Step 2 of Appendix \ref{append:B.1}, we have that when $\tau_x\asymp ( n M_{x,2+2\lambda,1}/\log d)^{1/(1+\lambda)}$,
$$
\left\|\widehat{\bm{\Sigma}}_x(\tau_x)-\bm{\Sigma}_x\right\|_\infty \lesssim  \left(\frac{M_{x,2+2\lambda,1}^{1/\lambda}\log d}{n}\right)^\frac{\lambda}{1+\lambda}=\zeta_0
$$
with probability at least $1-C\exp(-C\log d)$.

Then, we derive $\zeta_1$. Similarly, the $(j,k)$-th entry of $(\widehat{\bbm{\Sigma}}_{yx}(\tau_{yx})-\bbm{\Sigma}_{yx})$ can be upper bounded as
\begin{equation}\label{eq:Bilinear_zeta_1}
\begin{aligned}
  \left|(\widehat{\bbm{\Sigma}}_{yx}(\tau_{yx})-\bbm{\Sigma}_{yx})_{jk}\right|\leq &\left|\frac{1}{n}\sum_{i=1}^n\text{T}(y_{ij}x_{ik},\tau_{yx})-\bb E[\text{T}(y_{ij}x_{ik},\tau_{yx})]\right|\\
  & +\left|\bb E[\text{T}(y_{ij}x_{ik},\tau_{yx})]-\bb E[y_{ij}x_{ik}]\right|.
\end{aligned}
\end{equation}

When $\tau_{yx}\asymp \left(n M_{yx,1+\delta}/\log d\right)^{1/(1+\delta)}$, the second term can be bounded as
\begin{equation}
\begin{aligned}
  \left|\bb E[\text{T}(y_{ij}x_{ik},\tau_{yx})-y_{ij}x_{ik}]\right| \leq & \bb E[|y_{ij}x_{ik}|\cdot 1\{|y_{ij}x_{ik}|\geq \tau_{yx}\}]\\
  \leq &\bb E[|y_{ij}x_{ik}|^{1+\delta}]^\frac{1}{1+\delta}\cdot \bb P\left(|y_{ij}x_{ik}|\geq \tau_{yx}\right)^\frac{\delta}{1+\delta}\\
  \leq & E[|y_{ij}x_{ik}|^{1+\delta}]\cdot \tau_{yx}^{-\delta}\\
  \lesssim & \left(\frac{M_{yx,1+\delta}^{1/\delta}\log d}{n}\right)^{\frac{\delta}{1+\delta}}.
\end{aligned}
\end{equation}
For the first term of \eqref{eq:Bilinear_zeta_1}, by Bernstein's inequality and the upper bound of the variance that
$$
  \text{Var}(\text{T}(y_{ij}x_{ik},\tau_{yx}))\leq \bb E[(\text{T}(y_{ij}x_{ik},\tau_{yx}))^2]\leq \tau_{yx}^{1-\delta}M_{yx,1+\delta},
$$
we have that for any $t\geq 0$,
$$
\bb P\left(\left|\frac{1}{n}\sum_{i=1}^n\text{T}(y_{ij}x_{ik},\tau_{yx})-\bb E[\text{T}(y_{ij}x_{ik},\tau_{yx})]\right|\geq t\right)\leq 2\exp\left(-\frac{nt^2/2}{\tau_{yx}^{1-\delta}M_{yx,1+\delta}+\tau_{yx}t}\right).
$$
Letting $t=CM_{yx,1+\delta}\cdot\tau_{yx}^{-\delta}$, we obtain that with probability at least $1-C\exp(-C\log d)$,
$$
\left\|\frac{1}{n}\sum_{i=1}^n\text{T}(y_{ij}x_{ik},\tau_{yx})-\bb E[\text{T}(y_{ij}x_{ik},\tau_{yx})]\right\|_\infty\lesssim\left(\frac{M_{yx,1+\delta}^{1/\delta}\log d}{n}\right)^{\frac{\delta}{1+\delta}}.
$$
Combining the two pieces, we have that with probability at least $1-C\exp(-C\log d)$,
$$
\left\|\widehat{\bbm{\Sigma}}_{yx}(\tau_{yx})-\bbm{\Sigma}_{yx}\right\|_\infty \lesssim \left(\frac{M_{yx,1+\delta}^{1/\delta}\log d}{n}\right)^{\frac{\delta}{1+\delta}}=\zeta_1.
$$

Plugging $\zeta_0$ and $\zeta_1$ into \eqref{eq:Bilinear_init_bound}, we have
\begin{equation}
  \begin{split}
    & \|\widehat{\bbm{\Theta}}(R,\tau_x,\tau_{yx})-\bbm{\Theta}_*\|_\text{F} \\
    \lesssim & \|(\bm{\Sigma}_x^{-1})^\top\|_{1,\infty}\sqrt{s_*} \left(\left[\frac{M_{x,2+2\lambda,1}^{1/\lambda}\log d}{n}\right]^{\frac{\lambda}{1+\lambda}}\|\bbm{\Theta}_*\|_1 + \left[\frac{M_{yx,1+\delta}^{1/\delta}\log d}{n}\right]^{\frac{\delta}{1+\delta}}\right) \\
    \leq & \|(\bm{\Sigma}_x^{-1})^\top\|_{1,\infty}\sqrt{s_*}\left(M_{x,2+2\lambda,1}^{1/(1+\lambda)}\|\bbm{\Theta}_*\|_1+M_{yx,1+\delta}^{1/(1+\delta)}\right)\left(\frac{\log d}{n}\right)^{\frac{\delta}{1+\delta}}.
  \end{split}
\end{equation}

Hence, when
\begin{equation}
  n \gtrsim \|(\bm{\Sigma}_x^{-1})^\top\|_{1,\infty}^{\frac{1+\delta}{\delta}}\sigma_1^{-\frac{1+\delta}{2\delta}}\left(\beta_x/\alpha_x\right)^{\frac{1+\delta}{2\delta}}s_*^{\frac{1+\delta}{2\delta}}\left(M_{x,2+2\lambda,1}^{1/(1+\lambda)}\|\bbm{\Theta}_*\|_1+M_{yx,1+\delta}^{1/(1+\lambda)}\right)^{\frac{1+\delta}{\delta}}\log d,
\end{equation}
we have $\|\widehat{\bbm{\Theta}}(R,\tau_x,\tau_{yx})-\bbm{\Theta}_*\|_\mathrm{F} \lesssim (\alpha_x/\beta_x)^{1/2}\sigma_1$. Similarly, by Lemma \ref{lemma:3} and the Young-Eckart-Mirsky Theorem which implies the initial error bound holds.

\end{proof}

\subsection{Stability of Robust De-scaled Gradients}\label{append:D.2}

\begin{proposition}
  [Stability of Robust De-scaled Gradient Estimators]\label{prop:stability_bilinear}
  Under the conditions in Theorem \ref{thm:matrix_trace_reg}, with a probability at least $1-C\exp(-C\log d)$, the robust de-scaled gradient estimators satisfy
  \begin{equation}
    \begin{split}
      & \|\bm{G}_L(\bm{L},\bm{R};\tau) - \mathbb{E}[\nabla_{\bm{L}}f(\bm{L},\bm{R};z_i)](\bm{R}^\top\bm{R})^{-1/2}\|_\textup{F}^2 \lesssim \xi^2_L + \phi_{\lambda,\epsilon}\|\bm{L}\bm{R}^\top-\bm{L}_*\bm{R}_*^\top\|_\textup{F}^2,\\
      & \|\bm{G}_R(\bm{L},\bm{R};\tau) - \mathbb{E}[\nabla_{\bm{R}}f(\bm{L},\bm{R};z_i)](\bm{R}^\top\bm{R})^{-1/2}\|_\textup{F}^2 \lesssim \xi^2_L + \phi_{\lambda,\epsilon}\|\bm{L}\bm{R}^\top-\bm{L}_*\bm{R}_*^\top\|_\textup{F}^2,
    \end{split}
  \end{equation}
  where 
  $$
    \phi_{\lambda,\epsilon} = s\bar{\sigma}_1^{2\lambda}\overline{M}_{x,\textup{eff}}^2\left(\frac{\log d}{n}\right)^\frac{2\min(\epsilon,\lambda)}{1+\epsilon},
  $$ 
  $$
    \overline{M}_{x,\mathrm{eff}}^2=M_{x,2+2\lambda}^2M_{\mathrm{eff},1+\epsilon,s}^\frac{-2\lambda}{1+\epsilon}+M_{x,2+2\lambda}^\frac{2}{1+\lambda}M_{\mathrm{eff},1+\epsilon,s}^\frac{2\epsilon(\lambda-1)}{(1+\epsilon)(1+\lambda)}+M_{x,2+2\lambda}^\frac{2}{1+\lambda},
  $$
  $$
    \xi^2_L=s_1\left[\frac{M_{\textup{eff},1+\epsilon,s}^{1/\epsilon}\log d}{n}\right]^{\frac{2\epsilon}{1+\epsilon}},\quad\text{and}\quad
    \xi^2_R=s_2\left[\frac{M_{\textup{eff},1+\epsilon,s}^{1/\epsilon}\log d}{n}\right]^{\frac{2\epsilon}{1+\epsilon}}.
  $$
\end{proposition}~

\begin{proof}[Proof of Proposition \ref{prop:stability_bilinear}]

By definition,
\begin{equation}
  \begin{aligned}
    & \bm{G}_{L}(\bm{L},\bm{R};\tau) -\bb{E}\left[\nabla_{\bm L} f(\bm L,\bm R; z_i)\right](\bm R^\top\bm R)^{-1/2}\\ 
    = & \frac{1}{n}\sum_{i=1}^n\text{T}\Big(\left(\bm X_i\otimes (\bm A\bm X_i\bm B^\top-\bm Y_i)\right)\wt{\bm{R}};\tau\Big)-\mathbb{E}\left[\left(\bm X_i\otimes (\bm A\bm X_i\bm B^\top-\bm Y_i)\right)\wt{\bm{R}}\right]\\
    = & \Bigg\{\mathbb{E}\left[\text{T}\Big(\left(\bm X_i\otimes (\bm A_*\bm X_i\bm B_*^\top-\bm Y_i)\right)\wt{\bm{R}};\tau\Big)\right] - \mathbb{E}\left[\left(\bm X_i\otimes (\bm A_*\bm X_i\bm B_*^\top-\bm Y_i)\right)\wt{\bm{R}}\right] \Bigg\} \\
    + & \Bigg\{ \frac{1}{n}\sum_{i=1}^n\text{T}\Big(\left(\bm X_i\otimes (\bm A_*\bm X_i\bm B_*^\top-\bm Y_i)\right)\wt{\bm{R}};\tau\Big) - \mathbb{E}\left[\text{T}\Big(\left(\bm X_i\otimes (\bm A_*\bm X_i\bm B_*^\top-\bm Y_i)\right)\wt{\bm{R}};\tau\Big)\right] \Bigg\} \\
    + & \Bigg\{ \mathbb{E}\left[\left(\bm X_i\otimes (\bm A_*\bm X_i\bm B_*^\top-\bm Y_i)\right)\wt{\bm{R}}\right] - \mathbb{E}\left[\left(\bm X_i\otimes (\bm A\bm X_i\bm B^\top-\bm Y_i)\right)\wt{\bm{R}}\right] \\
    & + \mathbb{E}\left[\text{T}\Big(\left(\bm X_i\otimes (\bm A\bm X_i\bm B^\top-\bm Y_i)\right)\wt{\bm{R}};\tau\Big)\right] - \mathbb{E}\left[\text{T}\Big(\left(\bm X_i\otimes (\bm A_*\bm X_i\bm B_*^\top-\bm Y_i)\right)\wt{\bm{R}};\tau\Big)\right]\Bigg\} \\
    + & \Bigg\{\frac{1}{n}\sum_{i=1}^n\text{T}\Big(\left(\bm X_i\otimes (\bm A\bm X_i\bm B^\top-\bm Y_i)\right)\wt{\bm{R}};\tau\Big) - \frac{1}{n}\sum_{i=1}^n\text{T}\Big(\left(\bm X_i\otimes (\bm A_*\bm X_i\bm B_*^\top-\bm Y_i)\right)\wt{\bm{R}};\tau\Big) \\
    & - \mathbb{E}\left[\text{T}\Big(\left(\bm X_i\otimes (\bm A\bm X_i\bm B^\top-\bm Y_i)\right)\wt{\bm{R}};\tau\Big)\right] + \mathbb{E}\left[\text{T}\Big(\left(\bm X_i\otimes (\bm A_*\bm X_i\bm B_*^\top-\bm Y_i)\right)\wt{\bm{R}};\tau\Big)\right]\Bigg\} \\
    =: & \bm{M}_{L,1} + \bm{M}_{L,2} + \bm{M}_{L,3} + \bm{M}_{L,4},
  \end{aligned}
\end{equation}
where $\wt{\bm R}=\bm R(\bm R^\top\bm R)^{-1/2}=\bm R/\norm{\bm R}_2$ is a $s_R$ sparse vector with unit Euclidean norm. 

\noindent\textit{Step 1.} (Moment bounds)

\noindent Denote the $j$-th entry of $\left(\bm X_i\otimes (\bm A_*\bm X_i\bm B_*^\top-\bm Y_i)\right)\wt{\bm{R}}$ as
\begin{equation}\label{eq:moment3}
  v_{ij} = \bm{c}_j^\top\left(\bm X_i\otimes (\bm A_*\bm X_i\bm B_*^\top-\bm Y_i)\right)\wt{\bm{R}}=-\bm{c}_j^\top\bm J_i\widetilde{\bm{R}},
\end{equation}
satisfying $\mathbb{E}[|v_{ij}|^{1+\epsilon}]\leq M_{\text{eff},1+\epsilon,s}$, where $s=\max \left(\min \left(s_1, d_1\right), \min \left(s_2, d_2\right)\right)$.

In addition, the $j$-th entry of $\bm{M}_{L,3}$ is
\begin{equation}
  \begin{split}
    & \mathbb{E}\left[\bm c_j^\top\left(\bm X_i\otimes (\bm A_*\bm X_i\bm B_*^\top-\bm A\bm X_i\bm B^\top)\right)\wt{\bm{R}}\right]\\
    & + \mathbb{E}\left[\text{T}\left(\bm c_j^\top\left(\bm X_i\otimes (\bm A\bm X_i\bm B^\top-\bm Y_i)\right)\wt{\bm{R}}\right)\right] - \mathbb{E}\left[\text{T}\left(\bm c_j^\top\left(\bm X_i\otimes (\bm A_*\bm X_i\bm B_*^\top-\bm Y_i)\right)\wt{\bm{R}}\right)\right]\\
    =: & \mathbb{E}[q_{ij}] + \mathbb{E}[\text{T}(v_{ij}-q_{ij};\tau)] - \mathbb{E}[\text{T}(v_{ij};\tau)].
  \end{split}
\end{equation}
Let $\bbm\Delta_{\bm A}=\bm A_*-\bm A$ and $\bbm\Delta_{\bm B}=\bm B_*-\bm B$. Therefore, for $\bb E[|q_{ij}|^{1+\lambda}]$, we have
\begin{equation}\label{eq:bilinear_q1}
  \begin{aligned}
    \bb{E}[|q_{ij}|^{1+\lambda}]^{\frac{1}{1+\lambda}} & = \mathbb{E}\left[\left|\bm c_j^\top\left(\bm X_i\otimes (\bm A_*\bm X_i\bm B_*^\top-\bm A\bm X_i\bm B^\top)\right)\wt{\bm{R}}\right|^{1+\lambda}\right]^\frac{1}{1+\lambda}\\
    & \leq \mathbb{E}\left[\left|\bm c_j^\top\left(\bm X_i\otimes (\bbm\Delta_{\bm A}\bm X_i\bm B_*^\top)\right)\wt{\bm{R}}\right|^{1+\lambda}\right]^\frac{1}{1+\lambda}\\
    &\quad +\mathbb{E}\left[\left|\bm c_j^\top\left(\bm X_i\otimes (\bm A_*\bm X_i\bbm\Delta_{\bm B}^\top)\right)\wt{\bm{R}}\right|^{1+\lambda}\right]^\frac{1}{1+\lambda}.
  \end{aligned}
\end{equation}
For the first term of \eqref{eq:bilinear_q1}, by Theorem 16.2.2 of \cite{harvillematrix}, we have
\begin{equation}
  \begin{aligned}
      \bm c_j^\top\left(\bm X_i\otimes (\bbm\Delta_{\bm A}\bm X_i\bm B_*^\top)\right)\wt{\bm{R}}=&\operatorname{tr}\left(\text{mat}(\bm c_j)^\top\bbm\Delta_{\bm A}\bm X_i\bm B_*^\top\text{mat}(\wt{\bm R})\bm X^\top\right)\\
      =&\operatorname{tr}\left(\bm C_j^\top\bbm\Delta_{\bm A}\bm X_i\bm B_*^\top\bm B\bm X^\top\right)/\|\bm B\|_\text{F},
  \end{aligned}
\end{equation}
where $\text{mat}(\cdot)$ denotes the inverse of the $\text{vec}(\cdot)$ operator with appropriate dimensions. The term $\bm C_j = \text{mat}(\bm c_j)$ represents a sparse indicator matrix, where the entry corresponding to the index $j$ is 1 and all other entries are 0. Since $\bm C_j$ is rank-1, we decmopose it as $\bm C_j = \bm u_j \bm v_j^\top$, where each of $\bm u_j$ and $\bm v_j$ contains a single 1 and 0s elsewhere. In addition, denote the SVD of $\bm B_*\bm B^\top$ as $\bm B_*\bm B^\top=\sum_{k=1}^{r_B}\gamma_k \bm w_k\bm z_k^\top$, where $r_B$ is its rank, $\{\gamma_k\}$ are singular values and $\{\bm w_k\}$ and $\{\bm z_k\}$ are corresponding singular vectors. Then, for the first term of \eqref{eq:bilinear_q1}, we have
\begin{equation}
  \begin{aligned}
    &\mathbb{E}\left[\left|\bm c_j^\top\left(\bm X_i\otimes (\bbm\Delta_{\bm A}\bm X_i\bm B_*^\top)\right)\wt{\bm{R}}\right|^{1+\lambda}\right]^\frac{1}{1+\lambda}\\
    =& \|\bm B\|_\text{F}^{-1}\mathbb{E}\left[\left|\operatorname{tr}\left(\bm v_j\bm u_j^\top\Delta_{\bm A}\bm X_i\sum_{k=1}^{r_B}\gamma_k \bm w_k\bm z_k^\top\bm X^\top\right)\right|^{1+\lambda}\right]^\frac{1}{1+\lambda}\\
    =& \|\bm B\|_\text{F}^{-1}\mathbb{E}\left[\left|\sum_{k=1}^{r_B}\gamma_k\left(\bm u_j^\top\Delta_{\bm A}\bm X_i \bm w_k\right)\left(\bm z_k^\top\bm X^\top\bm v_j\right)\right|^{1+\lambda}\right]^\frac{1}{1+\lambda}\\
    \leq &\|\bm B\|_\text{F}^{-1}\sum_{k=1}^{r_B}\gamma_k\mathbb{E}\left[\left|\left(\bm u_j^\top\Delta_{\bm A}\bm X_i \bm w_k\right)\left(\bm z_k^\top\bm X^\top\bm v_j\right)\right|^{1+\lambda}\right]^\frac{1}{1+\lambda}\\
    \leq &  \|\bm B\|_\text{F}^{-1}\|\bm u_j^\top\bbm\Delta_{\bm A}\|_2 M_{x,2+2\lambda}^{1/(1+\lambda)}\sum_{k=1}^{r_B}\gamma_k\\
    \leq & \|\bm B\|_\text{F}^{-1}\|\bbm\Delta_{\bm A}\|_\text{op}\|\bm B_*\bm B^\top\|_\text{nuc}M_{x,2+2\lambda}^{1/(1+\lambda)},
  \end{aligned}
\end{equation}
where the first inequality is given by the Minkowski's inequality, the second is given by an analogous argument to \eqref{eq:moment2}, and $\|\cdot\|_\text{nuc}$ denotes the nuclear norm.

Next, we derive a upper bound for $\|\bm B_*\bm B^\top\|_\text{nuc}$. We have
\begin{equation}
  \begin{aligned}
    \|\bm B_*\bm B^\top\|_\text{nuc}\leq& \|\bm B\bm B^\top\|_\text{nuc} + \|\bbm\Delta_{\bm B}\bm B\|_\text{nuc}\\
    =&\operatorname{tr}\left(\bm B\bm B^\top\right)+\|\bbm\Delta_{\bm B}\bm B\|_\text{nuc}\\
    \leq & \|\bm B\|_\text{F}^2 + \|\bbm\Delta_{\bm B}\|_\text{F}\|\bm B\|_\text{F},
  \end{aligned}
\end{equation}
where the second inequality follows from the fact that $\|\bm M_1\bm M_2\|_\text{nuc} \leq \|\bm M_1\|_\text{F} \|\bm M_2\|_\text{F}$ for any conformable matrices $\bm M_1$ and $\bm M_2$. Since in the bilinear model, $\bbm\Sigma_*$ reduces to the scalar $\sigma_1=\|\bm A_*\otimes \bm B_*\|_\text{F}=\|\bm B_*\|_\text{F}^2=\|\bm A_*\|_\text{F}^2$, by the initialization condition \eqref{eq:condition1}, we obtain that $\|\bbm\Delta_{\bm B}\|_\text{F}\leq C\sqrt{\sigma_1}$. Moreover, $\|\bm B\|_\text{F}\leq \|\bbm\Delta_{\bm B}\|_\text{F}+\|\bm B_*\|_\text{F}\leq C\sqrt{\sigma_1}$. Therefore, we have that 
\begin{equation}
  \|\bm B_*\bm B^\top\|_\text{nuc}\leq C\sqrt{\sigma_1}\|\bm B\|_\text{F}.
\end{equation}
Finally, by Lemma \ref{lemma:2}, the first term of \eqref{eq:bilinear_q1} is upper bounded by
\begin{equation}
  \begin{aligned}
      \mathbb{E}\left[\left|\bm c_j^\top\left(\bm X_i\otimes (\bbm\Delta_{\bm A}\bm X_i\bm B_*^\top)\right)\wt{\bm{R}}\right|^{1+\lambda}\right]^\frac{1}{1+\lambda}&\leq C\sqrt{\sigma_1}\|\bbm\Delta_{\bm A}\|_\text{op}M_{x,2+2\lambda}^{1/(1+\lambda)}\\
      &\leq Cd(\bm F,\bm F_*)M_{x,2+2\lambda}^{1/(1+\lambda)}\\
      &\leq C\|\bbm\Theta-\bbm\Theta_*\|_\text{F}M_{x,2+2\lambda}^{1/(1+\lambda)}.
  \end{aligned}
\end{equation}
The upper bound of the second term of \eqref{eq:bilinear_q1} can be devived analogously. Therefore, we have 
\begin{equation}
  \bb{E}[|q_{ij}|^{1+\lambda}]\lesssim \|\bbm\Theta-\bbm\Theta_*\|_\text{F}^{1+\lambda}M_{x,2+2\lambda}.
\end{equation}

\noindent\textit{Step 2.} (Bound $\|(\bm{M}_{L,1})_{S_1}\|_\text{F}^2$)

\noindent For $\bm{M}_{L,1}$ and any $S_1$, we have
\begin{equation}
  \|(\bm{M}_{L,1})_{S_1}\|_\text{F}^2 = \sum_{j\in S_1}\Big|\mathbb{E}[\text{T}(v_{ij};\tau)] - \mathbb{E}[v_{ij}]\Big|^2.
\end{equation}
For any $j\in S_1$, similar to \eqref{eq:matreg_M1}, we have
\begin{equation}
  \begin{split}
    \Big|\mathbb{E}[\text{T}(v_{ij};\tau)] - \mathbb{E}[v_{ij}]\Big| \leq \mathbb{E}\Big[|v_{ij}|1\{|v_{ij}|\geq\tau\}\Big] 
    \leq M_{\text{eff},1+\epsilon,s}\cdot\tau^{-\epsilon} \asymp \left[\frac{M_{\text{eff},1+\epsilon,s}^{1/\epsilon}\log d}{n}\right]^\frac{\epsilon}{1+\epsilon},
  \end{split}
\end{equation}
with the truncation parameter $\tau\asymp(nM_{\text{eff},1+\epsilon,s}/\log d)^{1/(1+\epsilon)}$. Hence, for any $|S_1|\leq s_1$,
\begin{equation}
  \|(\bm{M}_{L,1})_{S_1}\|_\text{F}^2 \lesssim s_1\left[\frac{M_{\text{eff},1+\epsilon,s}^{1/\epsilon}\log d}{n}\right]^\frac{\epsilon}{1+\epsilon}.
\end{equation}~

\noindent\textit{Step 3.} (Bound $\|(\bm{M}_{L,2})_{S_1}\|_\text{F}^2$)

\noindent By definition, for any $S_1$,
\begin{equation}
  \|(\bm{M}_{L,2})_{S_1}\|_\text{F}^2 = \sum_{j\in S_1}\left|\frac{1}{n}\sum_{i=1}^n\text{T}(v_{ij};\tau) - \mathbb{E}[\text{T}(v_{ij};\tau)]\right|^2.
\end{equation}
For each $i=1,\dots,n$, $|\text{T}(v_{ij},\tau)|\leq \tau$, and its variance is upper bounded by
\begin{equation}
  \text{var}(\text{T}(v_{ij};\tau)) \leq \mathbb{E}\left[\text{T}(v_{ij};\tau)^2\right]\tau^{(1-\epsilon)}\cdot\mathbb{E}\left[|v_{ij}|^{1+\epsilon}\right] \leq \tau^{1-\epsilon}M_{\text{eff},1+\epsilon,s}.
\end{equation}
By Bernstein's inequality, for any $j\in S_1$, and $0<t\lesssim \tau^{-\epsilon}M_{\text{eff},1+\epsilon,s}$, 
\begin{equation}
  \mathbb{P}\left(\left|\frac{1}{n}\sum_{i=1}^n\text{T}(v_{ij};\tau) - \mathbb{E}[\text{T}(v_{ij};\tau)]\right|\geq t\right)\leq 2\exp\left(-\frac{nt^2}{4\tau^{1-\epsilon}M_{\text{eff}}}\right).
\end{equation}
Letting $t=CM_{\text{eff},1+\epsilon,s}^{1/(1+\epsilon)}\log(d)^{\epsilon/(1+\epsilon)}n^{-\epsilon/(1+\epsilon)}$, similar to \eqref{eq:matreg_M2_bern}, we have for $|S_1|\leq s_1$,
\begin{equation}
  \mathbb{P}\left(\max_{\substack{j\in S_1}}\left|\frac{1}{n}\sum_{i=1}^n\text{T}(v_{ij};\tau) - \mathbb{E}[\text{T}(v_{ij};\tau)]\right|\gtrsim \left[\frac{M_{\text{eff},1+\epsilon,s}^{1/\epsilon}\log d}{n}\right]^{\frac{\epsilon}{1+\epsilon}} \right)\leq C\exp\left(-C\log d\right).
\end{equation}
Hence, with probability at least $1-C\exp(-C\log d)$,
\begin{equation}
  \|(\bm{M}_{L,2})_{S_1}\|_\text{F}^2 \lesssim s_1\left[\frac{M_{\text{eff},1+\epsilon,s}^{1/\epsilon}\log d}{n}\right]^\frac{2\epsilon}{1+\epsilon}.
\end{equation}~

\noindent\textit{Step 4.} (Bound $\|(\bm{M}_{L,3})_{S_1}\|_\text{F}^2$)

\noindent By definition, for any $S_1$,
\begin{equation}
  \|(\bm{M}_{L,3})_{S_1}\|_\text{F}^2 = \sum_{j\in S_1}\Big|\mathbb{E}[q_{ij}] - \mathbb{E}\Big[\text{T}(q_{ij}+v_{ij};\tau) - \text{T}(v_{ij};\tau)\Big]\Big|^2.
\end{equation}
Similar to \eqref{eq:matreg_M3}, we have
\begin{equation}
  \begin{aligned}
    & \Big|\mathbb{E}[q_{ij}] - \mathbb{E}\Big[\text{T}(v_{ij}-q_{ij};\tau) - \text{T}(v_{ij};\tau)\Big]\Big| \\
    \leq & \Big|\mathbb{E}\Big[q_{ij}\cdot1\{|q_{ij}\geq\tau/2|\}\Big]\Big| + \Big|\mathbb{E}\Big[q_{ij}\cdot1\{|v_{ij}\geq\tau/2|\}\Big]\Big|\\
    \leq &~ \mathbb{E}\left[|q_{ij}|^{1+\lambda}\right]^{\frac{1}{1+\lambda}}\cdot\left(\frac{\mathbb{E}\Big[|q_{ij}|^{1+\lambda}\Big]}{\tau^{1+\lambda}}\right)^{\frac{\lambda}{1+\lambda}} + \mathbb{E}\left[|q_{ij}|^{1+\lambda}\right]^{\frac{1}{1+\lambda}}\cdot\left(\frac{\mathbb{E}\Big[|v_{ij}|^{1+\epsilon}\Big]}{\tau^{1+\epsilon}}\right)^{\frac{\lambda}{1+\lambda}}\\
    = &~ \mathbb{E}\left[|q_{ij}|^{1+\lambda}\right]\cdot\tau^{-\lambda}~+~\mathbb{E}\left[|q_{ij}|^{1+\lambda}\right]^{\frac{1}{1+\lambda}}\cdot\mathbb{E}\left[|v_{ij}|^{1+\epsilon}\right]^{\frac{\lambda}{1+\lambda}}\cdot\tau^{-\frac{\lambda(1+\epsilon)}{1+\lambda}}\\
    \leq &~ M_{x,2+2\lambda}\cdot\left(\frac{\log d}{nM_{\text{eff},1+\epsilon,s}}\right)^{\frac{\lambda}{1+\epsilon}}\cdot\|\bbm\Theta-\bbm\Theta_*\|_\mathrm{F}^{1+\lambda}\\
    & + M_{x,2+2\lambda}^{1/(1+\lambda)}\cdot M_{\text{eff},1+\epsilon,s}^{\lambda/(1+\lambda)}\cdot\left(\frac{\log d}{nM_{\text{eff},1+\epsilon,s}}\right)^{\frac{\lambda}{1+\lambda}}\cdot\|\bbm\Theta-\bbm\Theta_*\|_\mathrm{F}\\
    \leq & ~ \bar{\sigma}_1^{\lambda}\left(M_{x,2+2\lambda}M_{\text{eff},1+\epsilon,s}^{-\lambda/(1+\epsilon)}+M_{x,2+2\lambda}^{1/(1+\lambda)}\right)\left(\frac{\log d}{n}\right)^{\min\left(\frac{\lambda}{1+\lambda},\frac{\lambda}{1+\epsilon}\right)}\|\bbm\Theta-\bbm\Theta_*\|_\mathrm{F},
  \end{aligned}
\end{equation}
where $\bar{\sigma}_1=\max (\sigma_1(\bm L_*\bm R_*^\top),1)$. Therefore, we have
\begin{equation}
  \begin{split}
    & \|(\bm{M}_{L,3})_{S_1}\|_\text{F}^2 \\
    & \lesssim \bar{\sigma}_1^{2\lambda}\left(M_{x,2+2\lambda}^2M_{\text{eff},1+\epsilon,s}^{-2\lambda/(1+\epsilon)}+M_{x,2+2\lambda}^{2/(1+\lambda)}\right)\left(\frac{\log d}{n}\right)^{\min\left(\frac{2\lambda}{1+\lambda},\frac{2\lambda}{1+\epsilon}\right)}\|\bbm\Theta-\bbm\Theta_*\|_\mathrm{F}^2.
  \end{split}
\end{equation}~

\noindent\textit{Step 5.} (Bound $\|(\bm{M}_{L,4})_{S_1}\|_\text{F}^2$)

\noindent By definition, for any $S_1$,
\begin{equation}
  \|(\bm{M}_{L,4})_{S_1}\|_\text{F}^2 = \sum_{j\in S_1}\Bigg|\frac{1}{n}\sum_{i=1}^n\Big[\text{T}(q_{ij}+v_{ij};\tau)-\text{T}(v_{ij};\tau)\Big] -\mathbb{E}\Big[\text{T}(q_{ij}+v_{ij};\tau)-\text{T}(v_{ij};\tau)\Big] \Bigg|^2.
\end{equation}
For each $i=1,\dots,n$, we have $|\text{T}(q_{ij}+v_{ij};\tau)-\text{T}(v_{ij};\tau)|\leq 2\tau$, and hence,
\begin{equation}
\begin{aligned}
  \mathbb{E}\left[\Big(\text{T}(q_{ij}+v_{ij};\tau)-\text{T}(v_{ij};\tau)\Big)^2\right] &\leq (2\tau)^{1-\lambda}\cdot\mathbb{E}\left[|q_{ij}|^{1+\lambda}\right]\\
  &\lesssim \tau^{1-\lambda}M_{x,2+2\lambda}\norm{\bbm\Theta-\bbm\Theta_*}_\mathrm{F}^{1+\lambda}\\
  &=:V^2.
\end{aligned}
\end{equation}
By Bernstein's inequality, for any $j\in S_1$,
\begin{equation}\label{eq:bilinear_M4_bern}
  \begin{split}
    & \mathbb{P}\left(\left|\frac{1}{n}\sum_{i=1}^n\Big[\text{T}(q_{ij}+v_{ij};\tau)-\text{T}(v_{ij};\tau)\Big]-\mathbb{E}\Big[\text{T}(q_{ij}+v_{ij};\tau)-\text{T}(v_{ij};\tau)\Big]\right|\geq t\right)\\
    & \leq 2\exp\left(-\frac{Cnt^2}{V^2+\tau t}\right).
  \end{split}
\end{equation}
Let $t=C_1V\sqrt{\log d/n}+C_2\tau\log d/n$, where $C_1$ and $C_2$ are two positive constants. By a similar argument to that in Step 5 of Appendix \ref{append:B.2}, if $V^2\lesssim \tau t$, the RHS of \eqref{eq:bilinear_M4_bern} is upper bounded by
$$\begin{aligned}
  \text{RHS} &\leq 2\exp\left(-\frac{Cnt^2}{\tau t}\right)\\
  &=2\exp\left(-\frac{Cn(C_1V\sqrt{\log d/n}+C_2\tau\log d/n)}{\tau}\right)\\
  &\leq 2\exp\left(-\frac{CnC_2\tau\log d/n}{\tau}\right)\\
  &=2\exp\left(-C\log d\right).
\end{aligned}
$$
Conversely, if $V^2\gtrsim \tau t$, the RHS of \eqref{eq:bilinear_M4_bern} is upper bounded by
$$\begin{aligned}
  \text{RHS} &\leq 2\exp\left(-\frac{Cnt^2}{V^2}\right)\\
  &=2\exp\left(-\frac{Cn(C_1V\sqrt{\log d/n}+C_2\tau\log d/n)^2}{V^2}\right)\\
  &\leq 2\exp\left(-\frac{2CnC_1(V\sqrt{\log d/n})^2}{V^2}\right)\\
  &=2\exp\left(-C\log d\right).
\end{aligned}
$$
Then, with $K=1$, we have
\begin{equation}
  \begin{split} 
    \mathbb{P}\Bigg(&\max_{j\in S_1,k\in[K]}\left|\frac{1}{n}\sum_{i=1}^n\Big[\text{T}(q_{ijk}+v_{ijk};\tau)-\text{T}(v_{ijk};\tau)\Big]-\mathbb{E}\Big[\text{T}(q_{ijk}+v_{ijk};\tau)-\text{T}(v_{ijk};\tau)\Big]\right|\geq t\Bigg) \\ 
    & \lesssim d_1\exp(-C\log d) \leq C\exp(-C\log d).
  \end{split}
\end{equation}
For the upper bound $t$, plugging in the values of $V^2$ and $\tau$, we have
$$
\begin{aligned}
  t&=C_1\sqrt{\frac{V^2\log d}{n}}+C_2\tau\frac{\log d}{n}\\
  & \lesssim M_{x,2+2\lambda}^{1/2}M_{\text{eff},1+\epsilon,s}^{(1-\lambda)/(2+2\epsilon)}\left(\frac{\log d}{n}\right)^\frac{\min(\epsilon,\lambda)}{1+\epsilon}\|\bm L\bm R^\top-\bm L_*\bm R_*^\top\|_\mathrm{F}^\frac{1+\lambda}{2}+M_{\text{eff},1+\epsilon,s}\left(\frac{\log d}{n}\right)^\frac{\epsilon}{1+\epsilon}.
\end{aligned}
$$
Analogous to the derivation of \eqref{eq:matreg_m4_final}, it follows that
\begin{equation}
  \begin{split}
      \|(\bm{M}_{L,4})_{S_1}\|_\text{F}^2&\lesssim s_1M_{x,2+2\lambda}^\frac{2}{1+\lambda}M_{\text{eff},1+\epsilon,s}^\frac{2\epsilon(\lambda-1)}{(1+\epsilon)(1+\lambda)}\left(\frac{\log d}{n}\right)^\frac{2\lambda}{1+\lambda}\|\bm L\bm R^\top-\bm L_*\bm R_*^\top\|_\mathrm{F}^2\\
      &+s_1M_{\text{eff},1+\epsilon,s}^2\left(\frac{\log d}{n}\right)^\frac{2\epsilon}{1+\epsilon}.
  \end{split}
\end{equation}

Combining the upper bounds of this step and Step 4, we have
\begin{equation}
\begin{aligned}
    \|(\bm{M}_{L,3})_{S_1}\|_\text{F}^2+\|(\bm{M}_{L,4})_{S_1}\|_\text{F}^2
    \lesssim \phi_{\lambda,\epsilon}\|\bm{L}\bm{R}^\top-\bm{L}_*\bm{R}_*^\top\|_\text{F}^2+ s_1\left[\frac{M_{\text{eff},1+\epsilon,s}^{1/\epsilon}\log d}{n}\right]^{\frac{2\epsilon}{1+\epsilon}},
\end{aligned}
\end{equation}
where 
$$
\phi_{\lambda,\epsilon}=s\bar{\sigma}_1^{2\lambda}\overline{M}_{x,\text{eff}}^2\left(\frac{\log d}{n}\right)^\frac{2\min(\epsilon,\lambda)}{1+\epsilon},
$$
and 
$$
  \overline{M}_{x,\mathrm{eff}}^2=M_{x,2+2\lambda}^2M_{\mathrm{eff},1+\epsilon,s}^\frac{-2\lambda}{1+\epsilon}+M_{x,2+2\lambda}^\frac{2}{1+\lambda}M_{\mathrm{eff},1+\epsilon,s}^\frac{2\epsilon(\lambda-1)}{(1+\epsilon)(1+\lambda)}+M_{x,2+2\lambda}^\frac{2}{1+\lambda}.
$$

In summary, based on the upper bounds in steps 2 to 5, we have
\begin{equation}
  \sum_{j=1}^4\|(\bm{M}_{L,j})_{S_1}\|_\text{F}^2 \lesssim s_1\left[\frac{M_{\text{eff},1+\epsilon,s}^{1/\epsilon}\log d}{n}\right]^{\frac{2\epsilon}{1+\epsilon}} + \phi_{\lambda,\epsilon}\|\bm{L}\bm{R}^\top-\bm{L}_*\bm{R}_*^\top\|_\text{F}^2.
\end{equation}

\end{proof}

\subsection{Convergence Rates of Bilinear Model}

\begin{proof}[Proof of Theorem \ref{thm:bilinear_model}]

Similarly to the proof of matrix trace regression and GLMs, the initialization guarantees and robust gradient stabilities are verified in Appendices \ref{append:D.1} and \ref{append:D.2}.

Finally, we examine the conditions and apply Theorem \ref{thm:2}. Under Assumption \ref{asmp:moment2}, by Lemma 3.11 in \citet{bubeck2015convex}, we can show that the RCG condition holds with $\alpha=\alpha_x$ and $\beta=\beta_x$.

The de-scaled robust gradient estimators are stable: i.e., with probability at least $1-C\exp(-\log d)$,
\begin{equation}
  \|\bm{G}_L(\bm{L},\bm{R};\tau)_{S_1}\|_\text{F}^2\lesssim s_1\left[\frac{M_{\text{eff},1+\epsilon,s}^{1/\epsilon}\log d}{n}\right]^{\frac{2\epsilon}{1+\epsilon}} + \phi_{\lambda,\epsilon}\|\bm{L}\bm{R}^\top-\bm{L}_*\bm{R}_*^\top\|_\text{F}^2.
\end{equation}
Similarly, the similar bound holds for $\bm{G}_R(\bm{L},\bm{R};\tau)_{S_2}$,
\begin{equation}
  \|\bm{G}_R(\bm{L},\bm{R};\tau)_{S_2}\|_\text{F}^2\lesssim s_2\left[\frac{M_{\text{eff},1+\epsilon,s}^{1/\epsilon}\log d}{n}\right]^{\frac{2\epsilon}{1+\epsilon}} + \phi_{\lambda,\epsilon}\|\bm{L}\bm{R}^\top-\bm{L}_*\bm{R}_*^\top\|_\text{F}^2.
\end{equation}
As the sample size satisfies that
\begin{equation}
  \begin{aligned}
    &n\gtrsim \left(s\alpha_x^{-2}\bar{\sigma}_1^{2\lambda}\overline{M}_{x,\text{eff}}^2\right)^\frac{1+\epsilon}{2\min(\epsilon,\lambda)}\log d,\\
   \text{and}\quad & n\gtrsim \left(s\sigma_1^{-2}\alpha_x^{-3}\beta_x\right)^\frac{1+\epsilon}{2\epsilon}M_{\text{eff},1+\epsilon,s}^{1/\epsilon}\log d,
  \end{aligned}
\end{equation}
we have $\phi_{\lambda,\epsilon}\lesssim\alpha_x^2$ and $\xi_L^2+\xi_R^2\lesssim \alpha^3_x\beta^{-1}_x\sigma_1^2$. By Theorem \ref{thm:2},
\begin{equation}
  \|\bm{\Theta}^{(j)}-\bm{\Theta}_*\|_\text{F}^2 \lesssim (1-C\alpha_x\beta_x^{-1})^t \|\bm{\Theta}^{(0)}-\bm{\Theta}_*\|_\text{F}^2 + C\alpha_x^{-2}s\left[\frac{M_{\text{eff},1+\epsilon,s}^{1/\epsilon}\log d}{n}\right]^{\frac{2\epsilon}{1+\epsilon}}.
\end{equation}
After sufficient iterations, we have
\begin{equation}
  \|\wh{\bm{\Theta}}-\bm{\Theta}_*\|_\text{F}^2 \lesssim C\alpha_x^{-2}s\left[\frac{M_{\text{eff},1+\epsilon,s}^{1/\epsilon}\log d}{n}\right]^{\frac{2\epsilon}{1+\epsilon}}.
\end{equation}

\end{proof}

\section{Additional Simulation Results}\label{append:E}

To complement the findings in Section \ref{sec:5.3}, we report the performance of the proposed two-stage estimator for matrix logistic regression and bilinear models in the presence of heavy-tailed distributions.

\subsection{Matrix Logistic Regression}

We compare our method with unrobustified ScGD-SHT to highlight the benefit of truncation. The simulation setup is configured as follows: $p_1=q_1=p_2=q_2=6$, $K=1$, $s_{L,*}=s_{R,*}=5$, $s_L=s_R=7$, $n=0.5p_1q_1p_2q_2$, and $\sigma_1=2$. For the true values, we first generate $\bm L_1\in\mathbb{R}^{36}$ and $\bm R_1\in\mathbb{R}^{36}$, where the first 5 entries are drawn i.i.d. from $N(0,1)$ and the remaining entries are set to zero. Then, we let $\bm L_*=\sqrt{\sigma_1}\cdot \bm L_1/\|\bm L_1\|_2$ and $\bm R_*=\sqrt{\sigma_1}\cdot \bm R_1/\|\bm R_1\|_2$. For $\bm X_i$, we introduce a heterogeneous distribution. The entries of $\bm X_i$ corresponding to the nonzero entries of $\bbm\Theta_*$, which effectively contribute to the response $\bm Y_i$, are sampled from $N(0, 1)$. In contrast, entries in the inactive positions are drawn from a truncated $t_{1.1}$ distribution, denoted by $\bar{t}_{1.1}$ and $\bar{t}_{1.6}$, where values are clipped to $[-1000,1000]$ to ensure the existence of the second moment. This configuration mimics a scenario where a clean signal is embedded in heavy-tailed ambient noise.

To ensure a rigorous comparison reflecting the practical utility of each estimator, we impose a fixed computational budget across all methods. Specifically, in each replication, after the robust LASSO initialization, each algorithm is allowed to search for its optimal configuration within a grid of step sizes $\eta\in\{0.1,0.01,0.001\}$ and corresponding maximum iterations $J\in \{500,1000,2000\}$. The truncation parameter $\tau$ is selected by 5-fold cross-validation. We report the minimum estimation error achieved by each method over the three configurations.

\begin{figure}[htbp]
  \centering
  \includegraphics[width=0.6\textwidth]{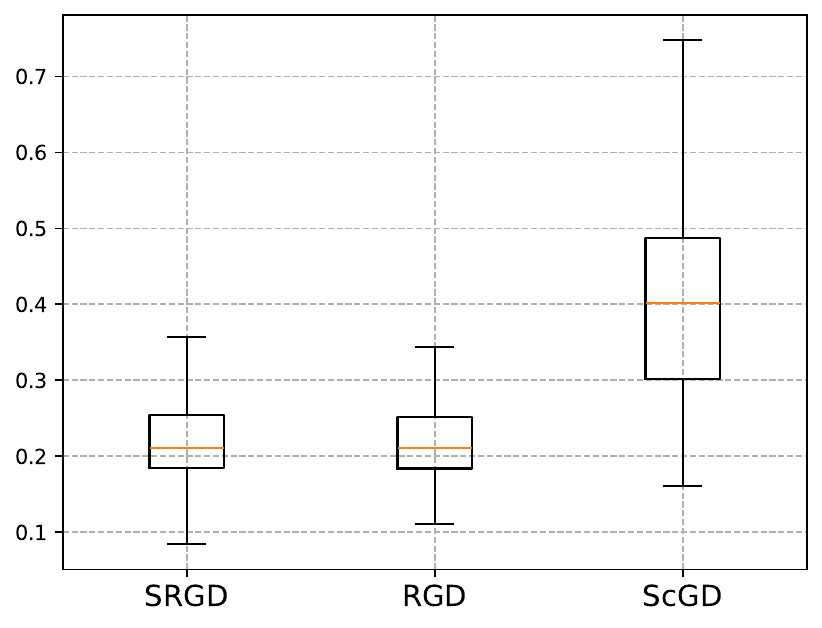}
  \caption{Boxplots of relative estimation errors on matrix logistic regression (200 repetitions).}
  \label{fig:exp3_logistic}
\end{figure}

Figure \ref{fig:exp3_logistic} illustrates the estimation errors of the three methods. The non-robust ScGD-SHT exhibits significantly higher estimation errors compared to the truncated estimators. In contrast, both SRGD--SHT and RGD-SHT maintain low estimation errors. This performance gap suggests that non-robust gradient updates are sensitive to the high-magnitude predictors arising from heavy-tailed distributions. The superior performance of SRGD--SHT and RGD-SHT demonstrates that element-wise truncation is critical for mitigating the influence of extreme leverage points.

\subsection{Bilinear Models}

For the bilinear models, the parameters are set as: $p_1=p_2=5$, $q_1=q_2=10$, $n=100$, $s_{L,*}=s_{R,*}=15$, $s_L=s_R=20$ and $\sigma_1=5$. For $\bm A_*$ and $\bm B_*$, their first 3 columns are drawn independently from $N(0,1)$, subsequently normalized to have unit Euclidean norm, and scaled by a factor of $\sqrt{\sigma_1}$ to control the signal strength. The remaining columns are zero-padded. For $\bm X_i$, analogous to the matrix logistic regression setup, we use the heterogeneous distribution. The entries of $\bm X_i$ at the intersection of the nonzero columns of $\bm A_*$ and $\bm B_*$, which effectively contribute to the response $\bm Y_i$, are sampled from $N(0, 1)$. In contrast, entries in the inactive positions are drawn from a truncated $t_{1.1}$ distribution, denoted by $\bar{t}_{1.1}$, where values are clipped to $[-100,100]$. For the random errors, we test random matrices with independent entries drawn from (a) $N(0,1)$, (b) $0.2t_{2.1}$, (c) $0.2t_{1.5}$, and (d) $0.2t_{1.1}$ distributions. The tuning parameters $\tau$ for all methods are selected via 5-fold cross-validation.

\begin{figure}[htbp]
  \centering
  \includegraphics[width=0.8\textwidth]{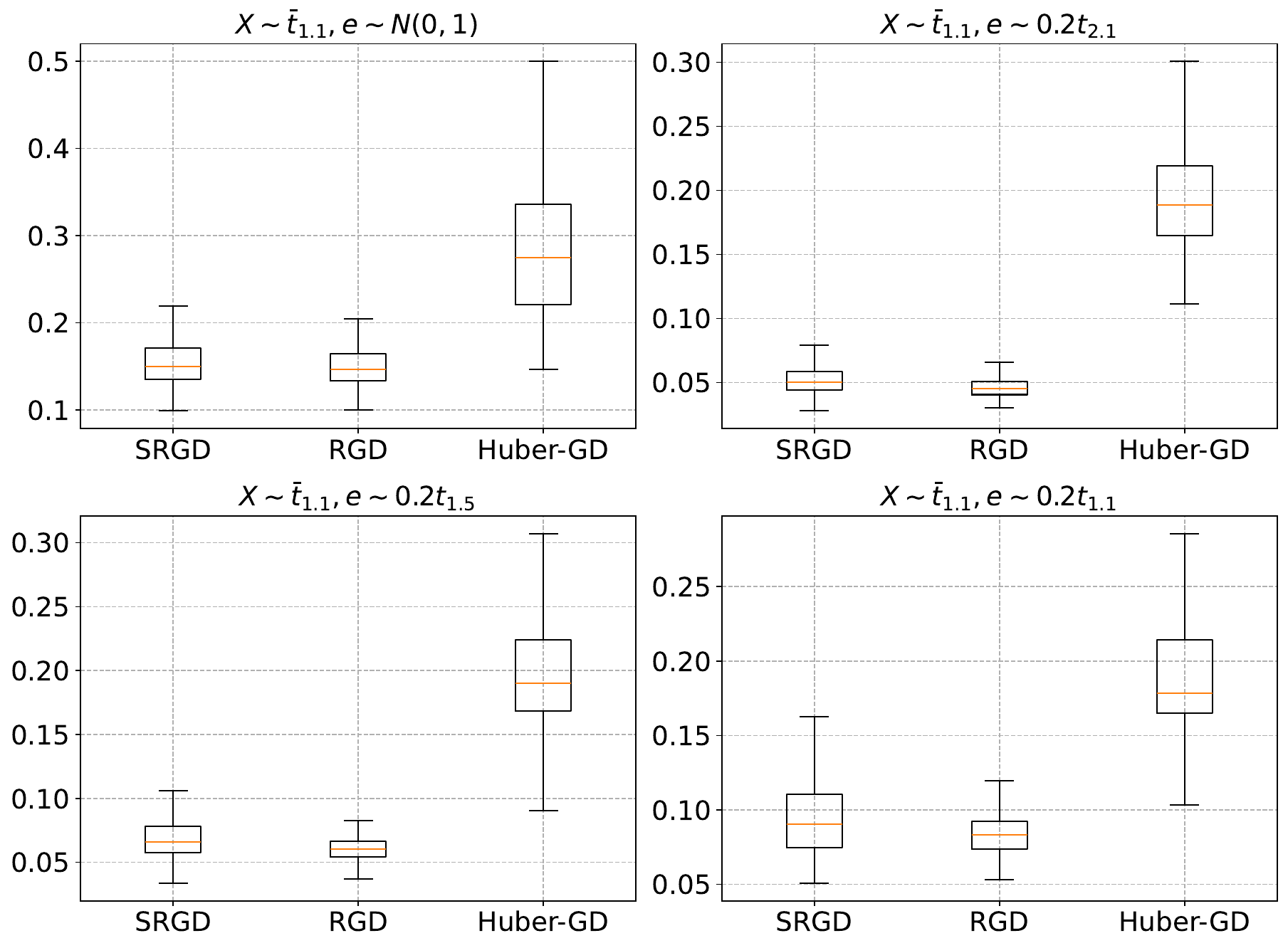}
  \caption{Boxplots of relative estimation errors on bilinear models (200 repetitions).}
  \label{fig:exp3_bilinear}
\end{figure}

Figure \ref{fig:exp3_bilinear} displays the estimation errors of the three robust methods. Huber-GD consistently exhibits significantly higher estimation errors compared to the truncation-based methods across all noise settings. It empirically confirms that the Huber loss is sensitive to high-leverage points induced by heavy-tailed predictors. Conversely, both SRGD--SHT and RGD--SHT achieve lower estimation errors, indicating that element-wise gradient truncation effectively reduces the influence of heavy-tailed predictors and heavy-tailed random errors. In addition, SRGD--SHT and RGD--SHT perform comparably in this experiment. This behavior is expected since the bilinear models are well-conditioned with $\kappa=1$.

\section{Additional Real Data Analysis}\label{append:F}
\setcounter{table}{0} 
\renewcommand{\thetable}{F.\arabic{table}}

This appendix provides details of the data set and additional results of the real data analysis.

\subsection{The EEG data}
To effectively capture the spatial correlation structure of the EEG signals via the Kronecker product structures, we reordered the 64 electrode channels based on their physical scalp topology. We first organized the electrodes into eight functional clusters, following a general anterior-to-posterior progression (from prefrontal to occipital). Within each functional cluster, the electrodes are arranged from the left hemisphere to the right hemisphere. The detailed ordering used in our analysis is presented in Table \ref{tab:eeg_ordering}. Row indices of $\bm{X_i}$ correspond to the electrodes listed in the table, following standard reading order (top to bottom, left to right).

\begin{table}[h]
\centering
\caption{Ordering of the 64 EEG electrodes in the EEG data analysis.}
\label{tab:eeg_ordering}
\begin{tabular}{l|l|l}
\hline
\textbf{Group} & \textbf{Region} & \textbf{Electrodes (Left $\rightarrow$ Right)} \\ \hline
1 & Pre-Frontal & FP1, AF7, AF1, FPZ, FP2, AF8, AF2, AFZ \\ \hline
2 & Frontal & F7, F5, F3, F1, FZ, F2, F4, F6 \\ \hline
3 & Fronto-Central & FT7, FC5, FC3, FC1, FCZ, FC2, FC4, FC6 \\ \hline
4 & Central & T7, C5, C3, C1, CZ, C2, C4, C6 \\ \hline
5 & Centro-Parietal & TP7, CP5, CP3, CP1, CPZ, CP2, CP4, CP6 \\ \hline
6 & Parietal & P7, P5, P3, P1, PZ, P2, P4, P6 \\ \hline
7 & Occipital & PO7, PO1, O1, OZ, POZ, PO2, PO8, O2 \\ \hline
8 & Lateral \& Misc & F8, FT8, T8, TP8, P8, X, Y, nd \\ \hline
\end{tabular}
\end{table}

Using the sparse Kronecker product structured robust estimates provided by the SRGD-SHT algorithm, we identify functional regions and time spans that are critical for classifying alcoholism. Figure \ref{fig:EEG_A1A2} displays the estimated $\bm A_1$ and $\bm A_2$, where the x-axis represents temporal segments of 62.5 ms each, and the y-axis corresponds to the anatomical subdivisions listed in order in Table \ref{tab:eeg_ordering}. Visual inspection reveals that the discriminative signals are not randomly distributed but are highly localized within specific spatio-temporal blocks. Specifically, non-zero coefficients are concentrated in two distinct spatial clusters: the anterior regions (indices 0--3), corresponding to the Pre-Frontal and Frontal cortices, and the posterior regions (indices 11--15), covering the Parietal, Occipital, and Lateral areas. Temporally, these regions are primarily activated during the interval of indices 3 to 7, translating to a latency window of 187.5 ms to 500 ms post-stimulus. In addition, a secondary activation is observed in the late latency window (812.5 ms to 1000 ms). These results empirically validate that the proposed method effectively filters out irrelevant noise and captures biologically meaningful patterns.

\begin{figure}[htbp]
  \centering
  \includegraphics[width=\textwidth]{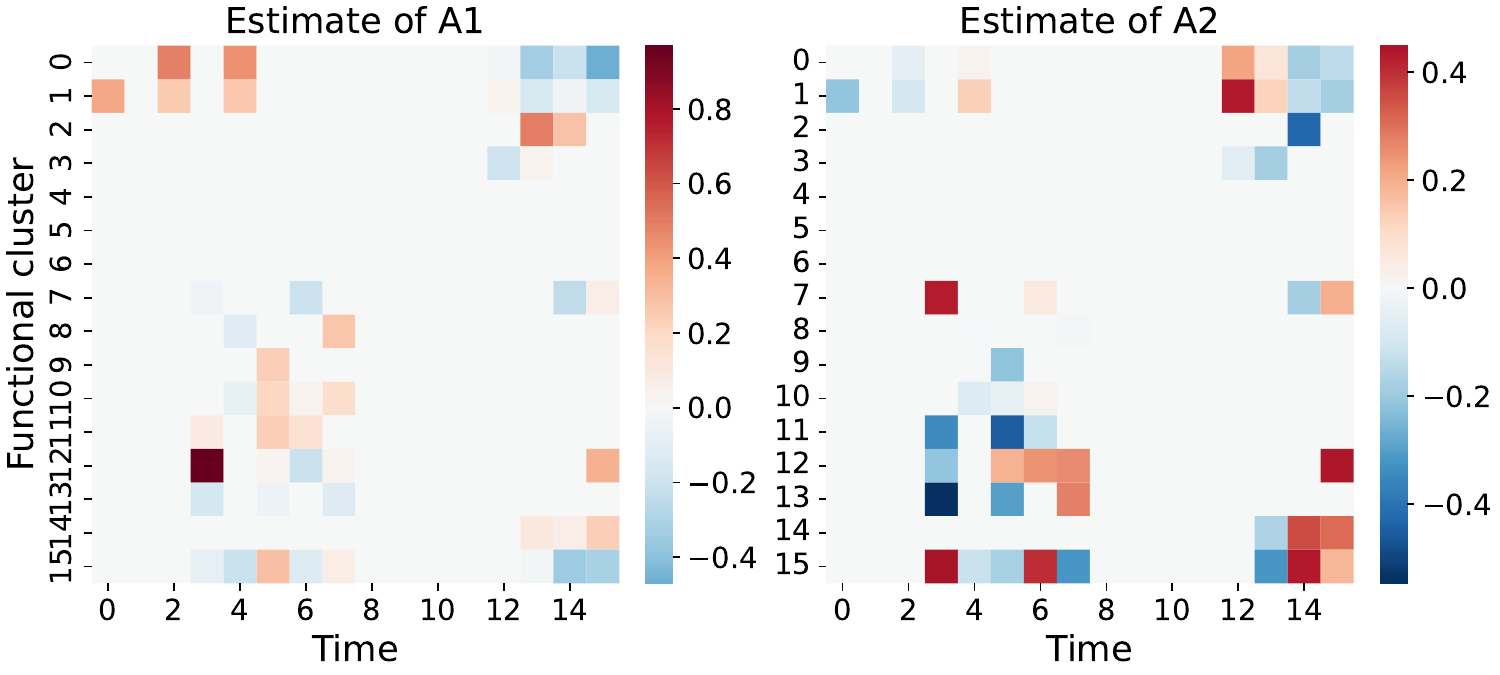}
  \caption{Estimates of $\bm A_1$ and $\bm A_2$.}
  \label{fig:EEG_A1A2}
\end{figure}

\subsection{The Macroeconomic Data}

We provide datails of the macroeconomic data set analyzed in Section \ref{sec:6.2}. In the data set, the 14 countries are the United States (USA), Canada (CAN), New Zealand (NZL), Australia (AUS), Norway (NOR), Ireland (IRL), Denmark (DNK), the United Kingdom (GBR), Finland (FIN), Sweden (SWE), France (FRA), the Netherlands (NLD), Austria (AUT), and Germany (DEU). The 10 economic indicators include the Consumer Price Index (CPI) for food (CPIF), energy (CPIE), and total items (CPIT); long-term interest rates (government bond yields, IRLT) and 3-month interbank rates (IR3); total industrial production excluding construction (PTEC) and total manufacturing production (PTM); Gross Domestic Product (GDP); and the total value of exports (ITEX) and imports (ITEM). The plot of sample kurtosis is provided in Figure \ref{fig:TS_kurt}, showing significant deviation from normality.

\begin{figure}[htbp]
  \centering
  \includegraphics[width=0.6\textwidth]{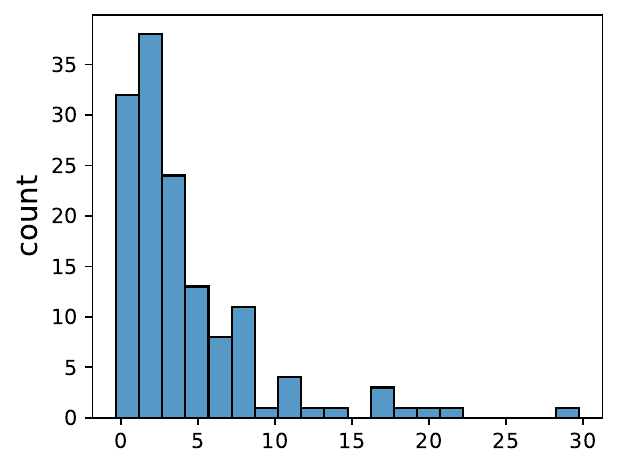}
  \caption{Sample excess kurtosis of the macroeconomic data.}
  \label{fig:TS_kurt}
\end{figure}

\end{document}